\documentclass[11pt]{article}
\usepackage[a4paper,left=20mm,top=20mm,right=20mm,bottom=25mm]{geometry}

\usepackage{graphicx}
\usepackage{mathptmx} 
\usepackage[mathscr]{euscript}
\usepackage{authblk}
\usepackage{natbib}
\usepackage{rotating}
\usepackage{mathtools}
\usepackage{amsfonts}
\usepackage{amssymb,amsthm,amsmath}
\usepackage{color}
\usepackage[noend]{algorithmic}
\usepackage{algorithm}
\usepackage{mathtools}
\usepackage{chngcntr}
\usepackage{apptools}
\AtAppendix{\counterwithin{lemma}{section}}
\AtAppendix{\counterwithin{theorem}{section}}
\AtAppendix{\counterwithin{corollary}{section}}
\AtAppendix{\counterwithin{definition}{section}}

\let\cite\citep

\renewcommand{\tilde}{\widetilde}
\renewcommand{\hat}{\widehat}

\newcommand{\child}{\mathsf{child}}

\newcommand{\AX}[1]{\textnormal{#1}}
\DeclareMathOperator{\lca}{lca}

\newcommand{\symdiff}{\mathbin{\triangle}}
\def\arrowedvec{\mathaccent"017E}
\newcommand{\G}{\arrowedvec{G}} 
\newcommand{\E}{\arrowedvec{E}}
\newcommand{\tT}{{\hat t_T}}
\newcommand{\cupdot}{\charfusion[\mathbin]{\cup}{\cdot}}

\newcommand{\sthin}{\mathsf{S}}
\newcommand{\sigmasthin}{\sigma_{/\mathsf{S}}}

\makeatletter
\def\moverlay{\mathpalette\mov@rlay}
\def\mov@rlay#1#2{\leavevmode\vtop{%
   \baselineskip\z@skip \lineskiplimit-\maxdimen
   \ialign{\hfil$\m@th#1##$\hfil\cr#2\crcr}}}
\newcommand{\charfusion}[3][\mathord]{
    #1{\ifx#1\mathop\vphantom{#2}\fi
        \mathpalette\mov@rlay{#2\cr#3}
      }
    \ifx#1\mathop\expandafter\displaylimits\fi}
\DeclareRobustCommand\bigop[1]{%
  \mathop{\vphantom{\sum}\mathpalette\bigop@{#1}}\slimits@
}
\newcommand{\bigop@}[2]{%
  \vcenter{%
    \sbox\z@{$#1\sum$}%
    \hbox{\resizebox{\ifx#1\displaystyle.9\fi\dimexpr\ht\z@+\dp\z@}{!}{$\m@th#2$}}%
  }%
}
\makeatother

\usepackage{wasysym}
\newcommand{\ROOT}{\circledcirc}
\newcommand{\LEAF}{\odot}
\newcommand{\SPEC}{\newmoon}
\newcommand{\HGT}{\triangle}
\newcommand{\DUPL}{\square}
\newcommand{\EHGT}{\mathcal{E}}
\newcommand{\Th}{\ensuremath{T_{\mathcal{\overline{E}}}}}

\newcommand{\bigjoin}{\bigop{\Join}}
\DeclareMathOperator{\join}{\Join}

\newcommand{\hc}{\emph{hc}}

\providecommand{\keywords}[1]{\textbf{\textit{Keywords: }} #1}

\newtheorem{theorem}{Theorem}%[section]
\newtheorem{corollary}{Corollary}%[theorem]
\newtheorem{proposition}{Proposition}%[theorem]
\newtheorem{lemma}{Lemma}%[theorem]
\newtheorem{definition}{Definition}%[section]
\newtheorem{fact}{Observation}
\begin{document}

\title{Best Match Graphs and Reconciliation of Gene Trees with Species Trees}
    
\author[1]{Manuela Gei{\ss}}
\author[2]{Marcos E.\ Gonz{\'a}lez}
\author[2]{Alitzel L{\'o}pez S{\'a}nchez}
\author[3]{Dulce I.\ Valdivia}
\author[4,5]{Marc Hellmuth}
\author[2]{Maribel Hern{\'a}ndez Rosales}
\author[1,6-12]{Peter F.\ Stadler}

\affil[1]{Bioinformatics Group, Department of Computer Science; and 
	Interdisciplinary Center of Bioinformatics, University of Leipzig,
	H{\"a}rtelstra{\ss}e 16-18, D-04107 Leipzig, Germany}
\affil[2]{CONACYT-Instituto de Matem{\'a}ticas, UNAM Juriquilla,
	Blvd.\ Juriquilla 3001,
	76230 Juriquilla, Quer{\'e}taro, QRO, M{\'e}xico}
\affil[3]{Universidad Aut{\'o}noma de Aguascalientes, 
	Centro de Ciencias B{\'a}sicas,
	Av.\ Universidad 940,
	20131 Aguascalientes, AGS, M{\'e}xico;
	Instituto de Matem{\'a}ticas, UNAM Juriquilla,
	Blvd.\ Juriquilla 3001, 
	76230 Juriquilla, Quer{\'e}taro, QRO, M{\'e}xico}
\affil[4]{Institute of Mathematics and Computer Science, 
	University of Greifswald, Walther-Rathenau-Stra{\ss}e 47, 
	D-17487 Greifswald, Germany}
\affil[5]{Center for Bioinformatics, Saarland University, Building E 2.1, 
	P.O.\ Box 151150, D-66041 Saarbr{\"u}cken, Germany}
\affil[6]{German Centre for Integrative Biodiversity Research (iDiv)
	Halle-Jena-Leipzig} 
\affil[7]{Competence Center for Scalable Data Services
	and Solutions} 
\affil[8]{Leipzig Research Center for Civilization Diseases,
	Leipzig University,
	H{\"a}rtelstra{\ss}e 16-18, D-04107 Leipzig}
\affil[9]{Max-Planck-Institute for Mathematics in the Sciences,
	Inselstra{\ss}e 22, D-04103 Leipzig}
\affil[10]{Inst.\ f.\ Theoretical Chemistry, University of Vienna,
	W{\"a}hringerstra{\ss}e 17, A-1090 Wien, Austria}
\affil[11]{Facultad de Ciencias, Universidad National de Colombia, Sede
	Bogot{\'a}, Colombia}
\affil[12]{Santa Fe Institute, 1399 Hyde Park Rd., Santa Fe,
	NM 87501, USA}

\date{}
\normalsize

\maketitle

\begin{abstract} 
  A wide variety of problems in computational biology, most notably the
  assessment of orthology, are solved with the help of reciprocal best
  matches. Using an evolutionary definition of best matches that captures
  the intuition behind the concept we clarify rigorously the relationships
  between reciprocal best matches, orthology, and evolutionary events under
  the assumption of duplication/loss scenarios. We show that the orthology
  graph is a subgraph of the reciprocal best match graph (RBMG). We
  furthermore give conditions under which an RBMG that is a cograph
  identifies the correct orthlogy relation. Using computer simulations we
  find that most false positive orthology assignments can be identified as
  so-called good quartets -- and thus corrected -- in the absence of
  horizontal transfer.  Horizontal transfer, however, may introduce also
  false-negative orthology assignments.

\end{abstract}

\sloppy

\section{Introduction}

The distinction between orthologous and paralogous genes has important
consequences for gene annotation, comparative genomics, as well as
<molecular phylogenetics due to their close correlation with gene function
\cite{Koonin:05}. Orthologous genes, which derive from a speciation as
their last common ancestor \cite{Fitch:70}, usually have at least
approximately equivalent functions \cite{Gabaldon:13}. Paralogs, in
contrast, tend to have related, but clearly distinct functions
\cite{Studer:09,Innan:10,Altenhoff:12,Zallot:16}. Phylogenetic studies
strive to restrict their input data to one-to-one orthologs since these
often evolve in an approximately clock-like fashion. In comparative
genomics, orthologs serve as anchors for chromosome alignments and thus are
an important basis for synteny-based methods \cite{Sonnhammer:14}.

Despite its practical importance, the mathematical interrelationships of
empirical ``pairwise best hits'' on the one hand, and reconciliations of
gene and species trees on the other hand have remained largely unexplored.
Practical workflows for orthology assignment directly use pairwise best
hits as initial estimate of orthologous gene pairs. Many of the commonly
used methods for orthology-identification, such as \texttt{OrthoMCL}
\cite{Li:03}, \texttt{ProteinOrtho} \cite{Lechner:14a}, \texttt{OMA}
\cite{Roth:08}, or \texttt{eggNOG} \cite{Jensen:08}, belong to this
class. Extensive benchmarking \cite{Altenhoff:16,Nichio:17} has shown that
these tools perform at least as well as methods such as
\texttt{Orthostrapper} \cite{Storm:02}, \texttt{PHOG} \cite{Datta:09},
\texttt{EnsemblCompara} \cite{Vilella:09}, or \texttt{HOGENOM}
\cite{Dufayard:05} that first independently reconstruct a gene tree $T$ and
a species tree $S$ and then determine orthologous and paralogous genes.

\begin{figure}
\centering
\includegraphics[width=0.9\textwidth]{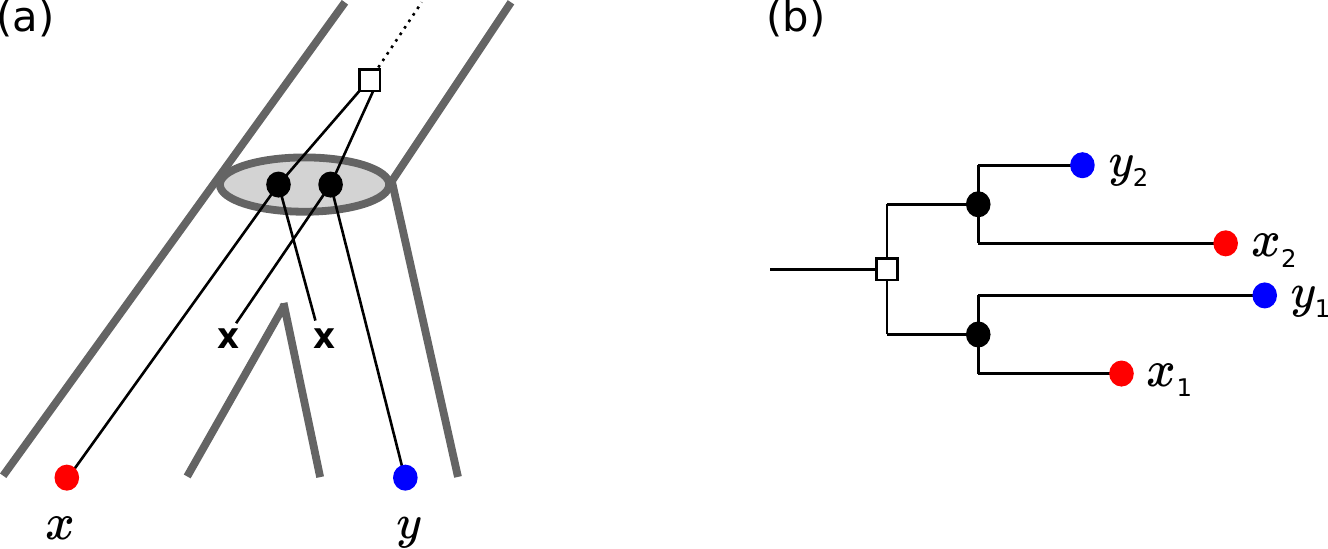} 
\caption{Pairwise best hits are not equivalent to orthology.  (a)
  Complementary losses of ancient paralogs following a later speciation
  event leaves only a single member of the gene family in each
  species. Hence, $x$ and $y$ are reciprocal best matches but not orthologs
  since their last common ancestor by construction is a duplication
  event. (b) Lineage specific rate differences between paralogs cause
  discrepancies between best hits and best matches. Here, the branch length
  in the tree represents sequence dissimilarity. In this example, the
  species (indicated by the leaf color) retain copies of the two paralogs
  originating from a duplication event pre-dating the separation of red and
  blue. While the gene $x_2$ evolves faster in the red species, the
  situation is reversed for $y_2$ in the blue species. While $\{x_1,y_1\}$
  and $\{x_2,y_2\}$ are orthologs and reciprocal best matches in the
  evolutionary sense, neither appears as a reciprocal best hit in terms of
  similarity (i.e., branch length). The only reciprocal best hit is
  $\{x_1,y_2\}$, which is neither a best match nor a pair of orthologs.  }
\label{fig:introxampl}
\end{figure}

The intuition behind the pairwise best hit approach is that a gene $y$ in
species $s$ can only be an ortholog of a gene $x$ in species $r$ if $y$ is
the closest relative of $x$ in $s$ and $x$ is at the same time the closest
relative of $y$ in $r$. Evolutionary relatedness is defined in terms of an
-- often unknown -- phylogenetic tree $T$. The notion of a best match or
closest relative thus is made precise by considering the last common
ancestors in $T$: $y$ is a best match for $x$ if the least common ancestor
$\lca_T(x,y)$ is not further away from $x$ (and thus not closer to the root
of the tree) than $\lca_T(x,y')$ for any other gene $y$ in species
$s$. This formally defines the \emph{best match} relation studied in
\cite{Geiss:19a}. The \emph{reciprocal best match} relation identifies the
pairs of genes that are mutually closest relatives between pairs of
species, see \cite{Geiss:19x}.

Two approximations are introduced when pairwise best hit approaches are
employed for orthology assessment.  First, it is well known that two genes
can be mutual closest relatives without being orthologs. The usual example
is the complementary loss of ancestrally present paralogs following a gene
duplication (Fig.~\ref{fig:introxampl}a).  Second, pairwise best hits as
determined by sequence (dis)similiarity are not necessarily pairs of most
closely related genes and \textit{vice versa}, evolutionarily most closely
related gene pairs do not necessarily appear as pairwise best hits
(Fig.~\ref{fig:introxampl}b).

We argue, therefore, that the relationship of pairwise best hits and
orthology has to be understood in (at least) two conceptually and
practically separate steps:
\begin{enumerate}
\item What is the relationship of pairwise best hits and reciprocal best
  matches?
\item What is the relation of reciprocal best matches and orthology?
\end{enumerate}
In this contribution we focus on the second question, which is a largely
mathematical problem. The first question, which is primarily a question of
inference from data, will be investigated in a companion paper that makes
use of some of the mathematical results derived here. The main aim of the
present contribution is, therefore, to connect formal results on the
structure of the orthology relation and the associated reconciliation maps
and gene trees with recent results on the mathematical structure of
(reciprocal) best match relations.

Symbolic ultrametrics \cite{Boecker:98} and 2-structures
\cite{ER1:90,ER2:90} provided a basis to show that orthology relations are
essentially equivalent to cographs
\cite{Hellmuth:13a,Hellmuth:17a,HW:16book}.  Moreover, in the absence of
horizontal gene transfer (HGT), reconciliation maps for an event-labeled
gene tree exist if and only if the species tree $S$ displays all triples
rooted in a speciation event that have leaves from three distinct species
\cite{HernandezRosales:12a,Hellmuth2017}.  This shows that it is possible
to infer species phylogenies from empirical estimates of orthology
\cite{Hellmuth:15a,Lafond:16,Lafond:14,Dondi:17}.  Although it is possible
to generalize many of the results, such as the characterization of
reconciliation maps for event-labeled gene trees to scenarios with
horizontal gene transfer \cite{Nojgaard:18a,HHM:19,Hellmuth2017} this
remains an active area of research.

Best matches as a mathematical structure have been studied only very
recently. \citet{Geiss:19a} gave two alternative characterizations of best
match digraphs and showed that they can be recognized in polynomial time.
In particular, there is a unique least resolved tree for each best match
digraph, which is displayed by the gene tree and can also be computed in
polynomial time. Reciprocal best matches naturally appear as the symmetric
part of these digraphs. Somewhat surprisingly, the undirected reciprocal
best match graphs seem to have a much more difficult structure
\cite{Geiss:19x}.

Although pairwise best hit methods do not attempt to explicitly construct
the gene tree $T$, they still make the assumption that there is some
underlying phylogeny for the provided homologous genes.  The distinction of
orthology and paralogy then amounts to assigning event labels
(``speciation'', ``duplication'', and possibly ``HGT'') to the inner
vertices of $T$.  While it is true that any gene tree, and thus also any
best match graph, can be reconciled with any species tree
\cite{Guigo:96,Page:97,Gorecki:06}, such a reconciliation may imply
unrealistically many duplication and deletion events. Moreover, the
existence of reconciliation maps for $T$ to some species tree cannot
generally be ensured, if the event labels are given
\cite{HernandezRosales:12a,Hellmuth2017}. Hence, the best match relation
(which constrains the gene tree \cite{Geiss:19a}), the event labels, the
existence of one or a particular reconcilation map, and the species tree
depend on each other or at least do constrain each other. In this
contribution we explore these dependencies in detail in the absence of
horizontal gene transfer.

We show that, in this setting, the true orthology graph (TOG) is a subgraph
of the reciprocal best match graph (RBMG). In other words, reciprocal best
matches can only produce false positive orthology assignments as long as
the evolution of a gene family proceeds via duplications, losses, and
speciations.  Computer simulations show that in broad parameter range the
TOG and RBMG are very similar, proving an \emph{a posteriori} justification
for the use of reciprocal best matches in orthology estimation. In
addition, we characterize a subset of the ``false positive'' edges in the
RBMG that cannot be present in the TOG.  Experimental results show that --
using so-called good quartets -- it is possible to remove nearly all false
positive orthology assignments. Our aim here is to understand those sources
of error and ambiguities in orthology detection that still persist even if
reciprocal best matches are inferred with perfect accuracy. Therefore, all
computer simulations reported here use perfect data as input. In a
companion paper, we address the question how well reciprocal best matches
can be inferred from (dis)ssimilarity data, and what can be done to make
this inital step more accurate. Finally, we discuss how these results can
potentially be generalized to the case that the evolutionary scenarios
contain HGT.

\section{Preliminaries}

A \emph{planted (phylogenetic) tree} is a rooted tree $T$ with vertex set
$V(T)$ and edge set $E(T)$ such that (i) the root $0_T$ has degree $1$ and
(ii) all inner vertices have degree $\deg_T(u)\ge 3$.  We write $L(T)$ for
the leaves (not including $0_T$) and $V^0=V(T)\setminus(L(T)\cup\{0_T\})$
for the inner vertices (also not including $0_T$).  To avoid trivial cases,
we will always assume that $|L(T)|\geq 2$.  The \emph{conventional root}
$\rho_T$ of $T$ is the unique neighbor of $0_T$. The main reason for using
planted phylogenetic trees instead of modeling phylogenetic trees simply as
rooted trees, which is the much more common practice in the field, is that
we will often need to refer to the time before the first branching
event. Conceptually, it corresponds to explicitly representing an outgroup.
For some vertex $v\in V(T)$, we denote by $T(v)$ the subtree of $T$ that is
rooted in $v$.  Its leaf set is $L(T(v))$.

On a rooted tree $T$ we define the \emph{ancestor order} by setting
$x\prec_T y$ if $y$ is a vertex of the unique path connecting $x$ with the
root $0_T$. As usual we write $x\preceq_T y$ if $x=y$ or $x\prec_T y$. In
particular, the leaves are the minimal elements w.r.t.\ $\prec_T$, and we
have $x\preceq 0_T$ for all $x\in V(T)$. This partial order is conveniently
extended to the edge set by defining each edge to be located between its
incident vertices, i.e., if $y\prec_T x$ and $e=xy$ is an edge, we set
$y \prec_T e \prec_T x$.  In this case, we write $e=xy$ to denote that $x$
is closer to the root than $y$. If $e=xy\in E(T)$, we say that $y$ is a
\emph{child} of $x$, in symbols $y\in\child(x)$, and $x$ is the
\emph{parent} of $y$ in $T$. We sometimes also write $y\succeq_T x$
  instead of $x\preceq_T y$. Moreover, if $x\preceq_T y$ or $y\preceq_T x$
in $T$, then $x$ and $y$ are called \emph{comparable}, otherwise the two
vertices are \emph{incomparable}.

For a non-empty subset of vertices $A\subseteq V$ of a rooted tree
$T=(V,E)$, we define $\lca_T(A)$, the \emph{last common ancestor of $A$},
to be the unique $\preceq_T$-minimal vertex of $T$ that is an ancestor of
every vertex in $A$. For simplicity we write
$\lca_T(x_1,\dots,x_k)\coloneqq\lca_T(\{x_1,\dots,x_k\})$ for a set
$A=\{x_1,\dots,x_k\}$ of vertices.  The definition of $\lca_T(A)$ is
conveniently extended to edges by setting
$\lca_T(x,e) \coloneqq \lca_T(\{x\}\cup e)$ and
$\lca_T(e,f) \coloneqq \lca_T(e\cup f)$, where the edges $e,f\in E(T)$ are
simply treated as sets of vertices. We note for later reference that
$\lca(A\cup B)=\lca(\lca(A),\lca(B))$ holds for non-empty vertex sets $A$,
$B$ of a tree.

Binary trees on three leaves are called \emph{triples}. We say that a
triple $xy|z$ is \emph{displayed} in a rooted tree $T$ if $x,y,$ and $z$
are leaves in $T$ and the path from $x$ to $y$ does not intersect the path
from $z$ to the root. The set of all triples that are displayed by the tree
$T$, is denoted by $r(T)$ and a triple set $R$ is said to be
\emph{compatible} if there exists a tree $T$ that displays $R$, i.e.,
$R\subseteq r(T)$.
  
Denote by $L(S)$ a set of species and denote by $\sigma: L(T)\to L(S)$ the
map that assigns to each gene $x\in L(T)$ a species $\sigma(x)\in L(S)$.  A
tree $T$ together with such a map $\sigma$ is denoted by $(T,\sigma)$ and
called \emph{leaf-colored tree}.
\begin{definition}
  Let $(T,\sigma)$ be a leaf-colored tree. A leaf $y\in L(T)$ is a
  \emph{best match} of the leaf $x\in L(T)$ if $\sigma(x)\neq\sigma(y)$ and
  $\lca(x,y)\preceq_T \lca(x,y')$ holds for all leaves $y'$ from species
  $\sigma(y')=\sigma(y)$.  The leaves $x,y\in L(T)$ are \emph{reciprocal
    best matches} if $y$ is a best match for $x$ and $x$ is a best match
  for $y$.
\end{definition}
The directed graph $\G(T,\sigma)$ with vertex set $L(T)$, vertex-coloring
$\sigma$, and edges defined by the best matches in $(T,\sigma)$ is known as
\emph{colored best match graph} (BMG) \cite{Geiss:19a}. The undirected
graph $G(T,\sigma)$ with vertex set $L(T)$, vertex-coloring $\sigma$, and
edges defined by the reciprocal best matches in $(T,\sigma)$ is known as
colored \emph{reciprocal best match graph} (RBMG) \cite{Geiss:19x}. We
sometimes write $n$-BMG, resp., $n$-RBMG to specify the number $n$ of
colors.

Throughout this contribution, $G=(V,E)$ and $\G=(V,\E)$ denote simple
undirected and simple directed graphs, respectively. We distinguish
directed arcs $(x,y)$ in a digraph $\G$ from edges $xy$ in an undirected
graph $G$ or tree $T$. For an undirected graph $G$ we denote by
$N(x)=\{y \mid y \in V(G), xy\in E(G)\}$ the neighborhood of some vertex
$x$ in $G$. The \emph{disjoint union} $G\cupdot H$ of two graphs $G=(V,E)$
and $H=(W,F)$ has vertex set $V\cupdot W$ and edge set $E\cupdot F$. Their
\emph{join} has again vertex set $V\cupdot W$ and its edge set is given by
$E(G\join H)=E\cupdot F \cupdot \{xy\mid x\in V,y\in W\}$. Thus the join of
$G$ and $H$ is obtained by connecting every vertex of $G$ to every vertex
of $H$.

A class of undirected graphs that plays an important role in this
contribution are \emph{cographs}, which are recursively defined
\cite{Corneil:81}:
\begin{definition}
  An undirected graph $G$ is a cograph if one of the following conditions
  is satisfied:
  \begin{itemize}
  \item[(1)] $G=K_1$, 
  \item[(2)] $G= H \join H'$, where $H$ and $H'$ are cographs,
  \item[(3)] $G= H \cupdot H'$, where $H$ and $H'$ are cograph.
  \end{itemize}
  \label{def:cograph}
\end{definition}
An undirected graph is a cograph if and only if it does not contain an
induced $P_4$ (path on four vertices) \cite{Corneil:81}. 

Every cograph $G$ is associated with a set of phylogenetic trees
  $\mathfrak{T}_G$, usually referred to as the \emph{cotrees} of $G$.
  Every cotree $T_G\in \mathfrak{T}_G$ correspond to a possible recursive
  construction of $G$. Since both the disjoint union and the join operation
  is associative, it is possible to join or unify two or more component
  cographs in a single construction step.  The leaves of $T_G$ correspond
  to the vertices of $G$. Each interior vertex of $T_G$ corresponds to either
  a join or a disjoint union operation. Its child-subtrees, furthermore,
  are exactly the cotrees of the component cographs that are joined or
  disjointly unified, respectively. The event type associated with an inner
  vertex $u$ will be denoted by $t_G(u)$. Each vertex $u$ of $T_G$ can be
  associated with an induced subgraph $G[L(T_G(u))]$.  A cotree $T_G$ is
  called \emph{discriminating} if any two adjacent inner nodes represent
  different types of events. If $T_G\in\mathfrak{T}_G$ and $T_G'$ is
  obtained from $T_G$ by contracting a non-discriminating edge, i.e., an
  edge $uv$ with $t_G(u)=t_G(v)$, then $T_G'\in\mathfrak{T}_G$. Every
  cograph has a unique discriminating cotree, which is obtained from any of
  its cotrees by contracting all non-discriminating edges
  \cite{Corneil:81}. We note, finally, that the discriminating cotree of
  $G$ coincides with the modular decomposition tree of $G$.
 
\section{Reconciliation Maps, Event Labelings, and Orthology Relations}

A \emph{gene} tree $T=(V,E)$ and a species tree $S=(W,F)$ are planted
phylogenetics trees on a set of (extant) genes $L(T)$ and species $L(S)$,
respectively.  We assume that we know which gene comes from which
species. Mathematically, this knowledge is represented by a map
$\sigma\colon L(T)\to L(S)$ that assigns to each gene the species in whose
genome it resides.  Best match approaches start from a set of genes taken
from a set of species.  Hence, the ``gene-species-association'' is known.
Moreover, species without sampled genes do not affect the best match graph
and we can w.l.o.g.\ assume that $\sigma$ is a surjective map to avoid
trivial cases.  Note, however, that the definitions and results presented
below naturally extend to general maps $\sigma$.  We write $(T,\sigma)$ for
a gene tree with given map $\sigma$.

An \emph{evolutionary scenario} comprises a gene tree and a species
  tree together with a map $\mu$ from $T$ to $S$ that identifies the
  locations in the species tree $S$ at which evolutionary events took place
  that are represented by the vertices of the gene tree $T$.  The
properties of the map $\mu$ of course depend on which types of evolutionary
events are considered. In order to model evolutionary scenarios we assume
that evolutionary events of different types do not occur concurrently. In
particular, speciation and duplication are always strictly temporally
ordered. Gene duplications therefore always occur along the edges of the
species tree. Vertices on $T$ that model speciation events, on the other
hand, must be mapped to inner vertices of $S$.

\emph{From here on we will consider only Duplication/Loss secenarios, that
  is we explicitly exclude horizontal gene transfer (HGT). We will briefly
  discuss the effects of HGT in Section~\ref{sect:HGT}.}

\begin{definition}[Reconciliation Map] \label{def:mu} Let $S=(W,F)$ and
  $T=(V,E)$ be two planted phylogenetic trees and let
  $\sigma\colon L(T)\to L(S)$ be a surjective map. A \emph{reconciliation}
  from $(T,\sigma)$ to $S$ is a map $\mu\colon V\to W\cup F$ satisfying
  \begin{description}
  \item[\AX{(R0)}] \emph{Root Constraint.} $\mu(x) = 0_S$ if and only if
    $x = 0_T$.
  \item[\AX{(R1)}] \emph{Leaf Constraint.}  If $x\in L(T)$, then
    $\mu(x)=\sigma(x)$.
  \item[\AX{(R2)}] \emph{Ancestor Preservation.}  $x\prec_{T} y$ implies
    $\mu(x)\preceq_S \mu(y)$.
  \item[\AX{(R3)}] \emph{Speciation Constraints.}  Suppose $\mu(x)\in W^0$.
    \begin{itemize}
    \item[(i)] $\mu(x) = \lca_S(\mu(v'),\mu(v''))$ for at least two
      distinct children $v',v''$ of $x$ in $T$.
    \item[(ii)] $\mu(v')$ and $\mu(v'')$ are incomparable in $S$ for any
      two distinct children $v'$ and $v''$ of $x$ in $T$.
    \end{itemize}
  \end{description}
  \label{def:reconc-map}
\end{definition}           

Several alternative definitions of reconciliation maps for Duplication/Loss
scenarios have been proposed in the literature, many of which have been
shown to be equivalent. Nevertheless, we add yet another one because
earlier variants do not clearly separate conditions pertaining to the
structural congruence of gene tree and species tree (Axioms \AX{(R0)},
\AX{(R1)}, and \AX{(R2)}) from conditions that (implicitly) distinguish
event types, here \AX{(R3.i)} and \AX{(R3.ii)}. This axiom system also
generalizes easily to situations with horizontal transfer as we shall see
in Section~\ref{sect:HGT}.  We proceed by showing that it is equivalent to
axioms that are commonly used in the literature, see e.g.\
\citet{Gorecki:06,Vernot:08,Doyon:11,Rusin:14,Hellmuth2017,Nojgaard:18a},
and the references therein.
\begin{lemma}
  Let $\mu$ be a map from $(T=(V,E), \sigma)$ to $S=(W,F)$ that satisfies
  \AX{(R0)} and \AX{(R1)}.  Then, $\mu$ satisfies Axioms \AX{(R2)} and
  \AX{(R3)} if and only if $\mu$ satisfies
  \begin{description}
  \item[\AX{(R2')}] \emph{Ancestor Constraint.}		\\
    Suppose $x,y\in V$ with $x\prec_{T} y$.
    \begin{itemize}
    \item[(i)] If $\mu(x), \mu(y) \in F$, then $\mu(x)\preceq_S \mu(y)$,
    \item[(ii)] otherwise, i.e., at least one of $\mu(x)$ and $\mu(y)$ is
      contained in $W$, $\mu(x)\prec_S\mu(y)$.
  \item[\AX{(R3')}] \emph{Inner Vertex Constraint.}  	\\
    If $\mu(x)\in W^0$, then
    \begin{itemize}
    \item[(i)] $\mu(x) = \lca_S(\sigma(L(T(x))))$ and  
    \item[(ii)] $\mu(v')$ and $\mu(v'')$ are incomparable in $S$ for any
      two distinct children $v'$ and $v''$ of $x$ in $T$.
    \end{itemize}
     \end{itemize}
  \end{description}
\end{lemma}
\begin{proof}
  Assume first that \AX{(R2)} and \AX{(R3)} are satisfied for $\mu$.\\
  Then property \AX{(R2'.i)} is satisfied since it is the restriction of
  \AX{(R2)} to  $\mu(x),\mu(y)\in F$.\\
  To see that \AX{(R2'.ii)} holds, let $x\prec_T y$ and $\mu(x)\in W$ or
  $\mu(y)\in W$. Assume first that $\mu(y)\in W$.  Property \AX{(R2)}
  implies $\mu(x)\preceq_S \mu(y)$.  Let $v$ be the child of $y$ that lies
  on the path from $y$ to $x$ in $T$, i.e., $x\preceq_T v \prec_T
  y$. Assume for contradiction that $\mu(x) = \mu(y)$.  By Property
  \AX{(R2)} we have $\mu(x) = \mu(v) = \mu(y)$.  For every other child $v'$
  of $y$, Property \AX{(R2)} implies $\mu(v') \preceq_S
  \mu(y)=\mu(v)$. Thus, $\mu(v)$ and $\mu(v')$ are comparable; a
  contradiction to \AX{(R3.ii)}. Hence, $\mu(x) \prec_S \mu(y)$ and
  \AX{(R2'.ii)} is satisfied. Now suppose $\mu(x)\in W$ and assume for
  contradiction that $\mu(x) = \mu(y)$. Thus $\mu(y)\in W$ and we can apply
  the same arguments as above to conclude that \AX{(R3.ii)} is not
  satisfied. Hence, $\mu(x) \prec_S \mu(y)$ and
  \AX{(R2'.ii)} is satisfied.\\
  In order to show that \AX{(R3')} is satisfied, let $x\in V$ such that
  $\mu(x)\in W^0$. Properties \AX{(R3'.ii)} and \AX{(R3.ii)} are
  equivalent.  It remains to show that \AX{(R3'.i)} is satisfied. From
  \AX{(R2)} we infer $\mu(y)\preceq_S\mu(x)$ for all
  $y\in \bigcup_{v\in\child(x)} L(T(v)) = L(T(x))$.  Thus,
  \begin{equation}
    \lca_S(\sigma(L(T(x)))) \preceq \mu(x).
  \end{equation}
  Property \AX{(R3.i)} implies that there are two distinct children
  $v',v''\in \child(x)$ with $\mu(x) = \lca_S(\mu(v'),\mu(v''))$. Again
  using \AX{(R3.ii)}, we know that the images $\mu(v')$ and $\mu(v'')$ are
  incomparable in $S$. The latter together with $\mu(y) \preceq_S \mu(v')$
  for all $y\in L(T(v'))$ and $\mu(y') \preceq_S \mu(v'')$ for all
  $y'\in L(T(v''))$ implies
  \begin{equation*}
    \lca_S(\mu(v'),\mu(v'')) = \lca_S(\sigma(L(T(v'))) \cup \sigma(L(T(v''))))
    \preceq_S \lca_S(\sigma(L(T(x)))).
  \end{equation*}
  In summary,
  $\lca_S(\sigma(L(T(x))))\preceq_S \mu(x) = \lca_S(\mu(v'),\mu(v''))
  \preceq_S \lca_S(\sigma(L(T(x))))$ implies that
  $\mu(x) = \lca_S(\sigma(L(T(x))))$ and Property \AX{(R3'.i)} is satisfied.\\
  Therefore, \AX{(R2)} and \AX{(R3)} imply \AX{(R2')} and \AX{(R3')}.

  Conversely, assume now that \AX{(R2')} and \AX{(R3')} are satisfied for
  $\mu$.  Clearly \AX{(R2')} implies \AX{(R2)}, and \AX{(R3'.ii)} implies
  \AX{(R3.ii)}.  It remains to show that \AX{(R3.i)} is satisfied.  Let
  $\mu(x)\in W^0$.  By \AX{(R2'.ii)} we have $\mu(x) \succ_S \mu(v_i)$ for
  all children $v_i\in \child(x) =\{v_1,\dots,v_k\}$, $k\geq 2$.
  Therefore, $\mu(x) \succeq_S \lca_S(\mu(v_1), \dots, \mu(v_k))$.  By
  \AX{(R3'.ii)}, the images $\mu(v_1), \dots, \mu(v_k)$ are pairwise
  incomparable in $S$. The latter and \AX{(R2'.i)} imply
  $\lca_S(\mu(v_1), \dots, \mu(v_k)) =
  \lca_S(\bigcup_{i=1}^k\sigma(L(T(v_i)))) = \lca_S(\sigma(L(T(x)))) =
  \mu(x)$.  It is easy to verify that
  $\lca_S(\mu(v_1), \dots, \mu(v_k)) = \lca_S(\mu(v'), \mu(v''))$ for at
  least two children $v',v''\in\child(x)$ is always satisfied. Hence,
  $\mu(x) = \lca_S(\mu(v'), \mu(v''))$ for some $v',v''\in\child(x)$ and
  thus, \AX{(R3.i)} is satisfied.\\
  Therefore, \AX{(R2')} and \AX{(R3')} imply \AX{(R2)} and \AX{(R3)}.
  \qed
\end{proof}

A reconciliation map $\mu$ from $(T,\sigma)$ to a species tree $S$
implicitly determines whether an inner node of $T$ corresponds to a
speciation or a duplication. Since we assume that distinct events are
represented by distinct nodes of the gene tree, all duplication events are
mapped to the edges of $S$. Vertices of $T$ mapped to vertices of $S$ thus
represent speciations.  We formalize this idea as follows:
\begin{definition} 
  Given a reconciliation map $\mu$ from $(T,\sigma)$ to $S$, the
  \emph{event labeling on $T$ (determined by $\mu$)} is the map
  $t_\mu:V(T)\to \{\ROOT,\LEAF,\SPEC,\DUPL\}$ given by:
\begin{equation*}
  t_\mu(u) = \begin{cases}
    \ROOT & \, \text{if } u=0_T \text{, i.e., } \mu(u)=0_S  \text{ (root)}\\
    \LEAF & \, \text{if } u\in L(T) \text{, i.e., } \mu(u)\in L(S)
    \text{ (leaf)}\\
    \SPEC & \, \text{if } \mu(u)\in V^0(S)  \text{ (speciation)}\\
    \DUPL & \, \text{else, i.e., } \mu(u)\in E(S) \text{ (duplication)}\\
\end{cases}
\end{equation*}
\label{def:event-rbmg}
\end{definition}
The symbols $\ROOT$ and $\LEAF$ identify the planted root $0_T$ and the
leaves of $T$, respectively. Inner vertices are labeled $\DUPL$ for
duplication and $\SPEC$ for speciation, respectively.

The event labeling $t_{\mu}$, by definition, is completely determined by a
reconciliation map $\mu$. This raises two related questions: (1) which
pattern of event labels can arise for reconciliation maps, and (2) what
restriction does a given event labeling impose on the reconciliation map?
To study these questions, we consider event-labeled trees $(T,t)$ where the
\emph{event labeling} of $T$ is a map
$t:V(T)\to \{\ROOT,\LEAF,\SPEC,\DUPL\}$ satisfying $t(0_T)=\ROOT$,
$t(x)=\LEAF$ for all $x\in L(T)$, and $t(x)\in\{\DUPL,\SPEC\}$ for
$x\in V^0(T)$. We interpret $\DUPL$ as gene duplication event and $\SPEC$
as speciation event.

A simple consequence of the Axioms \AX{(R0)}-\AX{(R3)} is the following
result which is stated here for later reference. For the sake of
completeness, we also provide a short proof.
\begin{lemma} 
  Let $\mu$ be a reconciliation map from the leaf-colored tree $(T,\sigma)$
  to $S=(W,F)$ and $x\in V(T)$ a vertex with $\mu(x)\in W^0$.  Then,
  $\sigma(L(T(v')))\cap \sigma(L(T(v''))) = \emptyset$ for any two distinct
  $v',v''\in \child(x)$.
  \label{lem:disjointSpecies}
\end{lemma}
\begin{proof}
  Assume for contradiction that there is a vertex
  $z\in \sigma(L(T(v')))\cap \sigma(L(T(v'')))$.  By Condition \AX{(R2')},
  we have $\mu(x)\succ_S\mu(v')\succeq_S z$ and
  $\mu(x)\succ_S\mu(v'')\succeq_S z$.  Thus, there is a path $P_1$ from
  $\mu(x)$ to $z$ that contains $\mu(v')$ and a path $P_2$ from $\mu(x)$ to
  $z$ that contains $\mu(v'')$.  However, Condition \AX{(R3.ii)} implies
  that $\mu(v')$ and $\mu(v'')$ are incomparable in $S$, that is, the
  subtree of $S$ consisting of the two paths $P_1$ and $P_2$ must contain a
  cycle; a contradiction.  \qed
\end{proof}
Lemma~\ref{lem:disjointSpecies} has a simple interpretation: Since
$\mu(x)\in W^0$, we have $t_{\mu}(x)=\SPEC$, i.e., $x$ represents a
speciation. The lemma thus states that any two subtrees of $T$ rooted in
distinct children of a speciation event are composed of genes from disjoint
sets of species. It suggests the following
\begin{definition}
  An event labeling $t:V(T)\to \{\ROOT,\LEAF,\SPEC,\DUPL\}$ is
  \emph{well-formed} if $t(x)=\SPEC$ implies that
  $\sigma(L(T(v')))\cap \sigma(L(T(v''))) = \emptyset$ for any two distinct
  $v',v''\in \child(x)$.
  \label{def:well-formed}
\end{definition}

Lemma~\ref{lem:disjointSpecies} suggests to ask for a characterization of
the event maps $t$ for a given leaf-labeled tree $(T,\sigma)$ for which
$(T,t,\sigma)$ admits a reconciliation map to some species tree. Definition
\ref{def:well-formed} suggests to start by considering among the
well-formed event labelings the one that designates every vertex of $T$
that is not identified as a duplication because it violates
Lemma~\ref{lem:disjointSpecies}.
\begin{definition} 
  Let $(T,\sigma)$ be a leaf-labeled tree.  The \emph{extremal event
    labeling} of $T$ is the map $\tT:V(T)\to\{\ROOT,\LEAF,\SPEC,\DUPL\}$
  defined for $u\in V(T)$ by
\begin{equation*}
  \tT(u) = \begin{cases}
    \ROOT & \, \text{if } u=0_{T} \\
    \LEAF & \, \text{if } u\in L(T) \\
    \DUPL & \, \text{if there are two children } u_1,u_2\in \child(u)
    \text{ such that}\\
    & \qquad \sigma(L(T(u_1)))\cap \sigma(L(T(u_2)))\neq\emptyset\\
    \SPEC & \, \text{otherwise} \\
  \end{cases}
\end{equation*}
\label{def:event-rbmg2}
\end{definition}
The extremal event labeling $\tT$ is completely determined by $(T,\sigma)$.
By construction, if $u\in V^0(T)$ is a duplication w.r.t.\ to the extremal
event labeling $\tT(u)=\DUPL$, then $t(u)=\DUPL$ for every well-formed
event labeling $t$ on $(T,\sigma)$. 

It is a well-known result that it is always possible to reconcile a given
pair of gene tree $T$ and species tree $S$, see e.g.\
\cite{Guigo:96,Page:97,Gorecki:06}. For convenience, we include a short
direct proof of this fact.
\begin{lemma}
  \label{lem:norestriction}
  For every tree $(T=(V,E),\sigma)$ there is a reconciliation map $\mu$ to
  any species tree $S$ with leaf set $L(S)=\sigma(L(T))$.
\end{lemma}
\begin{proof}
  Let $S=(W,F)$ be an arbitrary species tree with leaf set $L(S)$ and
  $e_0 = 0_S\rho_S$ be the unique root-edge of $S$.  Set $\mu(0_T) = 0_S$
  and $\mu(v) = \sigma(v)$ for all $v\in L(T)$.  Thus, \AX{(R0)} and
  \AX{(R1)} are satisfied. Now, set $\mu(v) = e_0$ for all
  $v\in V^0 = V\setminus (L(T)\cup \{0_T\})$.  Thus, $\mu(v)\notin W^0$ for
  all $v\in V^0$ and \AX{(R3)} is trivially satisfied. Finally, for all
  $v, v'\in V^0$ and $y\in L(T)$ with $y\prec_T v\prec_T v '$ we have by
  construction of $\mu$ that
  $\mu(y)\prec_T \mu(v) = \mu(v')\prec_T\mu(0_T)$.  Thus, \AX{(R2)} is
  satisfied.
  \qed
\end{proof}
The reconciliation map $\mu$ constructed in the proof of
Lemma~\ref{lem:norestriction} maps all inner vertices of the gene tree to
the edge above the root of the species tree $S$, and hence
$t_{\mu}(x)=\DUPL$ for all inner vertices of $T$. The root of $S$
  already contains $|L(T)|$ genes, one for each leaf of $T$. Every
  speciation event is therefore accompanied by complementary losses, and
  there are no further gene duplication events below the root.

The assignment of genes to species, i.e., a prescribed leaf coloring
$\sigma$, however, implies further restrictions. In fact, it is not
sufficient to require that the event labeling is well-formed. Instead, the
simultaneous knowledge of $(T,t,\sigma)$ gives rise to strong conditions on
the species trees $S$ with which $(T,t,\sigma)$ can be
reconciled. Following \cite{HernandezRosales:12a}, we denote by
$\mathcal{S}(T,t,\sigma)$ the set of triples $\sigma(a)\sigma(b)|\sigma(c)$
for which $ab|c$ is a triple displayed by $T$ such that (i) $\sigma(a)$,
$\sigma(b)$, $\sigma(c)$ are pairwise distinct species and (ii) the root of
the triple is a speciation event, i.e., $t(\lca(a,b,c))=\SPEC$. This set of
triples characterizes the existence of a reconciliation map:
\begin{proposition} \cite{HernandezRosales:12a,Hellmuth2017}
  \label{prop:inftriple}
  Given an leaf-labeled tree $(T,t,\sigma)$ with a well-formed event
  labeling $t$ and a species tree $S$ with $L(S)=\sigma(L(T))$, there is a
  reconciliation map $\mu:V(T)\to V(S)\cup E(S)$ such that the event
  labeling is consistent with Definition~\ref{def:event-rbmg} if and only
  if $S$ displays $\mathcal{S}(T,t,\sigma)$. In particular, $(T,t,\sigma)$
  can be reconciled with a species tree if and only if
  $\mathcal{S}(T,t,\sigma)$ is a compatible set of triples.
\end{proposition}  
An example for a $(T,t,\sigma)$ that does not admit a reconciliation map is
given in Fig.~\ref{fig:counterex_extreme} (top left). We note that the
characterization in Proposition \ref{prop:inftriple} can be evaluated in
polynomial time \cite{Hellmuth2017}.
  
\begin{figure}
  \centering
  \includegraphics[width=0.8\textwidth]{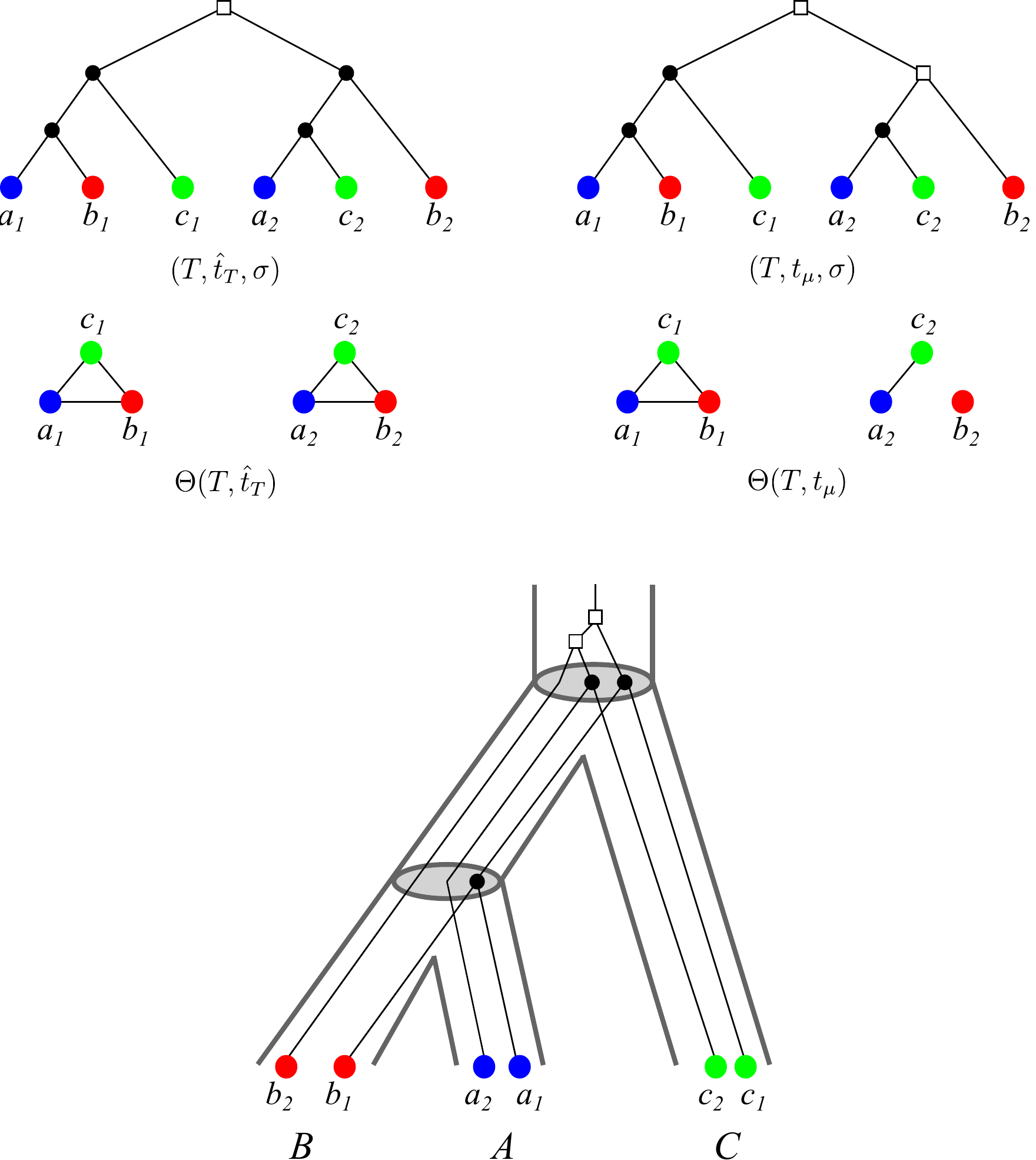}
  \caption{An example for $\Theta(T,t_{\mu})\subset \Theta(T,\tT)$.
    \emph{Top Left:} A gene tree $(T,\sigma)$ with extremal event labeling
    $\tT$, the corresponding orthology relation $\Theta(T,\tT)$ and map
    $\sigma(a_i)=A$, $\sigma(b_i)=B$ and $\sigma(c_i)=C$, $i=1,2$.  Here we
    obtain $AB|C, AC|B \in \mathcal{S}(T,\tT,\sigma)$ as conflicting
    species triples, making $\mathcal{S}(T,\tT,\sigma)$ incompatible.
    \emph{Top Right and Bottom:} The same tree $(T,\sigma)$ with another
    event labeling $t_\mu$ defined by the reconciliation map $\mu$ from
    $(T,\sigma)$ to the (tube-like) species tree $S$ as shown at the
    bottom.  The map $\mu$ is given implicitly by drawing $T$ into $S$. The
    corresponding orthology relation $\Theta(T,t_\mu)$ is shown below
    $(T,t_{\mu}, \sigma)$. Clearly, since $\mu$ exists,
    $\mathcal{S}(T,t_{\mu},\sigma) = \{AB|C\}$ is compatible (cf.\ Prop.\
    \ref{prop:inftriple}).}
  \label{fig:counterex_extreme}
\end{figure}

The event labeling $t$ on $T$ defines the orthology relation:
\begin{definition} \cite{Fitch:00}
  Two distinct leaves $x,y\in L(T)$ are \emph{orthologs (w.r.t.\ $t$)} if
  $t(\lca_T(x,y))=\SPEC$; they are \emph{paralogs} if
  $t(\lca_T(x,y))=\DUPL$.
\label{def:fitch}
\end{definition}
For completeness, we note that $t(\lca_T(x,y))=\LEAF$ if and only $x=y$,
and $0_T$ is never the $\lca$ of any of pair of leaves since the planted
root $0_T$ has degree $1$ by construction. We write $\Theta(T,t)$ for the
orthology relation obtained from $(T,t)$, i.e., the set of all unordered
pairs $\{x,y\}$ of orthologous genes in $L(T)$. For convenience we will not
distinguish between the irreflexive, symmetric binary relation
$\Theta(T,t)$ and the graph with vertex set $L(T)$ and edge set
$\Theta(T,t)$. Naturally, we say that an arbitrary relation $\Theta$ is an
orthology relation if there is an event-labeled phylogenetic tree $(T,t)$
such that $\Theta=\Theta(T,t)$. It is important to note that the orthology
relation $\Theta$ explicitly depends on the event labeling. Analogously,
one can also define the \emph{paralogy relation} $\overline{\Theta}$ by
$t(\lca_T(x,y))=\DUPL$. Both orthology and paralogy are irreflexive and
symmetric but not transitive, see Fig.\ \ref{fig:ortho}. We note that
orthology $\Theta$ and paralogy $\overline{\Theta}$ are complementary in
the graph-theoretical sense, i.e., $\{x,y\}$ is contained in exactly one of
$\Theta$ or $\overline{\Theta}$.

\begin{figure}
  \begin{tabular}{ccc}
    \begin{minipage}{0.61\textwidth}
      \includegraphics[width=\textwidth]{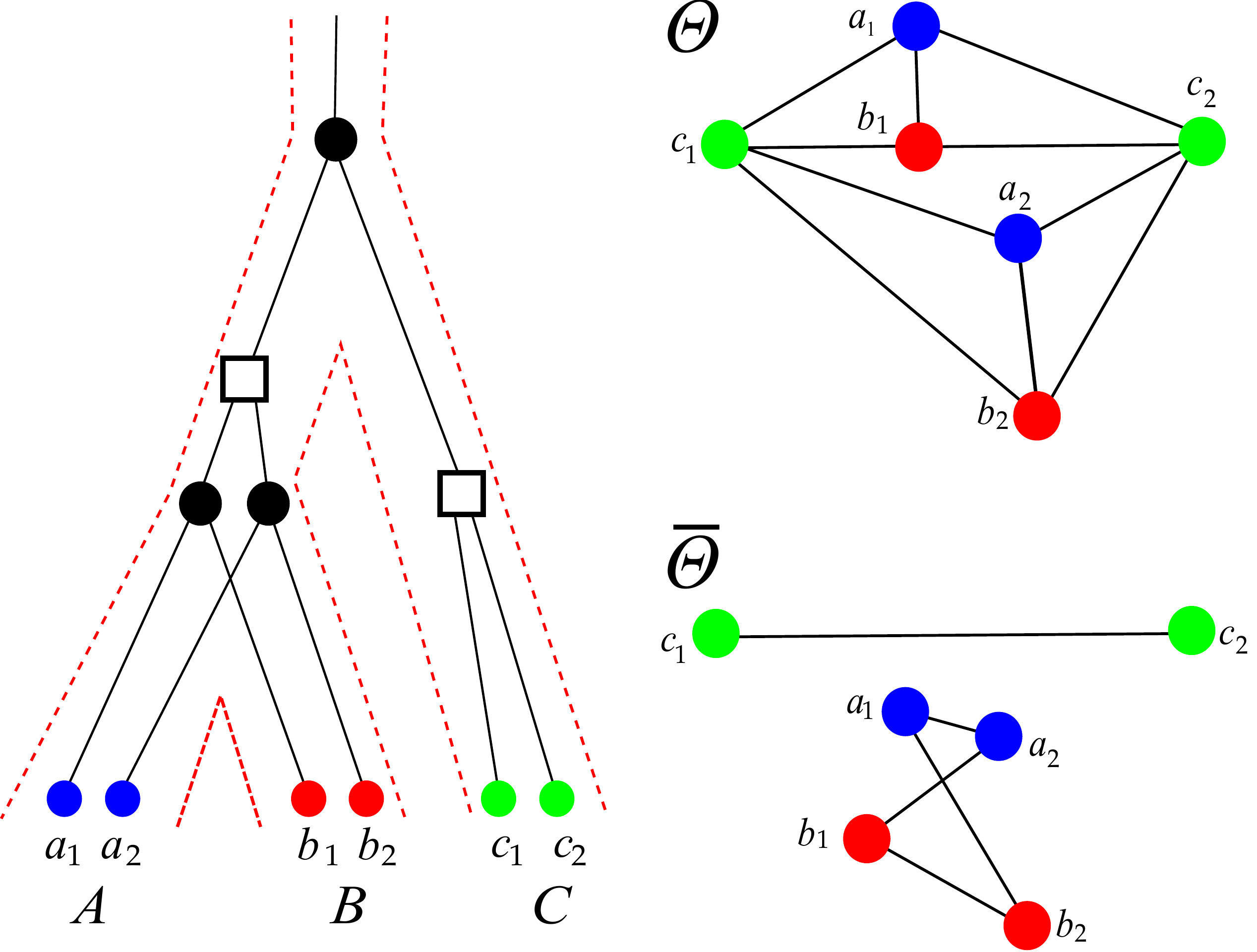}
    \end{minipage} &&
    \begin{minipage}{0.30\textwidth}
      \caption{Orthology and paralogy relations are symmetric but not
        transitive. In this evolutionary scenario with two speciations
        ($\SPEC$) and two duplications ($\DUPL$), the genes $a_1$ and $b_2$
        are both orthologs of $c_1$ but not of each other.  The leaves of
        the gene tree on the l.h.s.\ are colored corresponding to the three
        species $A$, $B$, and $C$.  The orthology graph $\Theta$ and its
        complement, the paralogy graph $\overline{\Theta}$, are shown on
        the r.h.s.}
      \label{fig:ortho}
    \end{minipage}
  \end{tabular}
\end{figure}
  
Based on the work of \citet{Boecker:98} it has been shown by
\citet{Hellmuth:13a} that valid orthology relations are exactly the
cographs:
\begin{proposition}
  An irreflexive, symmetric relation $\Theta$ on $L$ is an orthology
  relation if and only if it is a cograph. In this case, every cotree $T$
  of $\Theta$ with an event labeling $t$ assigning $\SPEC$ to join
  operations and $\DUPL$ to disjoint union operations satisfies
  $\Theta=\Theta(T,t)$. 
\label{prop:ortho-cograph}
\end{proposition}
There is a unique discriminating cotree $(T_{\Theta},t_{\Theta})$ for an
orthology relation $\Theta$, which is obtained from every other
(non-discriminating) cotree $(T,t)$ for $\Theta$ by contracting the inner
edges $uv$ of $T$ if and only if $t(u)=t(v)$
\cite{Boecker:98,Hellmuth:13a}.

It is natural then to ask under which conditions a given orthology relation
$\Theta$ is consistent with a leaf-labeled tree $(T,\sigma)$ in the sense
that there is a reconcilation map $\mu$ from $(T,\sigma)$ to some species
tree such that $\Theta=\Theta(T,t_{\mu})$.  We first consider the special
case $T=T_{\Theta}$.  As shown by \citet{HW:16book}, it is possible to
obtain the set of informative triples
$\mathcal{S}(T_{\Theta},t_{\Theta},\sigma)$ directly from $\Theta$ using
the following rule:\\
$\sigma(a)\sigma(b)|\sigma(c)\in \mathcal{S}(T,t,\sigma)$ if and only if
$\sigma(a),\sigma(b)$, and $\sigma(c)$ are pairwise different species and
either
\begin{enumerate}
\item[(a)] $(a,c),(b,c)\in \Theta$ and $(a,b)\not\in \Theta$ \emph{or}
\item[(b)] $(a,c),(b,c),(a,b)\in \Theta$ and there is a vertex
  $d\neq a,b,c$ with $(c,d)\in \Theta$ and $(a,d),(b,d)\notin \Theta$.
\end{enumerate}

\begin{theorem}
  Let $\Theta$ be a cograph with vertex set $L$ and associated cotree
  $(T_{\Theta},t_{\Theta})$ with leaf set $L$ and let $\sigma$ be a leaf
  coloring. Then there exists a reconciliation map $\mu$ from
  $(T_{\Theta},t_{\Theta},\sigma)$ to some species tree $S$ if and only if
  (i) $\mathcal{S}(T_{\Theta},t_{\Theta},\sigma)$ is compatible and (ii)
  the cograph $(\Theta,\sigma)$ is properly colored, i.e., for all
  $xy \in E(\Theta)$ we have $\sigma(x)\neq \sigma(y)$.
\label{thm:well-formed-cograph}
\end{theorem}
\begin{proof}
  By Proposition \ref{prop:inftriple}, it is necessary and sufficient that
  (i) the set of informative triples is compatible and (ii) the event map
  $t_{\Theta}$ is well-formed. Since $t_{\Theta}$ is the event labeling of
  the co-tree, Condition (ii) amounts to requiring that the leaf set
  $L(T(v_i))$ have pairwise disjoint sets of colors $\sigma(L(T(v_i)))$ for
  all children $v_i\in\child(u)$ of every join node $u$. Since the join
  $\Theta_i \join \Theta_j$ of the two cographs associated with $T(v_i)$
  and $T(v_j)$ introduces an edge $xy$ for all $x\in L(T(v_i))$ and all
  $y\in L(T(v_j))$, the resulting graph can only be properly colored if
  $\sigma(L(T(v_i)))\cap \sigma(L(T(v_j)))=\emptyset$. On the other hand,
  every edge in $\Theta$ is the result of a join operation, thus
  $(\Theta,\sigma)$ can only be well-colored if joins only appear between
  induced subgraphs with disjoint color sets. Thus $t_{\Theta}$ is
  well-formed if and only if $\sigma$ is a proper vertex coloring for
  $\Theta$.
\end{proof}

Under the assumption that a reconciliation map $\mu$ exists for
$(T,\sigma)$ to some species tree, the next results shows that the
orthology relation $\Theta(T,t_{\mu})$ is always a subgraph of the
orthology relation $\Theta(T,\tT)$ implied by $(T,\sigma)$ and its extremal
labeling $\tT$.
\begin{lemma}
  Let $(T,\sigma)$ be a leaf-labeled tree and $\mu$ a reconciliation map
  from $(T,\sigma)$ to some species tree $S$. Then
  $\Theta(T,t_{\mu})\subseteq \Theta(T,\tT)$.
\label{lem:tm-subgraph-tt}
\end{lemma}
\begin{proof}
  Let $u=\lca_T(x,y)$ and suppose $xy\in \Theta(T,t_{\mu})$. Then,
  $t_\mu(u)=\SPEC$ by definition of $\Theta(T,t_{\mu})$, i.e.,
  $\mu(u)\in V^0(S)$.  Therefore, Lemma \ref{lem:disjointSpecies} implies
  $\sigma(L(T(v)))\cap\sigma(L(T(v')))=\emptyset$ for all
  $v,v'\in\child_T(u)$.  Hence, $\tT(u)=\SPEC$ by definition of the
  extremal event labeling and thus $xy\in \Theta(T,\tT)$. \qed
\end{proof}
The converse of Lemma \ref{lem:tm-subgraph-tt} is generally not true, see
Fig.~\ref{fig:counterex_extreme} for an example. For later reference, we
note the following result which is an immediate consequence of Lemma
\ref{lem:tm-subgraph-tt} due to fact that orthology and paralogy relation
are complementary.
\begin{corollary}
  Let $(T,\sigma)$ be a leaf-labeled tree and $\mu$ a reconciliation map
  from $(T,\sigma)$ to some species tree $S$. Then
  $\overline{\Theta}(T,\tT)\subseteq \overline{\Theta}(T,t_\mu)$.
\label{cor:para-eventmap}
\end{corollary}
Lemma \ref{lem:tm-subgraph-tt}, in particular, implies that none of the
labelings $t_{\mu}$ (provided by any reconciliation map $\mu$) can yield
more speciation events in $T$, than the extremal labeling $\tT$. Moreover,
it is easy to see that $t_{\mu}(v) = \SPEC$ always implies 
$\tT(v) = \SPEC$, while $\tT(v) = \DUPL$ implies $t_{\mu}(v) = \DUPL$.

We briefly compare the formalism introduced here with the literatur on
maximum parsimony reconciliations. There, one considers reconciliation maps
$\eta: V(T)\to V(S)$ that map duplication events in $T$ also to vertices of
$S$. The mapping $\eta$ is then interpreted in such a way that the
duplication event $u$ took place along an edge in $S$ that is ancestral to
$\eta(u)$. The map $\eta$ in this setting does not completely determine the
event labeling. The last common ancestor map
\begin{equation}
  \hat\eta(v):= \lca_{S}(\sigma(L(T(v))))\,.
  \label{eq:eta-LCA}
\end{equation}
corresponds to one of the ``most parsimonious reconciliations''
\cite{Gorecki:06,Doyon:09} and can be obtained in polynomial time.  A
closely related reconciliation map can be defined in our setting. The
\emph{LCA-reconciliation map} introduced by \citet{Hellmuth2017} satisfies
the additional axiom
\begin{description}
\item[\AX{(LCA)}] $\mu(u) = v\lca_S(\sigma(L(T(u))))\in E(S)$ for all
  $u\in V(T)$ with $t(u) = \DUPL$, where $v$ denotes the unique parent of
  $\lca_S(\sigma(L(T(u))))\in V(S)$ in $S$.
\end{description}
The Axiom \AX{(LCA)} is the analog of Equ.(\ref{eq:eta-LCA}) for
duplication vertices in $T$, which in our formalism are necessarily mapped
to edges.  For speciation events, the corresponding condition is expressed
by \AX{(R3.i)}. \citet{Hellmuth2017} showed that the existence of a
reconciliation map from $(T,t,\sigma)$ implies also the existence of an
LCA-reconciliation map. Fig.\ \ref{fig:counterex_extreme} shows that an
LCA-reconciliation map does not necessarily have $\tT$ as its event
labeling. Even if $t_{\mu}=\tT$, then $\mu$ is not necessarily an
LCA-reconciliation map, see Fig.\ \ref{fig:counterex_extreme-2}.

\begin{figure}
  \centering
  \includegraphics[width=0.4\textwidth]{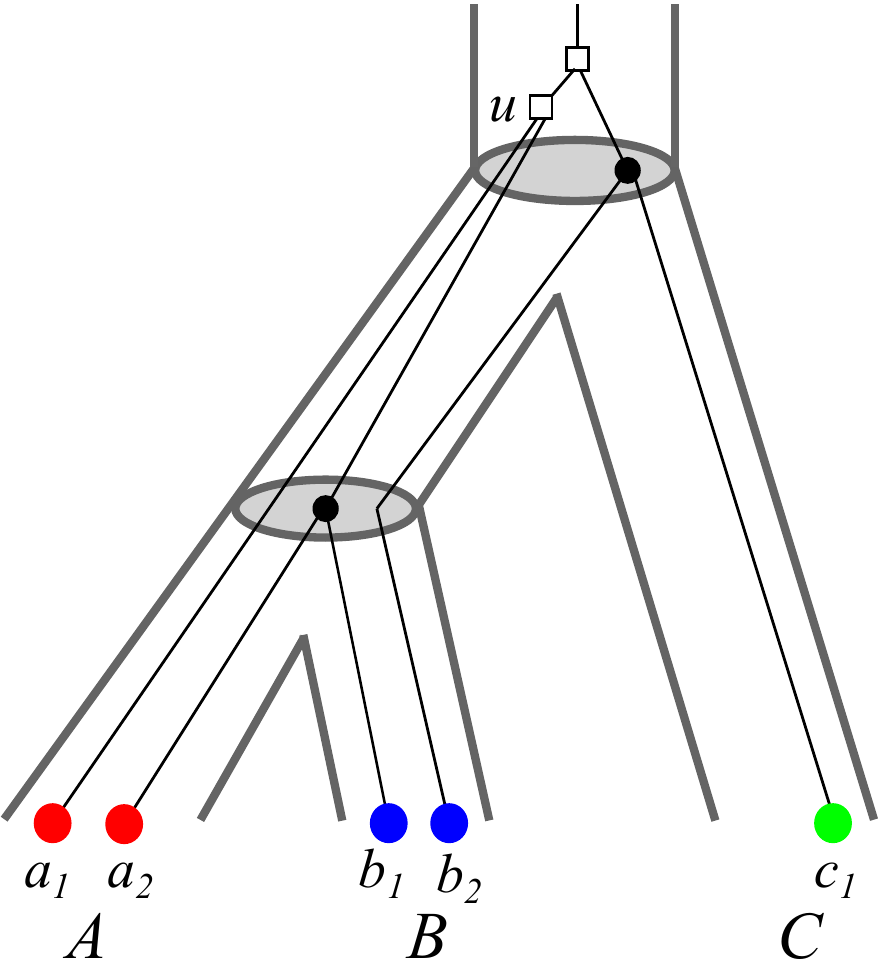}
  \caption{Reconciliation map $\mu$ from $(T,\sigma)$ to the (tube-like)
    species tree $S$.  The map $\mu$ is given implicitly by drawing
    $(T,\sigma)$ into $S$.  The map $\mu$ is not an LCA-reconciliation map
    since $\mu(u)$ does not map $u$ to the edge $v\lca_S(A,B)\in E(S)$
    where $v$ denotes the unique parent of $\lca_S(A,B)$ in $S$. However,
    $t_{\mu}$ and the extremal map $\tT$ coincide. }
  \label{fig:counterex_extreme-2}
\end{figure}

\section{Orthology and Reciprocal Best Matches}\label{sec:OrthologyBMG}

In this section, we further clarify the relationship between the orthology
relation and (reciprocal) best matches. As a main result, we find that the
reciprocal best match graph contains any possible orthology relation.

\begin{lemma}
  If $(T,\sigma)$ with leaf set $L$ explains the RBMG $(G,\sigma)$ and
  $\tT$ is the extremal event labeling of $(T,\sigma)$, then
  $\Theta(T,\tT)$ is a subgraph of the RBMG $G(T,\sigma)$.
  \label{lem:tt-subgraph-rbmg}
\end{lemma}
\begin{proof}
  Consider a vertex $u\in V^0(T)$ with $\child(u) = \{u_1,\dots, u_k\}$.
  If $\tT(u) = \DUPL$, then none of the edges $xy$ in $G$ with
  $x\in L(T(u_i))$ and $y\in L(T(u_j))$, $1\leq i<j\leq k$ is contained
  in $\Theta(T,\tT)$.\\
  Now suppose $\tT(u) = \SPEC$. For $x\in L(T(u_i))$ and $y\in L(T(u_j))$
  with $1\leq i<j\leq k$, we have $xy\in \Theta(T,\tT)$ and, by
  construction of $\tT$, $\sigma(x)\neq \sigma(y)$.  In particular,
  $\tT(u) = \SPEC$ implies that all distinct children
  $u_i,u_j\in \child(u)$ satisfy
  $\sigma(L(T(u_i)))\cap \sigma(L(T(u_j))) =\emptyset$.  Thus,
  $\lca_T(x,y) = u\preceq_T \lca_T(x',y)$ for all $x'\neq x$ with
  $\sigma(x') = \sigma(x)$ and $\lca_T(x,y) = u\preceq_T\lca_T(x,y')$ for
  all $y'\neq y$ with $\sigma(y') = \sigma(y)$, i.e., $x$ and $y$ are
  reciprocal best matches. Hence, $xy\in E(G)$ and thus
  $\Theta(T,\tT)\subseteq G(T,\sigma)$.  \qed
\end{proof}

Lemma \ref{lem:tm-subgraph-tt} and \ref{lem:tt-subgraph-rbmg} immediately imply 
\begin{theorem}
  Let $T$ and $S$ be planted trees, $\sigma: L(T)\to L(S)$ a surjective
  map, and $\mu$ a reconciliation map from $(T,\sigma)$ to $S$. If
  $xy\in \Theta(T,t_{\mu})$, then $x$ and $y$ are reciprocal best matches
  in $(T,\sigma)$.
	\label{thm:orthoIMPLYrbm}
\end{theorem}

\begin{fact}
Reciprocal best matches therefore cannot produce false negative orthology
assignments as long as the evolution of a gene family proceeds via
duplications, losses, and speciations only. 
\label{fact:false-pos}
\end{fact} 

The ``false positive'' edges in the RBMG compared to the orthology relation
are the consequence of a particular class of duplication events:
\begin{theorem}
  Let $(T,t,\sigma)$ be a leaf- and event-labeled gene tree, $G(T,\sigma)$
  and $\Theta(T,t)$ its corresponding RBMG and orthology relation,
  respectively. Moreover, let $a,b\in L(T)$, $v\coloneqq \lca_T(a,b)$, and
  $v_a,v_b\in\child_T(v)$ such that $a\preceq v_a\prec v$,
  $b\preceq v_b\prec v$. Then,
  $ab\in E(G(T,\sigma))\setminus E(\Theta(T,t))$ if and only if
  $t(v)=\DUPL$ and
  $\sigma(a),\sigma(b)\in\sigma(L(T(v_a)))\symdiff \sigma(L(T(v_b)))$,
  where ``$\symdiff$'' denotes the usual symmetric set difference.
  \label{thm:loss-edges}
\end{theorem}
\begin{proof}
  Suppose first $ab\in E(G(T,\sigma))\setminus E(\Theta(T,t))$. By
  definition of $\Theta(T,t)$, we immediately find $t(v)=\DUPL$. Since
  $ab\in E(G(T,\sigma))$, i.e., $a$ and $b$ are reciprocal best matches, it
  must hold $v\preceq_T\lca_T(a,b')$ for any $b'$ of color
  $\sigma(b)$. Hence, $\sigma(b)\notin\sigma(L(T(v_a)))$. Analogously, we
  conclude $\sigma(a)\notin\sigma(L(T(v_b)))$ and thus,
  $\sigma(a),\sigma(b)\in\sigma(L(T(v_a)))\symdiff \sigma(L(T(v_b)))$.
 
  Conversely, assume $t(v)=\DUPL$ and
  $\sigma(a),\sigma(b)\in\sigma(L(T(v_b)))\symdiff
  \sigma(L(T(v_a)))$. Since $t(v)=\DUPL$, $a$ and $b$ cannot be orthologs,
  i.e., $ab\notin E(\Theta(T,t))$. Moreover,
  $\sigma(a)\in\sigma(L(T(v_b)))\symdiff \sigma(L(T(v_a)))$ in particular
  implies $\sigma(a)\notin \sigma(L(T(v_b)))$ and therefore,
  $v\preceq_T \lca_T(a,b')$ for any $b'$ with
  $\sigma(b')=\sigma(b)$. Hence, $b$ is a best match for $a$ in species
  $\sigma(b)$. One similarly concludes that $a$ is a best match for
  $b$. Hence, $a$ and $b$ are reciprocal best matches, which concludes the
  proof.  \qed
\end{proof}

In practical application we usually do not know the event-labeled gene
tree. It is possible, however, to compute the reciprocal best matches
directly from sequence data. Therefore, it is of interest to investigate
the relationship of reciprocal best match graphs and orthology
relations.

\begin{definition} \cite{Geiss:19x} A tree $(T,\sigma)$ is \emph{least
    resolved} (w.r.t.\ the RBMG $G(T,\sigma)$ that it explains) if the
  contraction of any inner edge $e\in E(T)$ implies
  $G(T_e,\sigma)\ne G(T,\sigma)$.
\end{definition}
Since $G(T,\sigma)$ is completely determined by $(T,\sigma)$ we can drop
the reference to $G(T,\sigma)$ and often simply speak about a ``least
  resolved tree''.

\begin{lemma}
  Let $(G,\sigma)$ be an RBMG that is explained by $(T,\sigma)$.  If
  $(T,\sigma)$ is least resolved w.r.t.\ $(G,\sigma)$, then every inner
  edge $e=uv\in E(T)$ satisfies
  $\sigma(L(T(v)))\cap \sigma(L(T(u))\setminus L(T(v))) \neq \emptyset$.
\label{lem:edge-leastRes}
\end{lemma}
\begin{proof}
  For contraposition, assume that there is an inner edge $e=uv\in E(T)$
  with $\sigma(L(T(v)))\cap \sigma(L(T(u))\setminus L(T(v))) =\emptyset$.
  Hence, for all $x\in L(T(v))$ and $y\in L(T(u))\setminus L(T(v))$ we have
  $\lca_T(x,y) = u$ and $\sigma(x)=X\neq \sigma(y)=Y$.  It is easy to see
  that all such $x$ and $y$ form a reciprocal best match and thus,
  $xy\in E(G)$.  Clearly, $x$ and $y$ form also reciprocal best match in
  $(T_e,\sigma)$ and thus, each edge $xy\in E(G)$ with $x\in L(T(v))$ and
  $y\in L(T(u))\setminus L(T(v))$ is contained in $G(T_e,\sigma)$.  Since
  we have not changed the relative ordering of the $\lca's$ of the
  remaining vertices, all edges in $E(G)$ are contained in $G(T_e,\sigma)$.
  \qed
\end{proof}

The converse of  Lemma~\ref{lem:edge-leastRes} is not necessarily true.  As
an example, consider an inner edge $e=uv\in E(T)$ with
$\sigma(L(T(u))) = \sigma(L(T(v))) =\{c\}$.  It is easy to see that $e$ can
be contracted.

Lemma~\ref{lem:edge-leastRes} implies that if $(T,\sigma)$ is least
resolved w.r.t.\ $G(T,\sigma)$ and $u\in V^0(T)$ such that $u$ is incident
to some other inner vertex $v\in \child(u)$, then there is a child
$v'\neq v$ of $u$ which satisfies
$\sigma(L(T(v')))\cap \sigma(L(T(v)))\neq \emptyset$.  By construction of
$\tT$ we have $\tT(u) = \DUPL$.  The latter observation also implies the
following:
\begin{corollary}
  Suppose that $(T,\sigma)$ is least resolved w.r.t.\ $G(T,\sigma)$ and let
  $\tT$ be the extremal event labeling for $(T,\sigma)$. Then
  $\tT(u)= \SPEC$ if and only if all children of $u$ are leaves that are
  from pairwise distinct species.
\label{cor:SPEC}
\end{corollary}

\begin{lemma}\label{lem:reconc-extremal}
  Let $(T,\sigma)$ be some least resolved tree (w.r.t.\ some RBMG) with
  extremal event map $\tT$ and let $S(W,F)$ be a species tree with
  $L(S)=\sigma(L(T))$. Then there is a reconciliation map
    $\mu: V(T)\to V(S)\cup E(S)$ such that $t_{\mu}=\tT$.
\end{lemma}
\begin{proof}
  By Cor.\ \ref{cor:SPEC}, every inner vertex $u$ with $\tT(u)=\SPEC$ is
  only incident to leaves from pairwise distinct species.  However, this
  implies that the set of informative species triples
  $\mathcal{S}(T,\tT,\sigma)$ is empty, and thus, compatible. Hence,
  Proposition \ref{prop:inftriple} implies that there is a reconciliation
  map $\mu$ from $(T,\tT,\sigma)$ to any species tree $S$, defined by
  $\mu(0_T)=0_S$, $\mu(v) = 0_S\rho_S$ for every inner vertex $v\in V^0(T)$
  that is incident to another inner vertex in $T$, and
  $\mu(v) = x =\lca_S(\sigma(L(T(v))))$ for any inner vertex $v$ that is
  only incident to leaves that are from pairwise distinct species, and
  $\mu(v) = \sigma(v)$ for all leaves of $T$.  By construction of $\mu$, we
  have $\tT(u) = t_{\mu}(u)$ with $t_{\mu}(u)$ specified by Def.\
  \ref{def:event-rbmg} for all $u\in V(T)$. \qed
\end{proof}

\begin{corollary} 
  Let $(T,\sigma)$ be a least resolved tree explaining a co-RBMG
  $(G,\sigma)$. Then $(\Theta(T,\tT),\sigma)$ is a disjoint union of
  cliques.
  \label{cor:cliques}
\end{corollary}
\begin{proof}
  By Cor.\ \ref{cor:SPEC} all children of a speciation node $u$ w.r.t.\
  $\tT$ are leaves from pairwise distinct species. Thus the leaves
  $L(T(u))$ form a complete subgraph in $(\Theta(T,\tT),\sigma)$.  On the
  other hand, no ancestor of $u$ is a speciation, i.e., there is no edge
  $ab$ with $a\in L(T(u))$ and $b\notin L(T(u))$. Thus
  $(\Theta(T,\tT),\sigma)$ is a disjoint union of the cliques formed by the
  $L(T(u))$ with $\tT(u)=\SPEC$ possibly together with isolated vertices
  that are not children of any speciation node in $(T,\tT)$.
\end{proof}

Suppose that we know the orthology relation $\Theta(T,\tT)$ that is
obtained from a least resolved tree $(T,\sigma)$ that explains the RBMG
$(G,\sigma)$.  Lemma \ref{lem:reconc-extremal} implies that there is always
a reconciliation map $\mu$ from $(T,\sigma)$ to any species tree $S$ with
$L(S)=\sigma(L(T))$ such that $\tT$ is determined by $\mu$ as in Def.\
\ref{def:event-rbmg}.  Now we can apply Theorem \ref{thm:orthoIMPLYrbm} to
conclude that all orthologous pairs in $\Theta(T,\tT)$ are reciprocal best
matches.  In other words, all complete subgraphs of $\Theta(T,\tT)$ are
also induced subgraphs of the underlying RBMG $(G,\sigma)$.  Hence,
$\Theta(T,\tT)$ is obtained from $(G,\sigma)$ by removing edges such that
the resulting graph is the disjoint union of cliques, see the top-right
tree in Fig.\ \ref{fig:hc-triples} for an example.  However, Fig.\
\ref{fig:hc-triples} also shows that many edges have to be removed to
obtain $\Theta(T,\tT)$. 

This observation establishes the precise relationship of orthology
detection and clustering, since (graph) clustering can be interpreted as
the graph editing problem for disjoint unions of complete graphs
\cite{Boecker:11}. In many orthology prediction tools, such as e.g.\
\texttt{OMA} \cite{Roth:08}, orthologs are summarized as \emph{clusters of
  orthologous groups (COGs)} \cite{Tatusov1997} that are obtained from
reciprocal best matches. 

The results above show that the RBMGs contain the orthology
relation. Equivalently, RBMGs imply constraints on the event labeling. We
also observe that the RBMGs cannot provide conclusive evidence regarding
edges that \emph{must} correspond to orthologous pairs. In the following
sections we consider the constraints implied by the detailed structure of
RBMGs or BMGs in more detail.

\section{Classification of RBMGs}
\label{sec:classi-rbmg}

The structure of RBMGs has been studied in extensive detail by
\citet{Geiss:19x}. Although we do not have an algorithmically useful
complete characterization of RBMGs, there are partial results that can be
used to identify different subclasses of RBMGs based on the structure of
the connected components of the 3-colored subgraphs \cite[Thm.\
7]{Geiss:19x}. Let $\mathscr{C}(G,\sigma)$ be the set of the connected
components of the induced subgraphs on three colors of an RBMG
$(G,\sigma)$. Then every $(C,\sigma)\in\mathscr{C}(G,\sigma)$ is precisely
of one of the three types \cite[Thm.\ 5]{Geiss:19x}:
\begin{description}
\item[\textbf{Type (A)}] $(C,\sigma)$ contains a $K_3$ on three colors
  but no induced $P_4$.
\item[\textbf{Type (B)}] $(C,\sigma)$ contains an induced $P_4$ on three
  colors whose endpoints have the same color, but no induced 
  cycle $C_n$ on $n\geq 5$ vertices.
\item[\textbf{Type (C)}] $(C,\sigma)$ contains an induced cycle $C_6$,
  called \emph{hexagon}, such
  that any three consecutive vertices have pairwise distinct colors.
\end{description}
The graphs for which all $(C,\sigma)\in\mathscr{C}(G,\sigma)$ are of Type
(A) are exactly the RBMGs that are cographs, or co-RBMGs for short
\cite[Thm.\ 8 and Remark 2]{Geiss:19x}. Intuitively, these have a close
connection to orthology graphs because orthology graphs are cographs.

Connected components of Type (B) and Type (C), on the other hand, contain
induced $P_4$s and thus are neither cographs nor connected components of
cographs.  Obs.\ \ref{fact:false-pos} implies that RBMGs that contain
connected components of Type (B) and Type (C) introduce false positive
edges into estimates of the orthology relation. In Section
\ref{sect:goodbadugly} below we will address the question to what extent
and how such false-positives edges can be identified. We distinguish here
co-RBMGs, \emph{(B)-RBMGs}, and \emph{(C)-RBMGs} depending on whether
$\mathscr{C}(G,\sigma)$ contains only Type (A) components, at least one
Type (B) but not Type (C) component, or at least one Type (C) component.

Co-RBMGs have a convenient structure that can be readily understood in
terms of \emph{hierarchically colored cographs} (\hc-cographs) introduced
by \citet[Section 7]{Geiss:19x}.
\begin{definition}
  An undirected colored graph $(G,\sigma)$ is a \emph{hierarchically
    colored cograph (\hc-cograph)} if 
 \begin{description}
 \item[\AX{(K1)}] $(G,\sigma)=(K_1,\sigma)$, i.e., a colored vertex, or
 \item[\AX{(K2)}] $(G,\sigma)= (H_1,\sigma_{H_1}) \join (H_2,\sigma_{H_2})$
   and $\sigma(V(H_1))\cap \sigma(V(H_2))=\emptyset$, or
 \item[\AX{(K3)}]
   $(G,\sigma)= (H_1,\sigma_{H_1}) \cupdot (H_2,\sigma_{H_2})$ and
   $\sigma(V(H_1))\cap \sigma(V(H_2)) \in
   \{\sigma(V(H_1)),\sigma(V(H_2))\}$,
\end{description}
where both $(H_1,\sigma_{H_1})$ and $(H_2,\sigma_{H_2})$ are \hc-cographs
and $\sigma(x)=\sigma_{H_i}(x)$ for any $x\in V(H_i)$ for $i\in\{1,2\}$.
\label{def:hc-cograph}
\end{definition}
Not all properly colored cographs are \hc-cographs, see e.g.\
\citet{Geiss:19x} for counterexamples. However, for each  cograph $G$, there
exists a coloring $\sigma$ (with a sufficient number of colors) such that
$(G,\sigma)$ is an \hc-cograph. 
\begin{proposition}[Thm.\ 9 in \cite{Geiss:19x}]
A graph $(G,\sigma)$ a co-RBMG if and only if it is an \hc-cograph.
\label{prop:co-hc-equi}
\end{proposition}
Since orthology relations are necessarily cographs we can interpret
Proposition \ref{prop:co-hc-equi} as necessary condition for an RBMG to
correctly represent orthology.

The recursive construction of $(G,\sigma)$ in Def.\ \ref{def:hc-cograph}
also defines a corresponding \hc-cotree $(T^G_{\hc},t_{\hc},\sigma)$ whose
leaves are the vertices of $(G,\sigma)$, i.e., the $(K_1,\sigma)$ appearing
in \AX{(K1)}.  Each internal node $u$ of $T^G_{\hc}$ corresponds to either
a join \AX{(K2)} or a disjoint union \AX{(K3)} and is labeled by
$t_{\hc}:V(T^G_{\hc})\setminus L\to \{\DUPL,\SPEC\}$ such that
$t_{\hc}(u)=\SPEC$ if $u$ represents a join, and $t_{\hc}(u)=\DUPL$ if $u$
corresponds to a disjoint union.  Each inner vertex $u$ of $T^G_{\hc}$
represents the induced subgraph $(G,\sigma)[L(T^G_{\hc}(u))]$.
\begin{proposition}[Thm.\ 10 in \cite{Geiss:19x}]
  Every co-RBMG $(G,\sigma)$ is explained by its \hc-cotree
  $(T^G_{\hc},t_{\hc},\sigma)$.
\label{prop:co-hc-tree}
\end{proposition}

Now let $(T^G_{\hc},t_{\hc},\sigma)$ be the \hc-cotree of a co-RBMG
$(G,\sigma)$. Note, the structure of $T^G_{\hc}$ is solely determined by
the \hc-cograph structure of $(G,\sigma)$.  Somehwat surprisingly, the
mathematical structure of the \hc-cotree $(T^G_{\hc},t_{\hc},\sigma)$ and,
in particular, its coloring $t_{\hc}$ has a simple biological
interpretation.  Consider $\{v',v''\}=\child(u)$. If $t_{\hc}(u)=\SPEC$ in
the \hc-cotree, then
$\sigma(L(T^G_{\hc}(v')))\cap\sigma(L(T^G_{\hc}(v'')))=\emptyset$ in
agreement with Lemma~\ref{lem:disjointSpecies}. On the other hand, if
$t_{\hc}(u)=\DUPL$, then \AX{(K3)} implies
$\sigma(L(T^G_{\hc}(v')))\cap\sigma(L(T^G_{\hc}(v'')))\ne\emptyset$, in
which case $u$ indeed must be a duplication from the biological point of
view (contraposition of Lemma~\ref{lem:disjointSpecies}).

The \hc-cotree $(T^G_{\hc},t_{\hc},\sigma)$ of $(G,\sigma)$ will in general
not be discriminating and it is not necessarily possible to reduce
$(T^G_{\hc},t_\hc,\sigma)$ to a discriminating \hc-cotree
$(\hat T^G_{\hc},\hat t,\sigma)$ that still explains $(G,\sigma)$.
Although it is always possible to contract edges $uv$ of
$(T^G_{\hc},t_{\hc},\sigma)$ with $t_{\hc}(u) = t_{\hc}(u) =\SPEC$ (cf.\
\cite[Cor.\ 11]{Geiss:19x}), there are examples where edges $uv$ with
$t_{\hc}(u) = t_{\hc}(u) =\DUPL$ cannot be contracted to obtain a tree that
still explains $(G,\sigma)$ (cf.\ \cite[Fig.\ 15]{Geiss:19x}). We refer to
\cite{Geiss:19x} for more details and a characterization of edges that are
contractable. It is of interest, therefore, to ask whether there are true
orthology relations $\Theta$ that are not \hc-cographs, or equivalently,
when does a discriminating \hc-cotree $(\hat T,\hat t,\sigma)$ that is
obtained by edge-contraction from a given \hc-cotree
$(T^G_{\hc},t_{\hc},\sigma)$ still explains an RBMG $(G,\sigma)$? To answer
this question we provide first
\begin{definition}
  A tree $(T,t,\sigma)$ contains \emph{no losses}, if for all $x\in V(T)$
  with $t(x)=\DUPL$ we have $\sigma(L(T(v'))) = \sigma(L(T(v'')))$ for all
  $ v', v''\in\child(x)$.
\label{def:loss}
\end{definition}

\begin{theorem}\label{thm:cobmg-lossfree}
  Let $(T,\sigma)$ be a leaf-labeled tree such that there is a
    reconciliation map $\mu$ to some species tree and assume that
    $(T,t_{\mu},\sigma)$ does not contain losses. Then
  \begin{enumerate}
  \item The RBMG $G(T,\sigma)$ explained by $(T,\sigma)$ equals the
      colored cograph $(\Theta(T,t_{\mu}),\sigma)$.
  \item The unique disciminating cotree $(\hat T, \hat t, \sigma)$
      of $(\Theta(T,t_{\mu}),\sigma)$ explains the RBMG $(G,\sigma)$.
  \end{enumerate}
\end{theorem} 
\begin{proof}
  To simplify the notation, we set $(G,\sigma) = G(T,\sigma)$ and
  $(H,\sigma)=(\Theta(T,t_{\mu}),\sigma)$.
    
  We start with proving Statement (1).  By Theorem \ref{thm:orthoIMPLYrbm},
  $(H,\sigma)$ is a subgraph of $(G,\sigma)$ and $V(H)=V(G)$, hence it
  suffices to show that every edge $ab\in E(G)$ is also contained in
  $E(H)$. Assume, for contradiction, that this is not the case, i.e.,
  $ab \notin E(H)$, and thus $t_{\mu}(x)=\DUPL$ for
  $x\coloneqq \lca_T(a,b)$. Since $(T,t,\sigma)$ has no losses, we have
  $\sigma(L(T(v'))) = \sigma(L(T(v'')))$ for all $v', v''\in\child(x)$, and
  thus $a\in L(T(v'))$ and $b\in L(T(v''))$ for some pair of distinct
  children $v',v''\in\child(x)$ of $x$. From
  $\sigma(L(T(v'))) = \sigma(L(T(v'')))$ we know that there is a vertex
  $a'\in L(T(v''))$ with $\sigma(a')=\sigma(a)$.  Thus,
  $\lca_T(a,b)=x\succ_T \lca_T(a',b)$ for some $a'\in L(T(v''))$, which
  implies that $ab\notin E(G)$; a contradiction.  We conclude that
  $ab\in E(G)$ if and only if $ab\in E(H)$ and thus
  $(G,\sigma) = (H,\sigma)$.

  Let us now turn to Statement (2).  In order to show that
  $(\hat T, \hat t, \sigma)$ explains the RBMG $(G,\sigma)$ we first note
  that, since $(G,\sigma)$ is a cograph by Statement (1), there is a unique
  discriminating cotree $(\hat T, \hat t, \sigma)$ for
  $(G,\sigma)$. Furthermore, $(\hat T, \hat t, \sigma)$ is obtained from
  any cotree $(T,t_{\mu},\sigma)$ for $(G,\sigma)$ by contracting all edges
  $uv$ in $T$ with $t_{\mu}(u)=t_{\mu}(v)$ \cite{Hellmuth:13a}. It remains
  to show that $ab$ is an edge in $(G,\sigma)$ if and only if $ab$ forms a
  reciprocal best match in $(\hat T, \sigma)$.
  \par\noindent
  First consider duplications. Suppose, we have contracted the edge $xv$
  with $t_{\mu}(x)=t_{\mu}(v) = \DUPL$. By assumption, for all children
  $v',v''$ of $v$ we have $\sigma(L(T(v'))) = \sigma(L(T(v'')))$. Moreover,
  since $\sigma(L(T(v)))$ is the union of species $\sigma(L(T(w))))$ of its
  children $w$, we have
  $\sigma(L(T(v))) = \sigma(L(T(v')))=\sigma(L(T(v'')))$. Hence, after
  contraction of $xv$, the vertices $v'$ and $v''$ are now children of $x$
  and still satisfy $\sigma(L(\hat T(v'))) = \sigma(L(\hat T(v'')))$. In
  particular, $\sigma(L(\hat T(v'))) = \sigma(L(\hat T(w)))$ for every
  child $w$ of $x$.  By induction on the number of contracted edges, every
  vertex $x$ in $\hat T$ with $\hat t(x)=\DUPL$ still satisfies
  $\sigma(L(\hat T(v'))) = \sigma(L(\hat T(v'')))$ for all children
  $v',v''$ of $x$ in $\hat T$. Thus, the same argument as in the proof of
  Statement (1) implies that $ab$ cannot be a reciprocal best match in
  $\hat T$ for all $a\in L(T(v'))$ and $b\in L(T(v''))$. We also have
  $\lca_{\hat T}(a,b)=x$ for $a\in L(T(v'))$ and $b\in L(T(v''))$, and thus
  $\hat t(\lca_{\hat T}(a,b)) = \DUPL$.  Since $(\hat T, \hat t, \sigma)$
  is a cotree for the cograph $(G,\sigma)$,
  $\hat t(\lca_{\hat T}(a,b)) = \DUPL$ implies $ab\notin E(G)$. Therefore,
  $ab\notin E(G)$ unless $a$ and $b$ form a reciprocal best match in
  $(\hat T,\sigma)$.
  Let us now turn to speciation vertices. Lemma 47 in \cite{Geiss:19x}
  states, in particular, that all non-discriminating edges $uv$ with
  $t_{\mu}(u)=t_{\mu}(v)=\SPEC$ can be contracted to obtain a tree that
  still explains $(G,\sigma)$.  
  Thus, if $a$ and $b$ are reciprocal best matches in $(\hat T,\sigma)$,
  then $ab\in E(G)$. We conclude, therefore, that $ab\in E(G)$ if and only
  if $a$ and $b$ are reciprocal best matches in $(\hat T,\sigma)$.  \qed
\end{proof}
Prop.\ \ref{prop:co-hc-equi} shows that if the \emph{no loss} condition of
Def.\ \ref{def:loss} holds, then $(\Theta(T,t_{\mu}),\sigma)=G(T,\sigma)$
is a co-RBMG, an \hc-cograph, and an orthology relation.
  
The \emph{no loss} condition of Def.\ \ref{def:loss} is very restrictive,
however, and thus in general is will not be satisfied in real-life
data. Theorem~\ref{thm:well-formed-cograph} shows that orthology relations
correspond to properly colored cographs with compatible sets of the
informative triples. The characterization of co-RBMGs in \cite{Geiss:19x},
on the other hand, shows that only \hc-colorings may appear.  Since the
requirement that $\sigma$ is a proper coloring already implies disjointness
of the color sets for join operations, we can interpret the \hc-coloring
condition as a condition on duplication vertices. The offending vertices
are exactly those for which (i) $t(u)=\DUPL$ and (ii) there are two
children $v',v''\in\child(u)$ such that both
$\sigma(L(T(v')))\setminus \sigma(L(T(v'')))\ne\emptyset$ and
$\sigma(L(T(v'')))\setminus \sigma(L(T(v')))\ne\emptyset$. In this case,
there is a pair of species such that a different ``paralog group'' (that
is, a lineage of genes descending from a duplication) is missing in each of
them.  Every pair of vertices $a\in L(T(v'))$ with
$\sigma(a)\notin\sigma(L(T(v'')))$ and $b\in L(T(v''))$ with
$\sigma(b)\notin\sigma(L(T(v')))$ forms a best match and thus a false
positive orthology assignment. Since an RBMG is a cograph only if it is
hierarchically colored, the presence of such duplications implies that the
RBMG is not a cograph. At least in principle, therefore, it should be
possible to identify the false positive edges by means of a suitable
cograph-editing approach.

Before closing this section, we briefly return to the existence of
reconciliation maps.  Since every hc-cograph is a properly colored cograph,
Theorem \ref{thm:well-formed-cograph} immediately implies
\begin{corollary}
  Let $\Theta$ be an \hc-cograph with vertex set $L$ and associated
  \hc-cotree $(T^{\Theta}_{\hc},t_{\hc},\sigma)$ with leaf set $L$.  Then
  there exists a reconciliation map $\mu$ from
  $(T^{\Theta}_{\hc},t_{\hc},\sigma)$ to some species tree $S$ if and only
  if $\mathcal{S}(T_{\Theta},t_{\Theta},\sigma)$ is compatible.
  \label{cor:well-formed-hc-cograph}
\end{corollary}

\begin{figure}
  \begin{center}
    \includegraphics[width=1\textwidth]{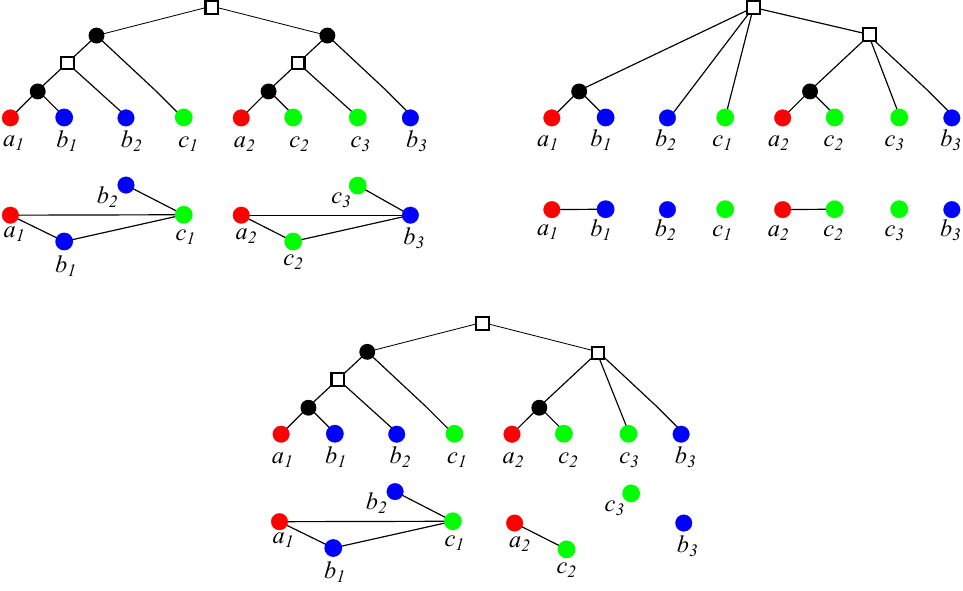}      
  \end{center}
  \caption{\emph{Top Left:} A (discriminating) \hc-cotree
    $(T^G_\hc,t_{\hc},\sigma)$.  Its corresponding \hc-cograph
    $(G,\sigma) = (\Theta(T^G_\hc,t_{\hc}),\sigma)$ is drawn below
    $(T^G_\hc,t_{\hc},\sigma)$. In fact, Prop.\ \ref{prop:co-hc-equi}
    implies that $(G,\sigma)$ is an RBMG. \emph{Top Right:} A tree
    $(T^*,\tT,\sigma)$ that is least resolved w.r.t.\ the RBMG $(G,\sigma)$
    together with extremal labeling $\tT$ and the resulting orthology
    relation $\Theta(T^*,\tT)$, where $(T^*,\tT)$ is not discriminating.
    \emph{Below:} A tree $(T,\tT,\sigma)$ together with extremal labeling
    $\tT$ that explains the RBMG $(G,\sigma)$ but is not least resolved
    w.r.t.\ $(G,\sigma)$.  The resulting orthology relation $\Theta(T,\tT)$
    is drawn below $(T,\tT,\sigma)$. }
  \label{fig:hc-triples}
\end{figure}

By Cor.\ \ref{cor:well-formed-hc-cograph}, it is not necessarily possible
to reconcile a (discriminating) \hc-cotree with any species tree. An
example is shown in Fig.\ \ref{fig:hc-triples}.  To be more precise, the
\hc-cotree $(T^G_{\hc},t_{\hc},\sigma)$ in Fig.\ \ref{fig:hc-triples}
yields the conflicting species triples $AB|C$ and $AC|B$. Hence, Prop.\
\ref{prop:inftriple} implies that $(T^G_{\hc},t_{\hc},\sigma)$ cannot be
reconciled with any species tree even though $(T^G_{\hc},\sigma)$ explains
the RBMG $(G,\sigma)$.  One can contract edges of $(T^G_\hc,\sigma)$ to
obtain a least resolved tree $(T^*,\sigma)$ that still explains
$(G,\sigma)$, see Fig.\ \ref{fig:hc-triples} (top right).  In agreement
with Lemma \ref{lem:reconc-extremal},
$\mathcal{S}(T^*,t_{\mu},\sigma) = \emptyset$ and thus, there is always a
reconciliation map $\mu$ from $(T^*,t_{\mu},\sigma)$ to any species tree
$S$ with $L(S)=\sigma(L(T))$. Moreover, in agreement with Theorem
\ref{thm:orthoIMPLYrbm}, all orthologous pairs in $\Theta(T^*,\tT,\sigma)$
are best matches. Although $(T^*,\sigma)$ explains $(G,\sigma)$, the two
graphs $(G,\sigma) = (\Theta(T^G_\hc,t),\sigma)$ and
$(\Theta(T^*,\tT),\sigma)$ are very different. In particular, by Corollary
\ref{cor:cliques}, $\Theta(T^*,\tT)$ is the disjoint union of cliques.

\begin{fact}
  In general it is not necessary to edit $(G,\sigma)$ to a disjoint union
  of cliques to obtain a valid orthology relation.
\end{fact}
An example is provided by the tree $(T,\tT,\sigma)$ in Fig.\
\ref{fig:hc-triples}.  Obviously, $\Theta(T,\tT)$ is not the disjoint union
of cliques. Moreover, $AB|C$ is the only informative triple displayed by
$(T,\tT,\sigma)$ where $A$, $B$, and $C$ correspond to the red, blue and
green species, respectively.  Prop.\ \ref{prop:inftriple} implies that
$(T,\tT,\sigma)$ can be reconciled with any species tree that displays
$AB|C$.  In other words, $\Theta(T,\tT)$ is already ``biologically
feasible'' and there is no need to remove further edges from
$\Theta(T,\tT)$.

\section{Non-Orthologous Reciprocal Best Matches}
\label{sect:goodbadugly}

In this section we investigate to what extent false positive orthology
assignments in the reciprocal best match graph can be identified.  Since
the orthology relation $\Theta$ must be a cograph, it is natural to
consider the smallest obstructions, i.e., induced $P_4$s in more
detail. First we note that every induced $P_4$ in an RBMG contains either
three or four distinct colors \cite[Sect.\ E]{Geiss:19x}. Each $P_4$ in an
RBMG $(G,\sigma)$ spans an induced subgraph of every BMG $(\G,\sigma)$ that
contains $(G,\sigma)$ as its symmetric part. These these induced subgraphs
of a BMG $(\G,\sigma)$ with four vertices are known as
\emph{quartets}. With respect to a fixed BMG, every induced $P_4$ belongs
to one of three distinct types which are defined in terms of its coloring
and the quartet in which it resides.  An induced $P_4$ with edges $ab$,
$bc$, and $cd$ is denoted by $\langle abcd \rangle$ or, equivalently,
$\langle dcba \rangle$.

\begin{definition}
  Let $(\G,\sigma)$ be a BMG explained by the tree $(T,\sigma)$, with
  symmetric part $(G,\sigma)$ and let
  $Q\coloneqq \{x,x',y,z\} \subseteq L(T)$ with $\sigma(x)=\sigma(x')$
  and pairwise distinct colors $\sigma(x)$, $\sigma(y)$, and $\sigma(z)$.
  The set $Q$, resp., the induced subgraph $(\G_{|Q},\sigma_{|Q})$ is
  \begin{itemize}
  \item a \emph{good quartet} if (i) $\langle xyzx'\rangle$ is an induced
    $P_4$ in $(G,\sigma)$ and (ii) $(x,z),(x',y)\in E(\G)$ and
    $(z,x),(y,x')\notin E(\G)$,
  \item a \emph{bad quartet} if (i) $\langle xyzx'\rangle$ is an induced
    $P_4$ in $(G,\sigma)$ and (ii) $(z,x),(y,x')\in E(\G)$ and
    $(x,z),(x',y)\notin E(\G)$, and 
  \item an \emph{ugly quartet} if $\langle xyx'z\rangle$ is an induced
    $P_4$ in $(G,\sigma)$.
  \end{itemize}
  \label{def:GoodBadUgly}
\end{definition}

If $Q$ is a good, bad, or ugly quartet we will refer to the underlying
induced $P_4$ as a good, bad, or ugly quartet, respectively.  Lemma 32 of
\cite{Geiss:19x} states that every quartet $Q$ in an RBMG $(G,\sigma)$ that
is contained in a BMG $(\G,\sigma)$ is either good, bad, or ugly. An
example of an RBMG containing good, bad, and ugly quartets is shown in
Fig.\ \ref{fig:P_4-classes}. Note that good, bad, and ugly quartets cannot
appear in RBMGs of Type \AX{(A)}. These are cographs and thus by definition
do not contain induced $P_4$s.

\begin{figure}
  \begin{center}
    \includegraphics[width=1\textwidth]{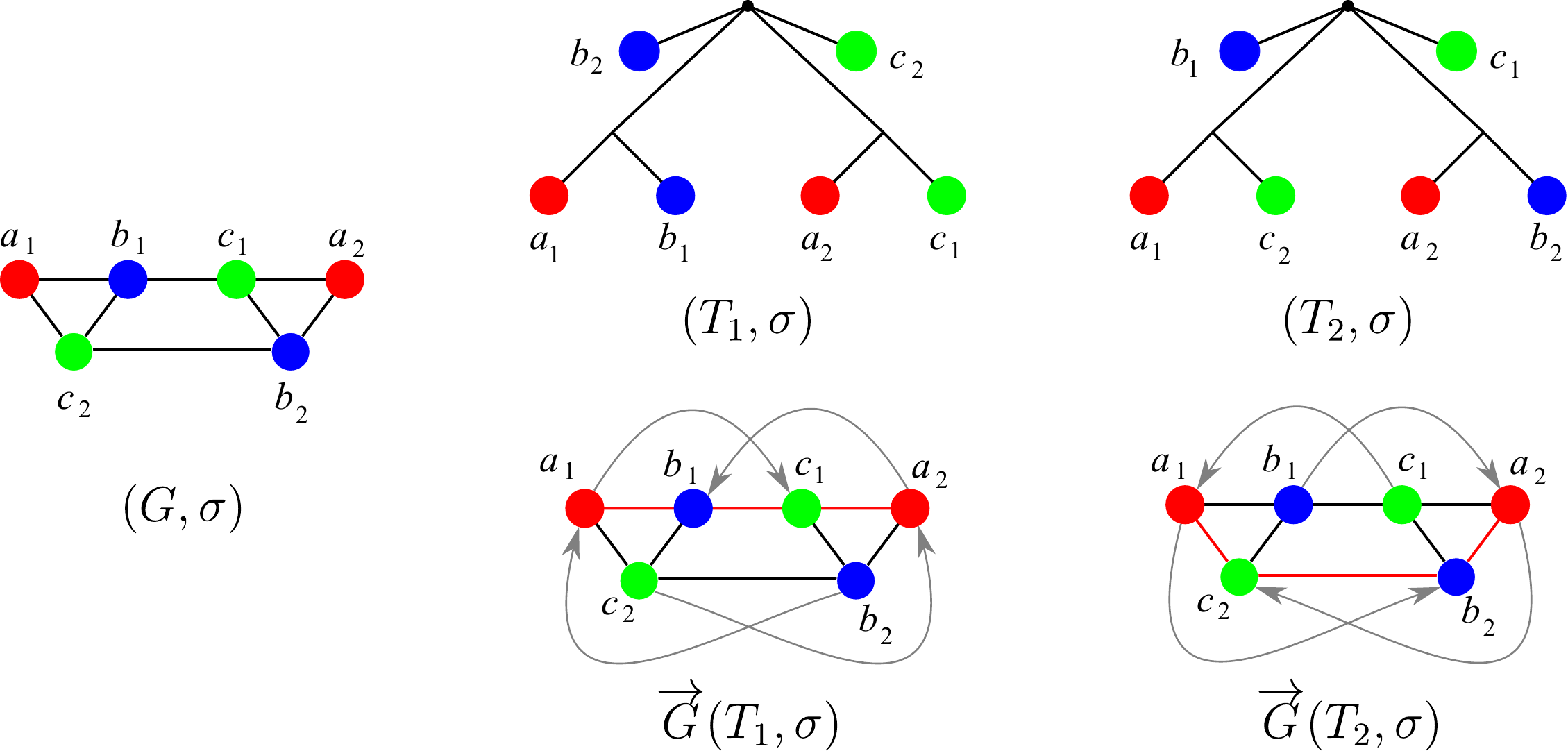}
  \end{center}
  \caption{The 3-RBMG $(G,\sigma)$ is explained by two trees
    $(T_1,\sigma)$ and $(T_2,\sigma)$. These induce distinct BMGs
    $\G(T_1,\sigma)$ and $\G(T_2,\sigma)$. In $\G(T_1,\sigma)$,
    $P^1 =\langle a_1b_1c_1a_2\rangle$ defines a good quartet, while
    $P^2 =\langle a_1c_2b_2a_2\rangle$ induces a bad quartet. In
    $\G(T_2,\sigma)$ the situation is reversed. The good quartets in
    $\G(T_1,\sigma)$ and $\G(T_2,\sigma)$ are indicated by red edges. The
    induced paths $\langle a_1 b_1 c_1 b_2\rangle$ and
    $\langle a_2 c_1 b_1 c_2\rangle$ are examples of ugly quartets.
    \hfill\break
    Figure reused from \cite{Geiss:19x}, \copyright Springer}
  \label{fig:P_4-classes}
\end{figure}

The location of good quartets (in contrast to bad and ugly quartets) turns
out to be strictly constrained. This fact can be used to show that the
``middle'' edge of any good quartet must be a false positive orthology
assignment:
\begin{lemma}
  Let $(T,\sigma)$ be some leaf-labeled tree and $\hat t_T$ the extremal
  event labeling for $(T,\sigma)$.  If $\langle xyzx'\rangle$ is a good
  quartet in the BMG $\G(T,\sigma)$, then $\tT(v)=\DUPL$ for
  $v\coloneqq \lca(x,x',y,z)$.
  \label{lem:dupl-extr}
\end{lemma}
\begin{proof}
  Lemma 36 of \citet{Geiss:19x} implies that for a good quartet
  $\langle xyzx' \rangle$ in $\G(T,\sigma)$ with
  $v\coloneqq \lca(x,x',y,z)$ there are two distinct children
  $v_1, v_2\in \child(v)$ such that $x,y \preceq_T v_1$ and
  $x',z\preceq_T v_2$.  Thus, in particular, $v_1$ and $v_2$ must be inner
  vertices in $(T,\sigma)$. Since $\sigma(x)=\sigma(x')$ by definition of a
  good quartet, we have
  $\sigma(L(T(v_1)))\cap \sigma(L(T(v_2)))\neq \emptyset$. Hence,
  $\tT(v)\neq \SPEC$ by definition of $\tT$ (cf.\ Definition
  \ref{def:event-rbmg2}). \qed
\end{proof}
As an immediate consequence of Lemma \ref{lem:dupl-extr} and Cor.\
\ref{cor:para-eventmap}, an analogous statement is true for event labelings
$t_\mu$ for a given reconciliation map:
\begin{corollary}\label{cor:P4-Editing}
  Let $T$ and $S$ be planted trees, $\sigma: L(T)\to L(S)$ a surjective
  map, and $\mu$ a reconciliation map from $(T,\sigma)$ to $S$.
  If $\langle xyzx'\rangle$ is a good quartet in the BMG $\G(T,\sigma)$,
  then $t_\mu(v)=\DUPL$ for $v\coloneqq \lca(x,x',y,z)$.
\end{corollary}
Given an RBMG $(G,\sigma)$ that contains a good quartet
$\langle xyzx' \rangle$ (w.r.t.\ to the underlying BMG $(\G,\sigma)$), the
edge $yz$ therefore always corresponds to a false positive orthology
assignment, i.e., it is not contained in the true orthology relation
$\Theta$.

Not all false positives can be identified in this way from good quartets,
however. The RBMG $G(T_1,\sigma)$ in Fig.\ \ref{fig:P4-classes_2}, for
instance, contains only one good quartet, that is
$\langle a_1c_2b_2a_2 \rangle$. After removal of the false positive edge
$c_2b_2$, the remaining undirected graph still contains the bad quartet
$\langle a_1b_1c_1a_2 \rangle$, hence, in particular, it still contains an
induced $P_4$ and is, therefore, not an orthology relation.

Neither bad nor ugly quartets can be used to unambiguously identify false
positive edges. For an example, consider Fig.\ \ref{fig:P4-classes_2}. The
two 3-RBMGs $G(T_1,\sigma)$ and $G(T_2,\sigma)$ both contain the bad
quartet $\langle a_1 b_1 c_1 a_2 \rangle$. As a consequence of Lemma
\ref{lem:disjointSpecies}, neither the root of $T_1$ nor the root of $T_2$
can be labeled by a speciation event. Hence, as $a_1,b_1,c_1,a_2$ reside
all in different subtrees below the root of $T_1$, all edges
$a_1b_1, b_1c_1, c_1a_2$ in $G(T_1,\sigma)$ correspond to false positive
orthology assignments.  On the other hand, the vertices $b_1$ and $c_1$
reside within the same 2-colored subtree below the root of $T_2$ and are
incident to the same parent in $T_2$. Therefore, one easily checks that
there exist reconciliation scenarios where $b_1$ and $c_1$ are orthologous,
hence the edge $b_1c_1$ must indeed be contained in the orthology
relation. Similarly, $\langle a_1 b_1c_1b_2 \rangle$ and
$\langle a_1 b_1a_3c_2 \rangle$ are ugly quartets in $G(T_1,\sigma)$ and
$G(T_2,\sigma)$, respectively. By the same argumentation as before, the
edges $a_1b_1$, $b_1c_1$, and $c_1b_2$ are false positives in
$G(T_1,\sigma)$. For $(T_2,\sigma)$, however, there exist reconciliation
scenarios, where $a_3$ and $c_2$ are orthologs.

\begin{figure}
  \begin{tabular}{lcr}
    \begin{minipage}{0.6\textwidth}
      \begin{center}
        \includegraphics[width=\textwidth]{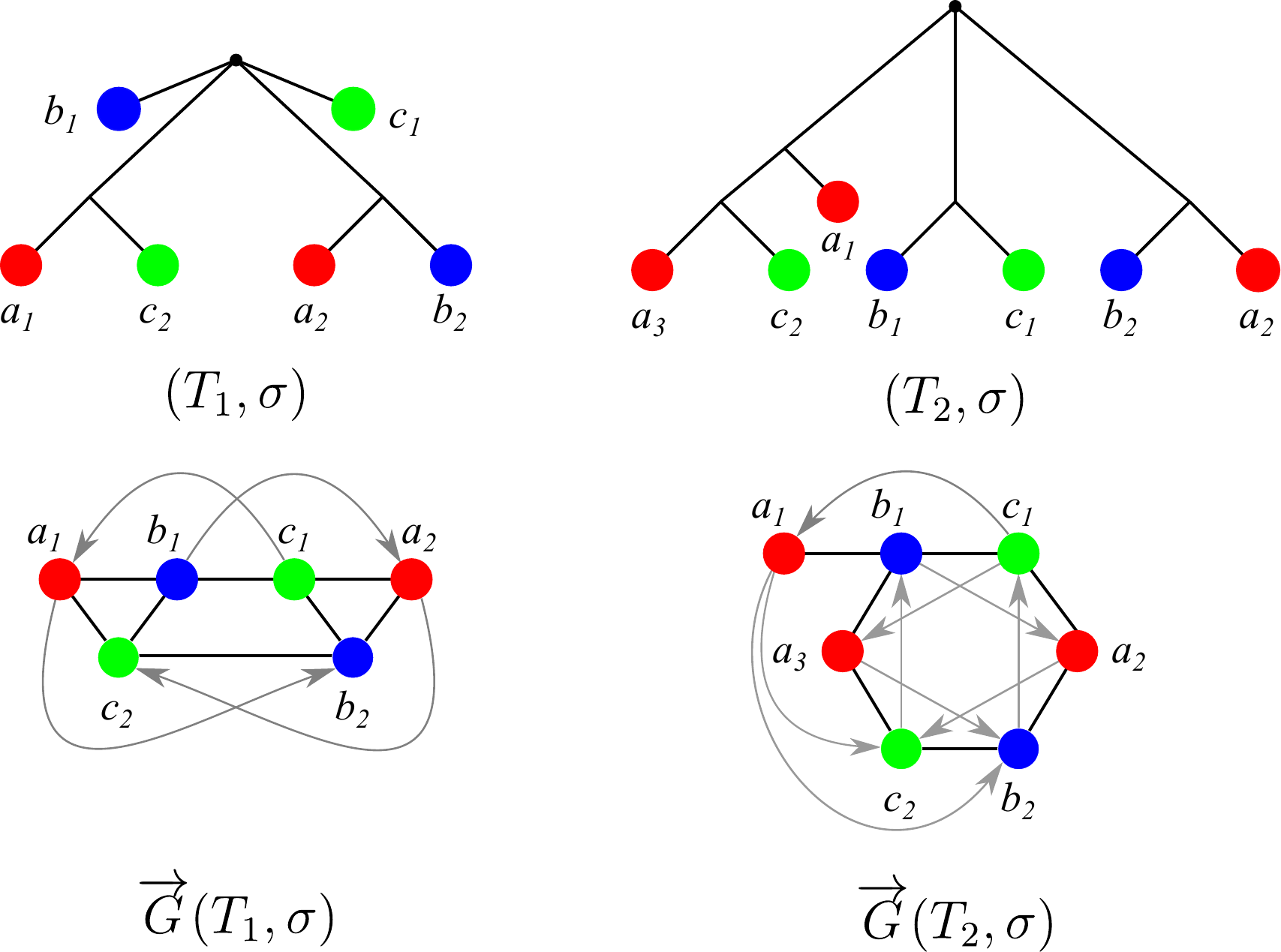}
      \end{center}
    \end{minipage}
    &    &
    \begin{minipage}{0.3\textwidth}
      \caption{Not all false positive orthology assignments can be
        identified using good quartets. Conversely, bad and ugly quartets
        do not unambiguously identifiy false positive edges.  See the text
        below Cor.\ \ref{cor:P4-Editing} for explanation.}
      \label{fig:P4-classes_2}
    \end{minipage}
  \end{tabular}
\end{figure}

Cor.\ 9 of \citet{Geiss:19x}, finally, implies that every (B)-RBMG and
every (C)-RBMG contains at least one good quartet. In particular,
therefore, there is at least one false positive orthology assignment that
can be identified with the help of good quartets. We shall see in
Section~\ref{ssect:simresDL}, using simulated data, that in practice the
overwhelming majority of false positive orthology assignments is already
identified by good quartets.

From a theoretical point of view it is interesting nevertheless that it is
possible to identify even more false positive orthology assignments
starting from Lemma~\ref{lem:disjointSpecies}. It implies that
  $t(\lca(x,y))=\DUPL$ whenever $x$ and $y$ are located in two distinct
  leaf sets defined for the the same connected component of an induced
  3-RBMG of Type (B) or (C).  Details can be found in \cite[Lemma
  25]{Geiss:19x} and the Supplemental Material. At least in our simulation
  data scenarios of this type that are not covered already by a good
  quartet seem to be exceedingly rare, and hence of little practical
  relevance.

\section{Simulations}

Although the edges in the RMBG cannot identify orthologous pairs with
certainty (as a consequence to Lemma~\ref{lem:norestriction}), there is a
close resemblance in practice, i.e., for empirically determined
scenarios. In order to explore this connection in more detail, we consider
simulated evolutionary scenarios $(T,S,\mu)$. These uniquely determine both
the (reciprocal) best match graph $\G(T,\sigma)$ and $G(T,\sigma)$, resp.,
and the orthology graph $\Theta$, thus allowing a direct comparison of
these graphs. Since we only analyze scenarios $(T,S,\mu)$, we did not use
simulations tools such as \texttt{ALF} \cite{ALF:11} that are designed to
simulate sequence data.

\subsection{Simulation Methods}

In order to simulate evolutionary scenarios $(T,S,\mu)$ we employ a
stepwise procedure:
\begin{itemize}
\item[(1)] \textbf{Construction of the species tree $S$.} We regard $S$ as
  an ultrametric tree, i.e., its branch lengths are interpreted as
  real-time.  Given a user-defined number of species $N$ we generate $S$
  under the \emph{innovations model} as described by \citet{Keller:2012}.
  The binary trees generated by this model have similar depth and
  imbalances as those of real phylogenetic trees from databases. 
\item[(2)] \textbf{Construction of the true gene tree $\tilde T$.}
  Traversing the species tree $S$ top-down, one gene tree $\tilde T$ is
  generated with user-defined rates $r_D$ for duplications, $r_L$ for
  losses, and $r_H$ for horizontal transfer events.  The number of events
  along each edge of the species tree, of each type of event, is drawn from
  a Poisson distribution with parameter $\lambda = \ell r_e$, where $\ell$
  is the length of the edge $e$ and $r_e$ is the rate of the event type.
  Duplication and horizontal transfer events duplicate an active lineage
  and occur only inside edges of $S$. For duplications, both offspring
  lineages remain inside the same edge of the species tree as the parental
  gene. In contrast, one of the two offsprings of an HGT event is
  transferred to another, randomly selected, branch of the species tree at
  the same time.  At speciation nodes all branches of the gene tree are
  copied into each offspring. Loss events terminate branches of $\tilde
  T$. Loss events may occur only within edges of the species tree that
  harbor more than one branch of the gene tree. Thus every leaf of $S$ is
  reached by at least one branch of the gene tree $\tilde T$.  All vertices
  $v$ of $\tilde T$ are labeled with their event type $t(v)$, in
  particular, there are different leaf labels for extant genes and lost
  genes.  The simulation explicitly records the reconciliation map, i.e.,
  the assignment of each vertex of $\tilde T$ to a vertex or edge of $S$.
\item[(3)] \textbf{Construction of the observable gene tree $T$ from
    $\tilde T$}.  The leaves of $\tilde T$ are either observable extant
  genes or unobservable losses.  As described by
  \citet{HernandezRosales:12a}, we prune $\tilde T$ in bottom-up order by
  removing all loss events and omitting all inner vertices with only a
  single remaining child.
\end{itemize}

Using steps (1) and (2), we simulated 10,000 scenarios for species trees
with 3 to 100 species (=leaves) and additional 4000 scenarios for species
trees with 3 to 50 leaves, drawn from a uniform distribution.  For each of
these species trees, exactly one gene tree was simulated as described
above.  The rate parameters were varied between $0.65$ and $0.99$ in steps
of $0.01$ for duplication and loss events. For HGTs, a rate in the range between $0.1$ and $0.24$, again in steps of $0.01$,
was used. A detailed list of all simulated scenarios can be found in the
Supplemental Material. For each of the 14,000 true gene trees $\tilde T$
the total number $S_n$ of speciation events, $L_n$ of losses, $D_n$ of
duplications, and $H_n$ of HGTs was determined.  Summary statistics of the
simulated scenarios are compiled in the Supplemental Material.

From each true gene tree $\tilde T$ we extracted the observable gene tree
$T$ as described in Step (3). For all retained vertices the
  reconciliation map $\mu$ and thus the event labeling $t=t_{\mu}$ remains
unchanged.  Since $\lca_T(x,y)=\lca_{\tilde T}(x,y)$ for all extant genes
$x,y\in L(T)$, it suffices to consider $T$. The leaf coloring map
$\sigma:L(T)\to L(S)$ is obtained from its definition, i.e., setting
$\sigma(v)=\mu(v)$ for all $v\in L(T)$. We can now extract the orthology
relation and reciprocal best match relation from each scenario.

The orthology relation $\Theta(T,t)$ is easily constructed from the event
labeled gene tree $(T,t)$, since $xy \in \Theta(T,t)$ if and only if
$t(\lca_T(x,y)) = \SPEC$. An efficient way to compute $\Theta(T,t)$
  and the RBMG $(G,\sigma)$ that avoids the explicit evaluation of
  $\lca_T()$ is described in the Supplemental Material. For each
reconciliation scenario $(T,S,\mu)$, we also identify all good quartets in
the BMG $(\G,\sigma)$ and then delete the middle edge of the corresponding
$P_4$ from the RBMG $(G,\sigma)$. The resulting graph will be referred to
as $(G_4,\sigma_4)$.

\subsection{Simulation Results for Duplication/Loss Scenarios}
\label{ssect:simresDL}

In order to assess the practical relevance of co-RBMGs we measured the
abundance of non-cograph components in the simulated RBMGs. More precisely,
we determined for each simulated RBMG the connected components of its
restrictions to any three distinct colors and determined whether these
components are cographs, graphs of Type (B), or graphs of Type (C). In
order to identify these graph types, we used algorithms of \cite{Hoang:13}
to first identify an induced $P_4$ belonging to a good quartet. If one
exists, we check for the existence of an induced $P_5$ and then test
whether its endpoints are connected, thus forming a hexagon characteristic
for the a Type (C) graph. Otherwise, the presence of the $P_4$ implies Type
(B), while the absence of induced $P_4$s guarantees that the component is a
cograph.

We did not encounter a single Type (C) component in 14,000 simulated
scenarios. As we shall see this is a consequence of the fact that all
simulated trees are binary. To see this, we consider the structure of
connected 3-RBMG of Type (C) in some more detail, generalizing some
technical results by \citet{Geiss:19x}:
\begin{lemma}
  Let $(G,\sigma)$ be a connected 3-RBMG containing the induced $C_6$
  $\langle x_1 y_1 z_1 x_2 y_2 z_2\rangle$ with three distinct colors $r$,
  $s$, and $t$ such that $\sigma(x_1)=\sigma(x_2)=r$,
  $\sigma(y_1)=\sigma(y_2)=s$, and $\sigma(y_1)=\sigma(y_2)=t$. Then, every
  tree $(T,\sigma)$ that explains $(G,\sigma)$ must satisfy the following
  property: There exist distinct $v_1,v_2,v_3\in \child(v)$ where
  $v\coloneqq \lca_T(x_1,x_2,y_1,y_2,z_1,z_2)$ such that either
  $x_1,y_1\preceq_T v_1$, $x_2,z_1\preceq_T v_2$, $y_2,z_2\preceq_T v_3$ or
  $y_1,z_1\preceq_T v_1$, $x_2,y_2\preceq_T v_2$, $x_1,z_2\preceq_T v_3$.
  \label{lem:C6}
\end{lemma}
\begin{proof}
  If $|V(G)|>6$, then, due to the connectedness of $\G$, at least one
  of the six vertices of the induced $C_6$ is adjacent to more than one
  vertex of one of the colors $r,s,t$, hence the first statement
  immediately follows from Lemma 39(iii) in \citet{Geiss:19x}. Now consider
  the special case $|V(G)|=6$.  By Cor.\ 9 of \citet{Geiss:19x},
  $\G(T,\sigma)$ contains a good quartet. W.l.o.g.\ let
  $\langle x_1 y_1 z_1 x_2\rangle$ be a good quartet, thus
  $(x_1,z_1),(x_2,y_1)\in E(\G)$ and $(z_1,x_1),(y_1,x_2)\notin
  E(\G)$. This, in particular, implies
  $\lca_T(x_2,z_1)\prec_T\lca_T(x_1,z_1)$, thus there are distinct children
  $v_1, v_2\in \child(v)$ such that $x_1\preceq_T v_1$ and
  $x_2, z_1\preceq_T v_2$. Moreover, as $x_1y_1\in E(G)$ and
  $(y_1,x_2)\notin E(\G)$, we have
  $\lca_T(x_1,y_1)\prec_T \lca_T(x_2,y_1)$, hence $y_1\preceq_T v_1$. Now
  consider $y_2$. Since $x_1y_2\notin E(G)$ and $x_2y_2\in E(G)$, it must
  hold $\lca_T(x_2,y_2)\preceq_T \lca_T(x_1,y_2)$, hence
  $y_2\notin L(T(v_1))$. Assume, for contradiction, that
  $y_2\preceq_T v_2$. Then, as $y_2z_2\in E(G)$ and
  $\lca_T(y_2,z_1)\preceq_T v_2$, we clearly have $z_2\preceq_T
  v_2$. However, this implies $\lca_T(x_2,z_2)\prec_T\lca_T(x_1,z_2)$,
  contradicting $x_1z_2\in E(G)$. We therefore conclude that there must
  exist a vertex $v_3\in\child(v)\setminus \{v_1,v_2\}$ such that
  $y_2\preceq_T v_3$. One easily checks that this implies
  $z_2\preceq_T v_3$, which completes the proof.  \qed
\end{proof}

\begin{theorem}
  If $(T,\sigma)$ is a binary leaf-labeled tree, then $G(T,\sigma)$ does
  not contain a connected component of Type \AX{(C)}.
\end{theorem}
\begin{proof}
  By Obs.\ 6 of \cite{Geiss:19x}, the restriction $(T_{rst},\sigma_{rst})$
  of $(T,\sigma)$ explains the subgraph $(G_{rst},\sigma_{rst})$ of
  $G(T,\sigma)$ that is induced by vertices with color $r$, $s$, or
  $t$. Thm.\ 2 of \cite{Geiss:19x} shows, furthermore, that every connected
  component of $(G_{rst},\sigma_{rst})$ is explained by restriction
  $(T',\sigma')$ of $(T_{rst},\sigma_{rst})$ to the corresponding vertices.
  Now suppose $(T,\sigma)$ is a binary. Then both $(T_{rst},\sigma_{rst})$
  and $(T',\sigma')$ are also binary. By contraposition of Lemma
  \ref{lem:C6}, no $C_6$ as specified in Lemma \ref{lem:C6} can be
  explained by $(T',\sigma')$, and thus $G(T,\sigma)$ cannot contain a
  connected component of Type \AX{(C)}.
\end{proof}

\begin{figure}
  \begin{center}
    \setlength\tabcolsep{2pt} 
    \begin{tabular}{ccc}
      \begin{minipage}{0.315\textwidth}
        \textbf{(a)}\vspace{0.25cm}
        \includegraphics[width=\textwidth]{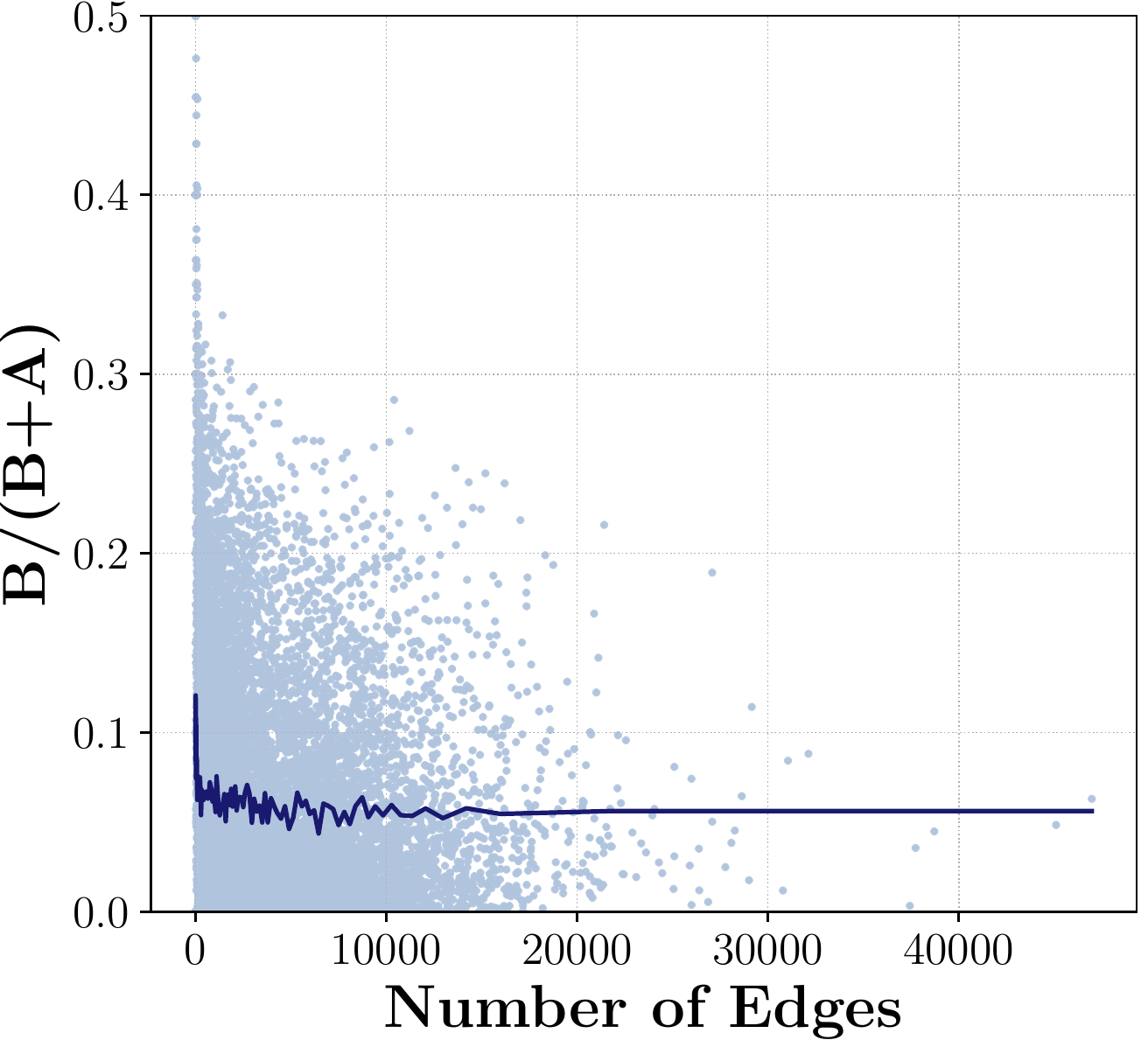} 
      \end{minipage}
      &
      \begin{minipage}{0.29\textwidth}
        \textbf{(b)}\vspace{0.25cm}
        \includegraphics[width=\textwidth]{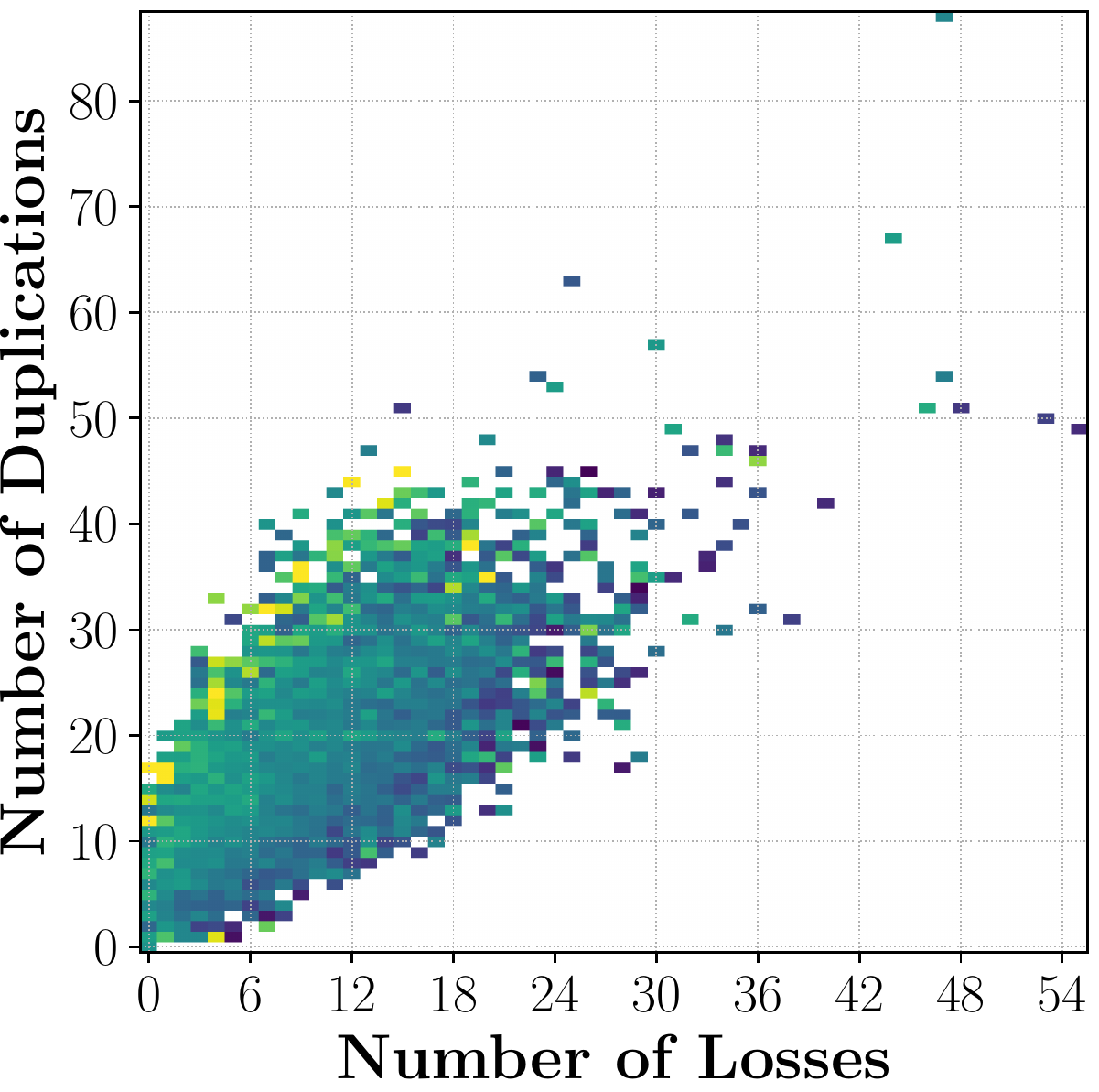}         
      \end{minipage}
      &
      \begin{minipage}{0.365\textwidth}
        \vspace{0.5cm}
        \includegraphics[width=\textwidth]{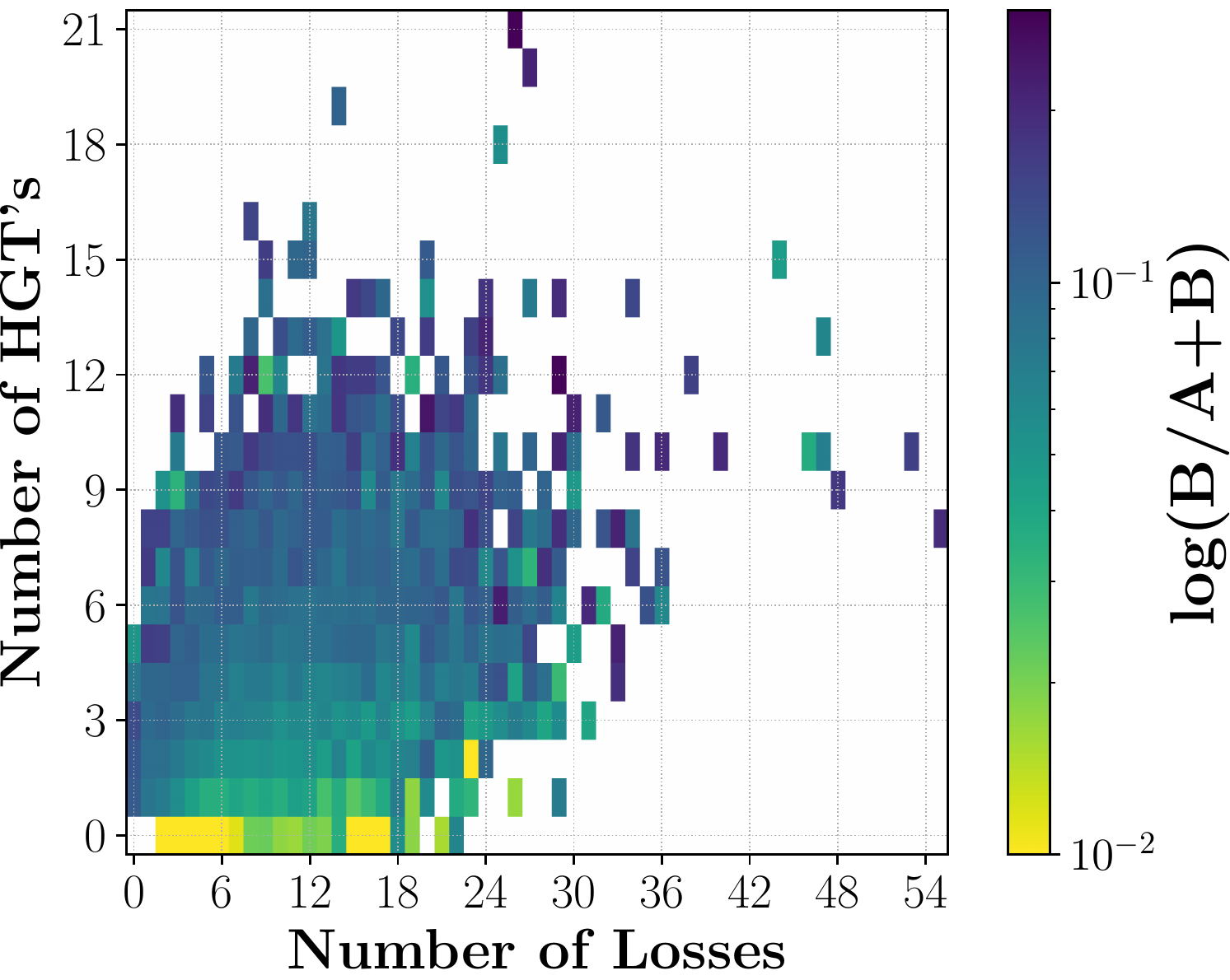}         
      \end{minipage}  
    \end{tabular}
  \end{center}
  \caption{Relative abundance $\eta=\frac{B}{B+A}$ of (B)-RBMGs in the
    simulation data. Panel (a) shows the dependence on the number of edges
    in the BMG in every simulated scenario, and its average depicted by the
    line in darker blue. Scatter plots (b) show the dependence of $\eta$ on
    the number of duplications and losses, and HGTs and losses,
    respectively.}
  \label{fig:B-RBMG}
\end{figure}

Although events that generate more than two offspring lineages are
logically possible in real data, most multifurcations in phylogenetic trees
are considered to be ``soft polytomies'', arising from data that are
  insufficient to produce a fully resolved, binary trees
\cite{Purvis:93,Kuhn:11,Sayyari:18}. Type (C) 3-RBMGs thus should be very
unlikely under biologically plausible assumptions on the model of
evolution. Here we only consider the abundance of Type (B) components
relative to all Type \AX{(A)} and \AX{(B)} components. We denote their
ratio by $\eta$. The results are summarized in Fig.~\ref{fig:B-RBMG}. We
find that $\eta$ is usually below 20\% and increases with the number of
loss and HGT events. More precisely, 83.47\% of the 14,000 scenarios have
at least one Type (B) component and 16.53\% do not have Type (B) components
at all. Among all 3-colored connected components taken from the
restrictions to any three colors, 94.41\% are of Type (A) and 5.59\% are of
Type (B).

A graph $G$ is called $P_4$-sparse if every induced subgraph on five
vertices contains at most one induced $P_4$ \cite{Jamison:92}. The interest
in $P_4$-sparse graphs derives from the fact that the cograph editing
problem is solvable in linear time from $P_4$-sparse graphs
\cite{LIU201245}. It is of immediate practical interest, therefore, to
determine the abundance of $P_4$-sparse RBMGs that are not cographs. Among
the 14,000 simulated scenarios, we found that about 20.9\% of the 3-colored
Type (B) components are $P_4$-sparse, while the majority contains
``overlapping'' $P_4$s. We then investigated the corresponding
$\sthin$-thin graphs. An undirected colored graph $(G,\sigma)$ is called
$\sthin$-thin if no distinct vertices are in relation $\sthin$. Two
vertices $a$ and $b$ are in relation $\sthin$ if $N(a)=N(b)$ and
$\sigma(a)=\sigma(b)$. Somewhat surprisingly, this yields a reversed
situation, where more than two thirds of the $\sthin$-thin 3-colored Type
(B) components are now $P_4$-sparse, while only a minority of 31.32\% is
not $P_4$-sparse.  An example of an undirected colored graph $(G,\sigma)$
and its corresponding $\sthin$-thin version $(G/\sthin, \sigmasthin)$,
which we found during our simluations, is shown in Panel (B) of Fig.\
\ref{fig:P4-sparse}.

\begin{figure}
  \begin{center}
    \begin{tabular}{ccc}
      \begin{minipage}{0.60\textwidth}
        \includegraphics[width=\textwidth]{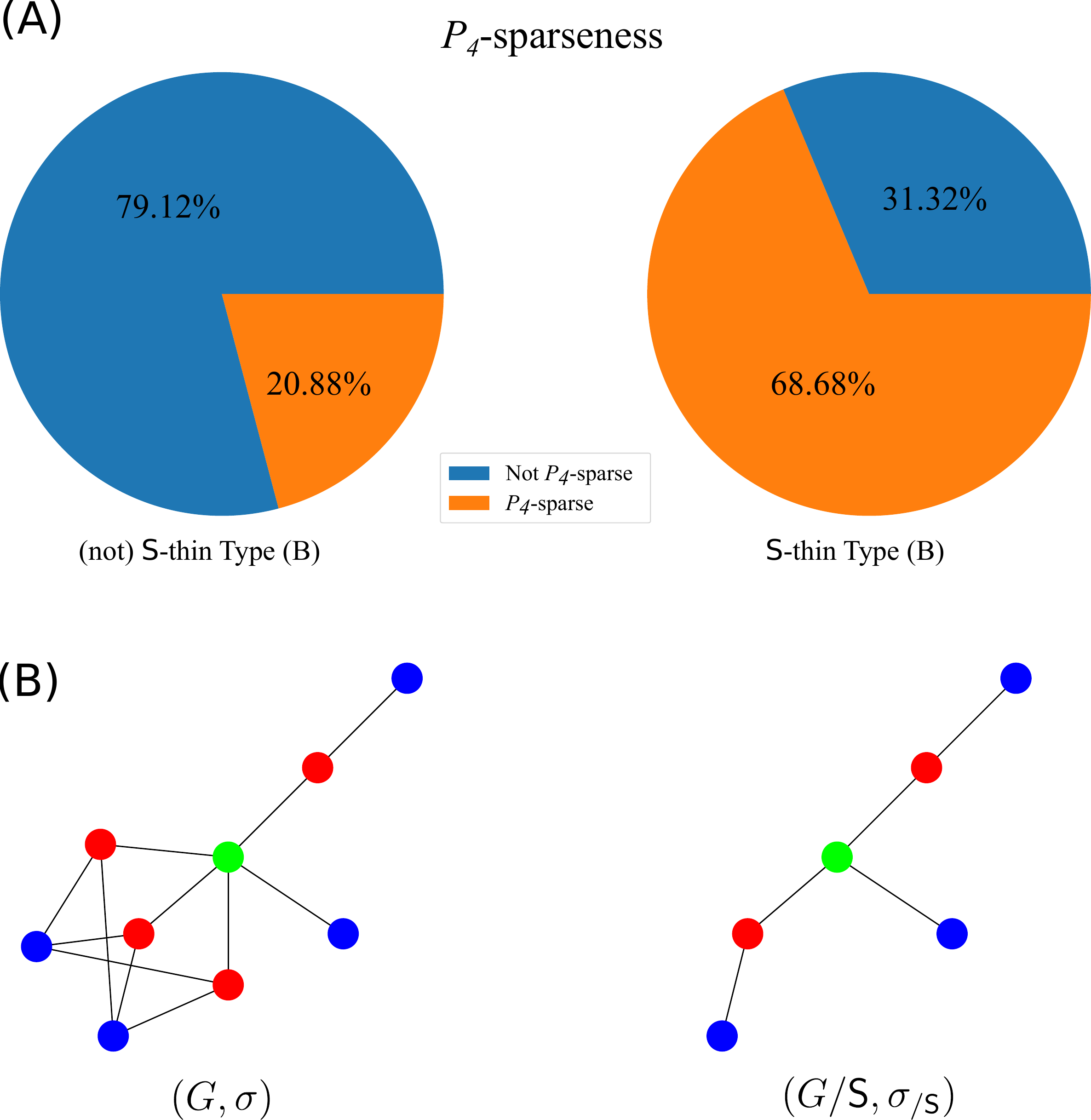}
      \end{minipage}
      &&
      \begin{minipage}{0.30\textwidth}
        \caption{\emph{Top:} Among our 14,000 simulated scenarios we found
          that a majority of 79.12\% of the (not necessarily $\sthin$-thin)
          3-colored Type (B) components are not $P_4$-sparse. For the
          corresponding $\sthin$-version of those 3-colored components only
          31.32\% are not $P_4$-sparse while 68.68\% are
          $P_4$-sparse. \emph{Below:} One of the simulated 3-colored Type
          (B) components $(G,\sigma)$, which is not $\sthin$-thin, and its
          corresponding $\sthin$-thin version $(G/\sthin, \sigmasthin)$.}
        \label{fig:P4-sparse}
      \end{minipage}
    \end{tabular}
  \end{center}
\end{figure}  

Next we investigated the relationship of the RBMG $G(T,\sigma)$ and the
orthology graph $\Theta$ (see Fig.\ \ref{fig:DL}). We empirically confirmed that
$E(\Theta)\subseteq E(G(T,\sigma))$ in the absence of HGT (not shown). Also
following our expectations, the fraction
$|E(G(T,\sigma))\setminus E(\Theta)|/|E(G(T,\sigma))|$ of false-positive
orthology predictions in an RBMG is small as long as duplications and
losses remain moderate (l.h.s.\ panel in Fig.~\ref{fig:DL}). Most of the
false positive orthology calls are associated with large numbers of losses
for a given number of duplication.

\begin{figure}[htbp!]
  \begin{center}
    \begin{tabular}{ccccc}
      %\hspace{0.5cm}
      %$\frac{|E(G(T,\sigma))\setminus E(\Theta)|}{|E(G(T,\sigma))|}$ &&
      %$ \frac{|E(G_4(T,\sigma))\setminus E(\Theta)|}{|E(G_4(T,\sigma))|}$\\  
      \begin{minipage}{0.4\textwidth}
      	\centering $\frac{|E(G(T,\sigma))\setminus E(\Theta)|}{|E(G(T,\sigma))|}$
        \includegraphics[width=\textwidth]{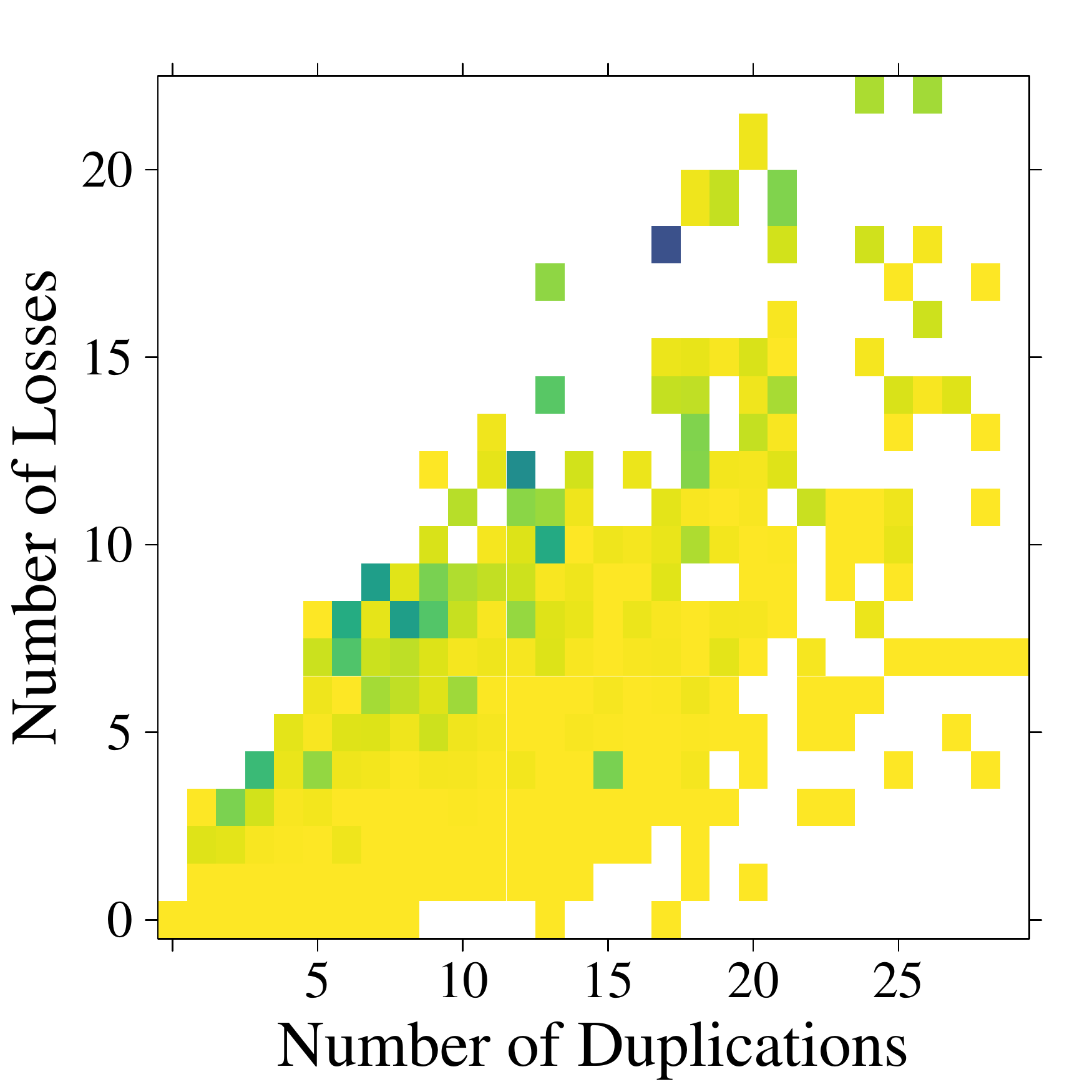}
      \end{minipage}
      & &
      \begin{minipage}{0.4\textwidth}
      	\centering$ \frac{|E(G_4(T,\sigma))\setminus E(\Theta)|}{|E(G_4(T,\sigma))|}$
        \includegraphics[width=\textwidth]{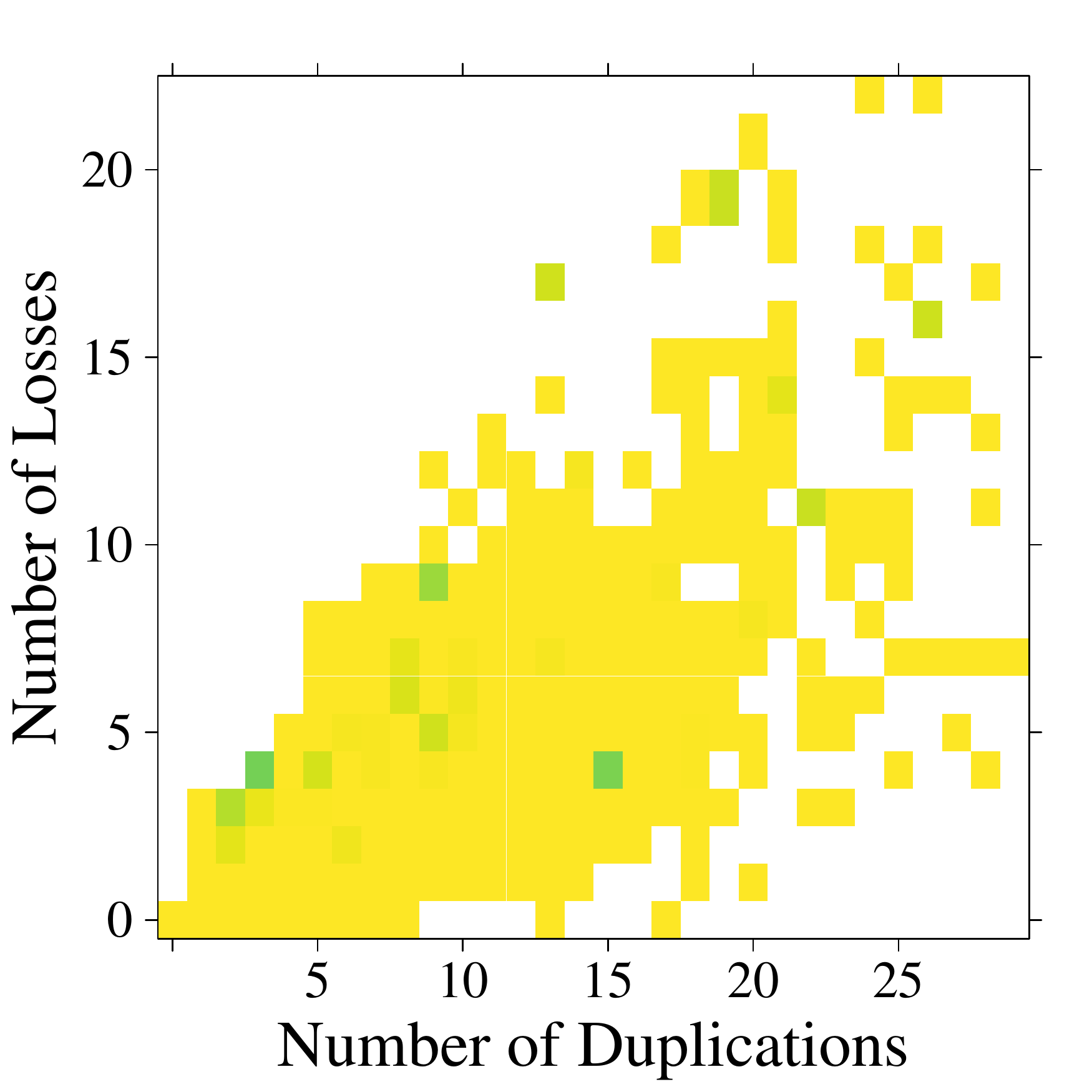}
      \end{minipage}
	  \begin{minipage}{0.1\textwidth} \vspace{0.19cm}
	  	\includegraphics[width=0.69\textwidth]{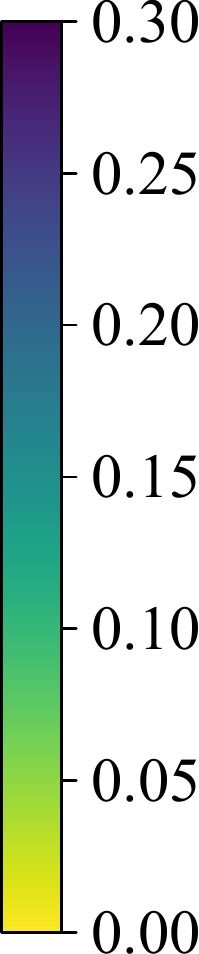}
	  \end{minipage}
    \end{tabular}
  \end{center}
  \caption{Fraction of non-orthology edges in the reciprocal best match
    graph (l.h.s.). The $x$-axis, resp., $y$-axis indicate the total number
    of duplications, resp., losses in the simulated scenarios.  Most of the
    false positive orthology assignments in the l.h.s.\ panel are removed
    by deleting the middle edge of good quartets (r.h.s.\ panel).  White
    background indicates \emph{no data}. }
  \label{fig:DL}
\end{figure}

We find that good quartets eliminate nearly all false positive edges from
the RBMG and leave a nearly perfect orthology graph (r.h.s.\ panel in
Fig.~\ref{fig:DL}).  As we have seen so far, reciprocal best matches indeed
form an excellent approximation of orthology in duplication-loss
scenarios. In particular, the good quartets identify nearly all false
positive edges, making it easy to remove the few remaining $P_4$s using a
generic cograph editing algorithm \cite{LIU201245}.

\section{Outlook: Evolutionary Scenarios with Horizontal Gene Transfer}
\label{sect:HGT}

The benign results above beg the question how robust they are under
HGT. Gene family histories with HGT have been a topic of intense study in
recent years \cite{Doyon2010,THL:11,BAK:12,Nojgaard:18a}. Following the
so-called DTL-scenarios as proposed e.g.\ by \citet{THL:11,BAK:12} we
relax the notion of reconciliation maps, since ancestry is no longer
preserved. We replace Axiom \AX{(R2)} by
\begin{description}
\item[\AX{(R2w)}] \emph{Weak Ancestor Preservation.}\\
  If $x\prec_T y$, then either $\mu(x)\preceq_S\mu(y)$ or $\mu(x)$ and
  $\mu(y)$ are incomparable w.r.t.\ $\prec_S$.
\end{description}
and add the following constraints
\begin{description}
\item[\AX{(R3.iii)}] \emph{Addition to the Speciation Constraint.}\\
  If $\mu(x)\in W^0$, then $\mu(v)\preceq_T\mu(x)$ for all $v\in\child(x)$.
\item[\AX{(R4)}] \emph{HGT Constraint.}\\
  If $x$ has a child $y$ such that $\mu(x)$ and $\mu(y)$ are incomparable, 
  then $x$ also has a child $y'$ with $\mu(y')\preceq_S \mu(x)$.
\end{description}
Property \AX{(R2w)} equivalently states that if $x\prec_T y$, then we must
not have $\mu(y)\prec_S\mu(x)$, which would invert the temporal
order. Property \AX{(R3.iii)} (which follows from \AX{(R2)} but not from
\AX{(R2w)}) ensures that the children of speciation events are still mapped
to positions that are comparable to the image of the speciation
node. Condition \AX{(R4)}, finally, requires that every horizontal transfer
event also has a vertically inherited offspring. Note that condition
\AX{(R4)} is void if \AX{(R2)} holds. In summary the axioms \AX{(R0)},
\AX{(R1)}, \AX{(R2w)}, \AX{(R3.i)}, \AX{(R3.ii)}, \AX{(R3.iii)}, and
\AX{(R4)} are a proper generalization of Def.\ \ref{def:reconc-map}. We
note that these axioms are not sufficient to ensure time consistency,
however. We refer to \citet{Nojgaard:18a} for details.  Our choice of
axioms also rules out some scenarios that may appear in reality (or
simulations), but which are not observable when only evolutionary
divergence is available as measurement. For example, Condition \AX{(R3.ii)}
excludes scenarios in which HGT events have no surviving vertically
inherited offspring.

We furthermore extend the event map $t$ for a gene tree $T$ to include HGT
as an additional event type denoted by the symbol $\HGT$. We define
$t:V(T)\to \{\ROOT,\LEAF,\SPEC,\DUPL,\HGT\}$ such that $t(u)=\HGT$ if and
only if $u$ has a child $v$ such that $\mu(u)$ and $\mu(v)$ are
incomparable. Since the offsprings of an HGT event are not equivalent, it
is useful to introduce an edge labeling $\lambda:E(T)\to\{0,1\}$ such that
$\lambda(uv)=1$ if $\mu(u)$ and $\mu(v)$ are incomparable w.r.t.\
$\prec_S$. This edge labeling is investigated in detail by
\citet{Geiss:18a} as the basis of Fitch's xenology relation. Alternatively,
the asymmetry can be handled by enforcing an ordering of the vertices, see
\cite{Hellmuth:17a}.

Evolutionary scenarios with horizontal transfer may lead to a situation
where two genes $x,y$ in the same species, i.e., with
$\sigma(x)=\sigma(y)$, derive from a speciation, i.e.,
$\lca_T(x,y)=\SPEC$. This is the case when the two lineages underwent an
HGT event that transferred a copy back into the lineage in which the other
gene has been vertically transmitted. We call such genes
\emph{xeno-orthologs} and exclude them from the orthology relation, see
Fig.\ \ref{fig:hgt-scenario}. This choice is motivated (1) by the fact
that, by definition, genes of the same species cannot be recognized as
reciprocal best matches, and (2) from a biological perspective they behave
rather like paralogs.  In scenarios with HGT we therefore modify the
definition of the orthology graph such that $E(G_1 \join G_2)$ is replaced
by
\begin{equation}
  E(G_1 \tilde\join G_2) := 
  E(G_1) \cup E(G_2) \cup \{ uv \mid u \in V(G_1),v \in V(G_2) \text{ and }
  \sigma(u) \neq \sigma(v) \}\,.
\end{equation}

\begin{figure}
  \begin{minipage}{0.5\textwidth}
  \begin{center}
    \includegraphics[width=0.6\textwidth]{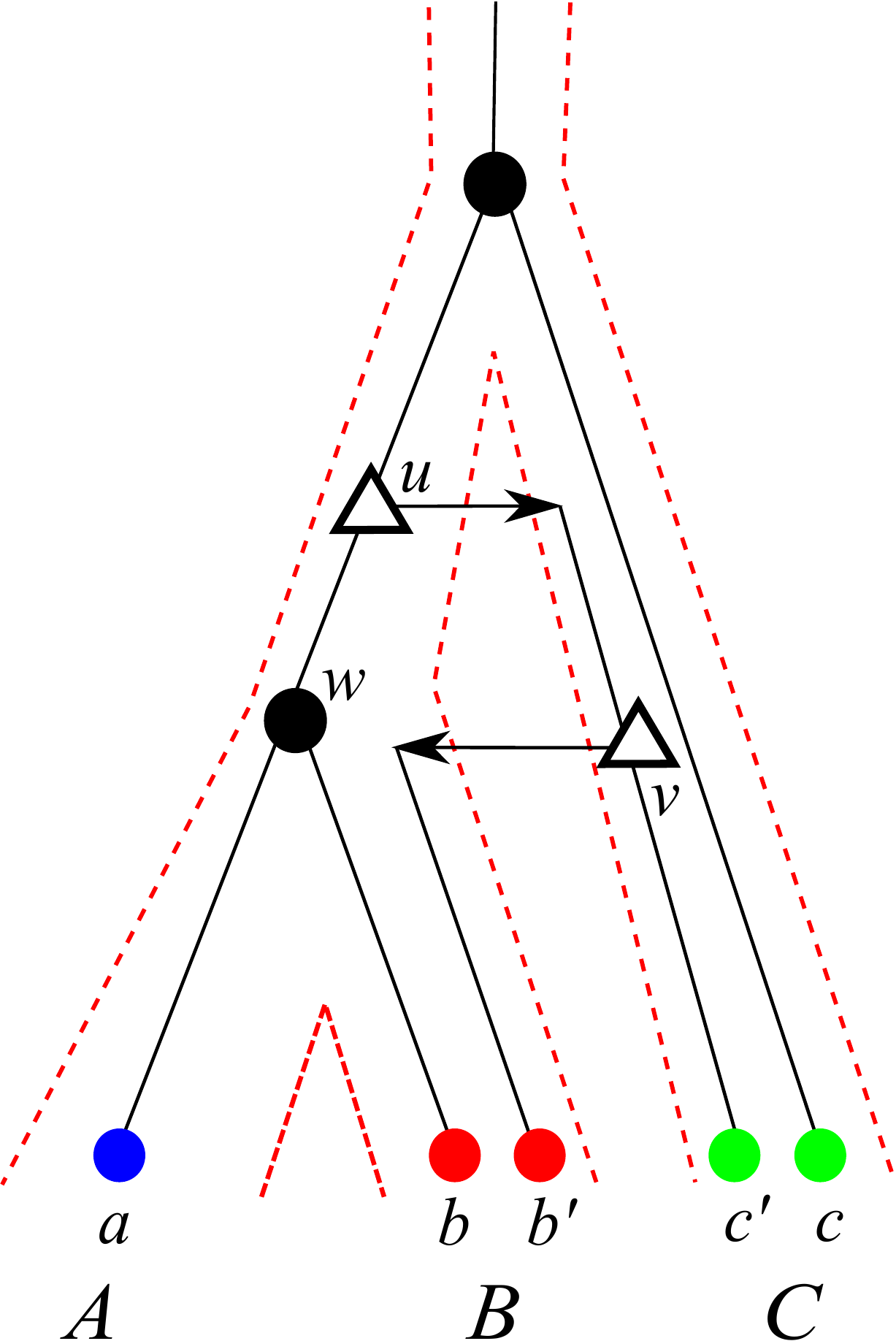}      
  \end{center}
\end{minipage}
\begin{minipage}{0.5\textwidth}
  \caption{A gene tree $(T,t,\lambda,\sigma)$ reconciled with a species
    tree $S$.  Here, we have two transfer edges $uv$ and $vb'$ with
    $t(u)=t(v)=\HGT$.  For the two children $w$ and $v$ of $u$ it holds
    $\sigma(L(T(w)))\cap \sigma(L(T(v)))\neq \emptyset$, a property that is
    shared with duplication vertices.  For the two children $b'$ and $c'$
    of $v$ it holds $\sigma(L(T(b')))\cap \sigma(L(T(c')))= \emptyset$, a
    property that is shared with speciation vertices.  In this example, $c$
    and $c'$ are xeno-orthologs and the pairs $(c,c'),(c',c)$ will be
    excluded from the resulting orthology relation. }
  \label{fig:hgt-scenario}
\end{minipage}
\end{figure}

The extremal map $\tT$ as in Def.\ \ref{def:event-rbmg2} cannot easily be
extended to include HGT, as the events $\SPEC$ and $\DUPL$ on some vertex
$u$ are solely defined on two exclusive cases: either $\sigma(L(T(u_1)))$
and $\sigma(L(T(u_2)))$ are disjoint or not for $u_1,u_2\in \child(u)$.  Both
cases, however, can also appear when we have HGT (see Fig.\
\ref{fig:hgt-scenario} for an example). That is, the fact that
$\sigma(L(T(u_1)))$ and $\sigma(L(T(u_2)))$ are disjoint or not, does not
help to unambiguously identify the event types in the presence of HGT.

Prop.\ \ref{prop:inftriple} can be generalized to the case that
$(T,t,\lambda,\sigma)$ contains HGT events.  The existence of
reconciliation maps from an event-labeled tree $(T,t,\lambda,\sigma)$ to an
\emph{unknown} species tree can be characterized in terms of species
triples $\sigma(a)\sigma(b)|\sigma(c)$ that can be derived from
$(T,t,\lambda,\sigma)$ as follows: Denote by
$\EHGT \coloneqq \{e\in E(T,t,\lambda,\sigma)\mid\lambda(e)=1\}$ the
set of all transfer edges in the labeled gene tree and let $(\Th,t,\sigma)$
be the forest obtained from $(T,t,\lambda,\sigma)$ by removing all
transfer edges.  By definition, $\mu(x)$ and $\mu(y)$ are incomparable for
every transfer edge $xy$ in $T$.  The set
$\mathcal{S}(T,t,\lambda,\sigma)$ is the set of triples
$\sigma(a)\sigma(b)|\sigma(c)$ where $\sigma(a)$, $\sigma(b)$, $\sigma(c)$
are pairwise distinct and either
\begin{enumerate}
\item $ab|c$ is a triple displayed by a connected component $T'$ of $\Th$
  such that the root of the triple is a speciation event, i.e.,
  $t(\lca_{T'}(a,b,c))=\SPEC$.
\item or $a,b\in L(\Th(x))$ and $c\in L(\Th(y))$ for some transfer edge
  $xy$ or $yx$ of $T$.
\end{enumerate}
\begin{proposition} \cite{Hellmuth2017}
  \label{prop:inftriple-new}
  Given an event-labeled, leaf-labeled tree $(T,t,\sigma)$.  Then, there is
  a reconciliation map $\mu:V(T)\to V(S)\cup E(S)$ to some species tree $S$
  if and only if $\mathcal{S}(T,t,\sigma)$ is compatible. In this case,
  $(T,t,\sigma)$ can be reconciled with every species tree $S$ that
  displays the triples in $\mathcal{S}(T,t,\sigma)$.
\end{proposition}
Here, we have not added additional constraints on reconciliation maps that
ensure that the map is also ``time-consistent'', that is, genes do not
travel ``back'' in the species tree, see \cite{Nojgaard:18a} for further
discussion on this.  However, Prop.\ \ref{prop:inftriple-new} gives at
least a necessary condition for the existence of time-consistent
reconciliation maps.  A simple proof of Prop.\ \ref{prop:inftriple-new} for
the case that $T$ is binary and does not contain HGT events can be found in
\cite{HernandezRosales:12a}. Moreover, generalizations of reconciling
event-labeled gene trees with species networks have been established by 
\citet{HHM:19}.

\begin{figure}[]
  \begin{center}
    \begin{tabular}{lll}
      \textbf{(a)} & \textbf{(b)}	& \textbf{(c)} \\		
      \includegraphics[width=0.30\textwidth]{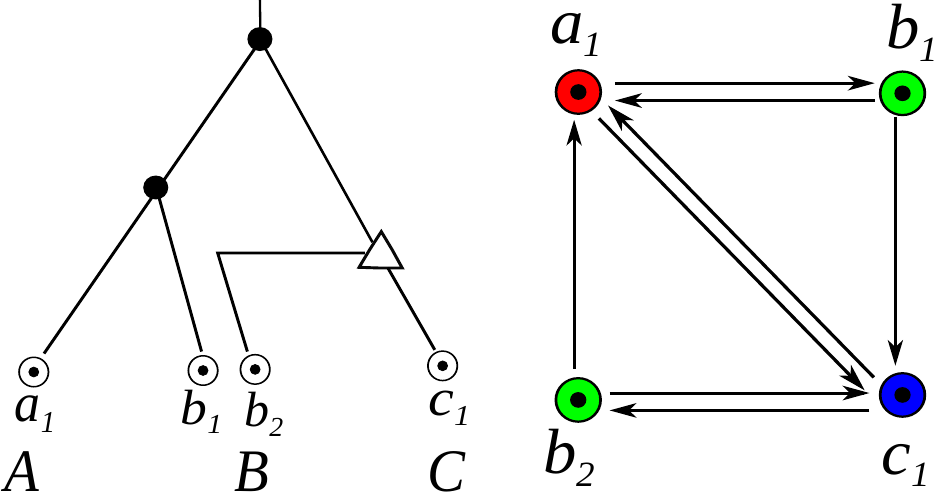} &	
      \includegraphics[width=0.30\textwidth]{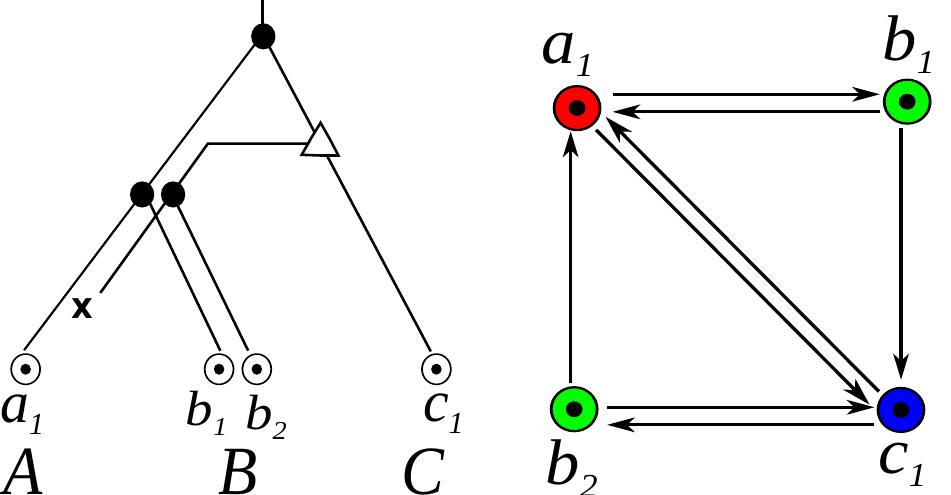} &
      \includegraphics[width=0.30\textwidth]{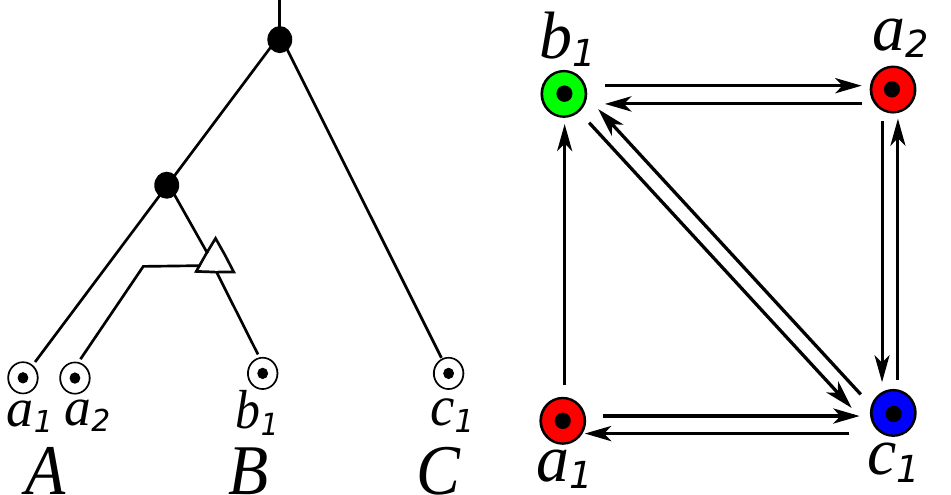} \\ 
      \textbf{(d)} & \textbf{(e)}	& \textbf{(f)} \\      
      \includegraphics[width=0.30\textwidth]{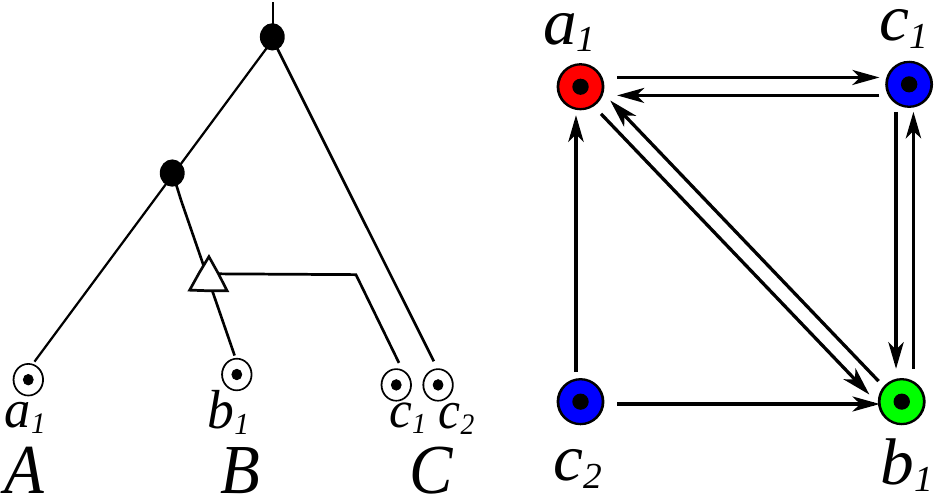} &
      \includegraphics[width=0.30\textwidth]{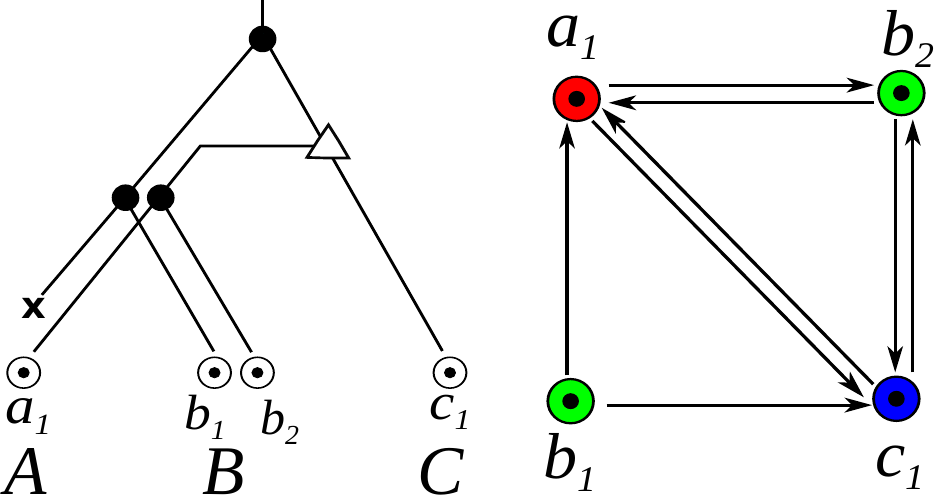} & 
      \includegraphics[width=0.30\textwidth]{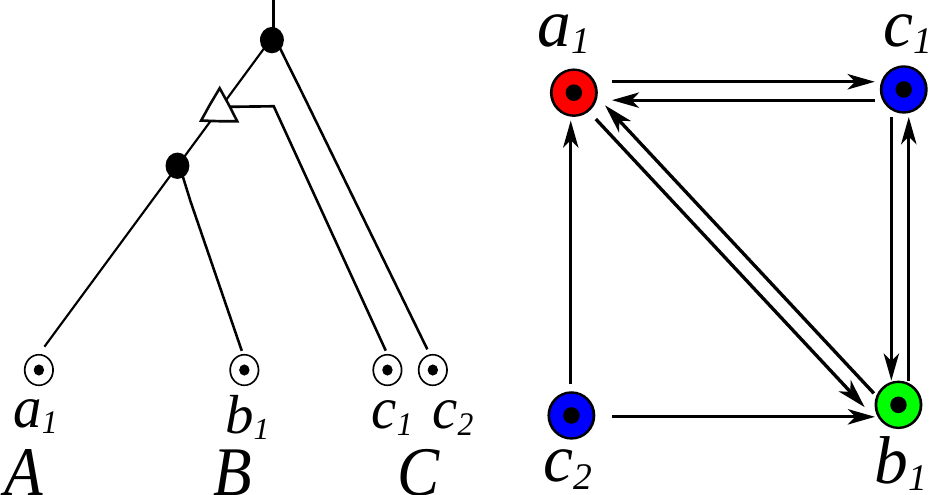} \\      
    \end{tabular}
  \end{center} 
  \caption{Scenarios with four genes, three species, and a single HGT event
    for which RBMG $G(T,\sigma)$ and orthology relation $\Theta(T,t)$
    differ.  The BMG is shown for each scenario.  In the first two cases
    (a) and (b), $G(T,\sigma)$ contains an induced $P_4$ in the RBMG, which
    might serve as indication for HGT events. In the remaining cases, the
    $G(T,\sigma)$ is a cograph, which does not represent the correct
    orthology relation, however.  In scenario (c), the graph $G(T,\sigma)$
    is a triangle with an attached edge, while the orthology relation is
    given by $\Theta(T,t)=K_4-e-f$ with the missing edges $e=a_1a_2$ and
    $f=a_1b_1$, where the latter results from the xenologous pair
    $a_2,b_1$. In the remaining three cases (d)-(f), the RBMG is
    $K_3\cupdot K_1$ compared to the orthology relation
    $\Theta(T,t)=K_4-e$, where the edge $e$ again corresponds to the edge
    between genes of the same species.}
  \label{fig:Dulce}
\end{figure}

In contrast to pure DL scenarios, it is no longer guaranteed that all true
orthology relationships are also reciprocal best matches.
Fig.~\ref{fig:Dulce} gives counterexamples. In three of these scenarios the
RBMG contains an induced $P_4$ that mimics a good quartet. Removal of the
middle edge of good quartets therefore not only reduces false positives in
DL scenarios but also introduces additional false negatives in the presence
of HGT.

\begin{figure}[t]
  \includegraphics[width=\textwidth]{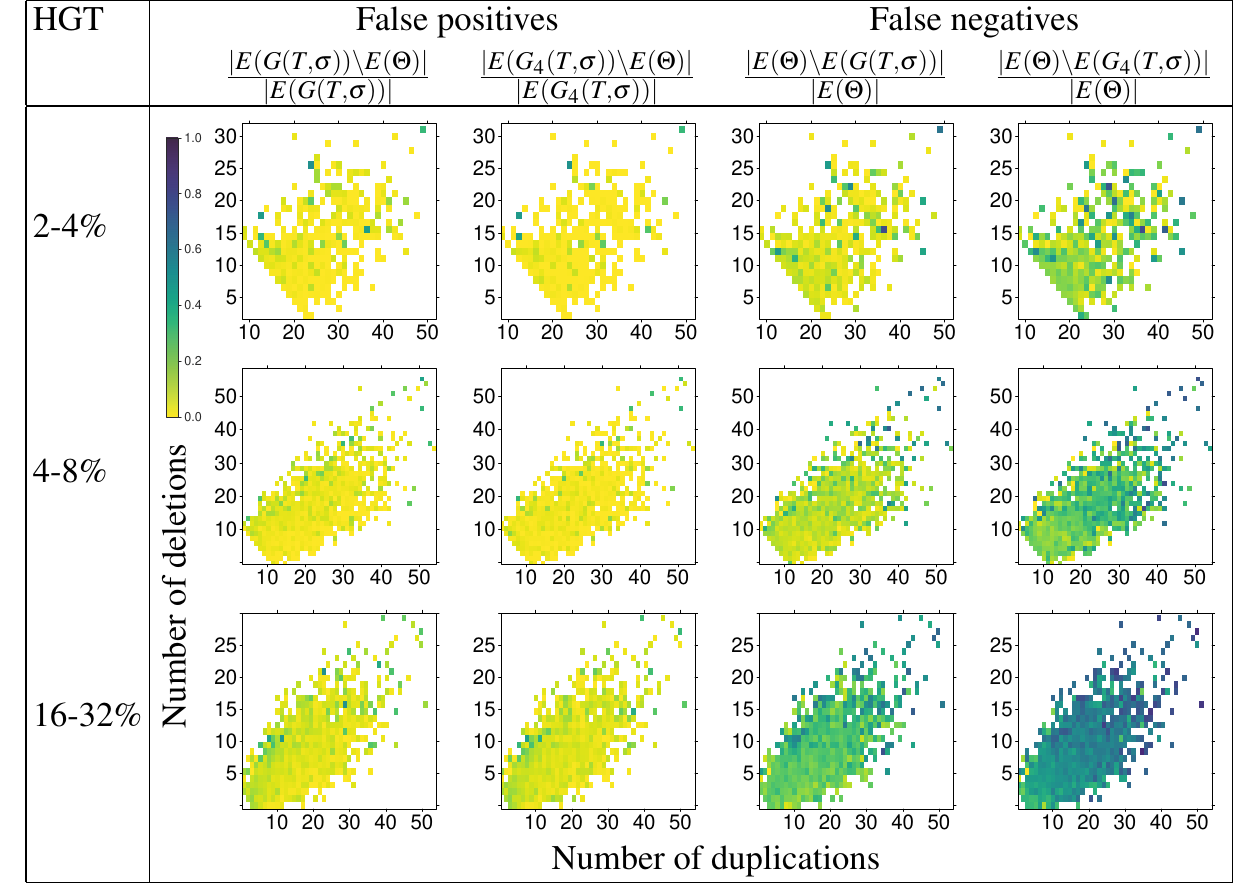}
  \caption{Dependence of the fraction of false positive and false negative
    orthology assignments in RBMGs in the presence of different levels of
    HGT, measured as percentage of HGT events among all events in the
    simulated true gene trees $\tilde T$. As in Fig.\ \ref{fig:DL}, data
    are shown as functions of the number of duplication and loss events in
    the scenario. While the number of false positives seems to depend very
    little on even high levels of HGT, the fraction of false negatives is
    rapidly increasing. Since HGT introduces good quartets that comprise
    only true orthology edges, their removal further increases the false
    positive rate (last column). }
  \label{fig:htgsim}
\end{figure}
         
\section{Discussion}

In the theoretical part of this contribution we have clarified the
relationships between (reciprocal) best match graphs (RBMGs), orthology,
reconciliation map, gene tree, species tree, and event map for the case of
duplication loss scenarios.

The orthology graph $\Theta$ is necessarily a subgraph of the RBMG. In the
absence of HGT, RBMGs therefore produce only false positive but no false
negative orthology assignments. Using not only reciprocal best matches but
all best matches, furthermore, shows that good quartets identify almost all
false positive edges. Removing the central edge of all good quartets in
$(\G,\sigma)$ yields nearly perfect orthology estimates. This, however,
implies that orthology inference is not solely based on reciprocal best
matches. Instead, it is necessary to also include certain directional best
matches, namely those that identify good quartets.

We observed that a small number of HGT events can cause large deviations
between the RBMG $(G,\sigma)$ and the orthology graph $\Theta$. However, we
have considered here the worst-case scenario, where HGT events occur between
relatively closely related organisms. While this is of utmost relevance in
some cases, for instance for toxin and virulence genes in bacteria, it is
of little concern e.g.\ for the evolution of animals. In the latter case,
xenologs almost always originate from bacteria or viruses, i.e., from
outgroups. The xenologs then form their own group of co-orthologs and
behave as if they would have been lost in the species outside the subtree
that received the horizontally transfered gene. 

From a more theoretical point of view, our empirical findings in the HGT
case beg two questions: (1) Are there \emph{local} features in the (R)BMG
that make it possible to unambiguously identify HGT, at least in some
cases? (2) What kind of additional information can be integrated to
distinguish good quartets arising from duplication/loss events that can be
safely removed from those that are introduced by HGT and should be
``repaired'' in a different manner. Most obviously, one may ask whether the
Fitch relation is sufficient (we conjecture that this is the case)
\cite{Geiss:18a,HS:19}, or whether it suffices to know that a leaf is a
(recent) result of transfer (we conjecture that this is not enough in
general).

The identification of edges in the RBMG that should or should not be
removed has important implication for orthology detection approaches that
enforce the cograph structure of the predicted orthology relation by means
of cograph editing. While this is an NP-complete problem \cite{LIU201245}
in general, the complexity of the colored version, i.e., editing a properly
colored graph to the nearest \hc-cograph remains open. The removal of false
positive edges identified by good quartets empirically reduces the number
of induced $P_4$s drastically. This observation also suggests to consider
\hc-cograph editing with a given best match relation. We suspect that the
additional knowledge of the directed edges makes the problem tractable
since it already implies a unique least resolved tree that captures much of
the cograph structure.

Cograph editing would be fully content with \hc-cographs, i.e.,
co-RBMGs. These are not necessarily ``biologically feasible'' in the sense
that they can be reconciled with a species tree. It will therefore be of
interest to consider the problem of editing an \hc-cograph to another
\hc-cograph that is reconcilable with some or a given species tree -- a
problem that has been considered already for orthology relations
\cite{Lafond:16,Lafond:14}.  Since the obstructions are conflicting triples
with a speciation at their top node, the offending data are conflicting
orthology assignments. It seems natural therefore to phrase the problem not
as an arbitrary editing problem but instead to ask for a maximal induced
sub-\hc-cograph that implies a compatible triple set. If it is indeed true
that triples necessarily displayed by the species tree can be extracted
directly from the c(R)BMG, it will be of practical use to consider the
corresponding edge deletion problem for c(R)BMGs. In particular, it would
be interesting to know whether the latter problem is the same as asking for
the maximal compatible subset of triples implied by the c(R)BMG or co-BMG?

\section*{Acknowledgements}
This work was support in part by the German Federal Ministry of Education
and Research (BMBF, project no.\ 031A538A, de.NBI-RBC) and the Mexican
Consejo Nacional de Ciencia y Tecnolog{\'i}a (CONACyT, 278966 FONCICYT
2).

\begin{appendix}
	\section{Extraction of Graphs from Simulated Evolutionary Scenarios}
	
	The simulations of evolutionary scenarios (described in the main text)
	result in an event-labeled gene tree $(T,t,\sigma)$ as well as an explicit
	reconciliation map $\mu:V(T)\to V(S)\cup E(S)$. From these data we have to
	construct the orthology graph $\Theta(T,t)$ and the RBMG
	$G(T,\sigma)$. This can be achieved in $O(L^2)$ time using Tarjan's
	off-line lowest common ancestors algorithm \cite{Tarjan:79,Gabow:83} to
	first tabulate $\lca_T(x,y)$ for all $x,y\in L$ in quadratic total time or
	with the help of additional data structures that then allow to answer least
	common ancestor queries in constant time \cite{Harel:84,Schieber:88}. As
	show below, it is also possible to avoid computation of the $\lca()$
	function altogether.
	
	\subsection{Orthology Graphs}
	
	The orthology relation $\Theta(T,t)$ is easily constructed from the event-labeled gene tree $(T,t)$ by a simple recursive construction. 
	For each
	$v\in \tilde T$ we define a graph $\Theta(v)$ recursively: if $v$ is a
	leaf, then $\Theta(v)$ is the $K_1$ with vertex set $\{v\}$ whenever $v$ is
	an extant gene and $\Theta(v)=\emptyset$, the empty graph, if $v$ is a loss
	event.  For inner vertices we set
	\begin{equation}
	\Theta(v) =
	\begin{cases}
	\displaystyle\bigjoin_{u\in\child(v)} \Theta(u)
	& \textrm{if } t(v)=\SPEC \\
	\displaystyle\bigcup_{u\in\child(v)}  \Theta(u)
	& \textrm{otherwise}
	\end{cases}
	\end{equation}
	Since $H\join \emptyset = H\cup\emptyset = H$, there is no contribution of
	the loss-leafs. Thus $\Theta(v)$ can be computed in exactly the same manner
	from the observable gene tree $T$. Hence,
	$\Theta(\rho_{T})=\Theta(\rho_{\tilde T})=:\Theta$ is the orthology graph
	of the scenario. Note that the planted root $0_T$ does not appear as the
	last common ancestor of any two leaves in $L(T)$, hence it suffices to
	consider the root. Although the next result is an immediate consequence of
	the definition of cographs and their corresponding cotrees
	\cite{Corneil:81}:
	\begin{lemma}\label{lem:spec}
		Let $(T,t,\sigma)$ be an event-labeled, leaf-labeled tree. Then
		$xy\in E(\Theta(v))$ if and only if $t(\lca_{T}(x,y))=\SPEC$.
	\end{lemma}
	By construction, $\Theta(u)$ is an induced subgraph of $\Theta(v)$ whenever
	$u\preceq_T v$. It is thus sufficient to store the binary $|L|\times|L|$
	adjacency matrix of $\Theta$. Traversing $T$ in post-order, one sets
	$\Theta_{xy}=1$, i.e., $xy\in E(\Theta)$, for all $xy$ with
	$x\in L(T(u_1))$ and $y\in L(T(u_2))$ where $u_1$ and $u_2$ are distinct
	children of $v$, if and only if $v$ is a speciation vertex. Since the pair
	$x,y$ is considered exactly once, namely when $v=\lca(x,y)$ is encountered
	in the traversal of $T$, the total effort is $O(|L|^2)$.
	
	\subsection{Best Match Graphs} 
	
	In order to compute the BMG $\G(T,\sigma)$ we associate every inner vertex
	$v$ with the lists $L_r(v):= \{x\in L(T(v))| \sigma(x)=r\}$ of leaves below
	$v$ with color $r$. We have $L_r(v)=\bigcup_{u\in\child(v)} L_r(u)$ for
	inner vertices, while leaves are initialized with $L_r(v)=\{v\}$ if
	$\sigma(v)=r$, and $L_r(v)=\emptyset$ if $\sigma(v)\ne r$. Again this can
	be achieved in not more than quadratic time. Now define
	$C_{\neg s}(v):=\{u\in\child(v)| L_s(u)=\emptyset\}$ and
	$C_{s}(v):=\{u\in\child(v)| L_s(u)\ne\emptyset\}$. Best matches can be
	retrieved directly from these auxiliary sets:
	\begin{lemma}\label{lem:bm}
		Let $u_1$ and $u_2$ be two distinct children of some inner vertex $v$ of
		the leaf-colored tree $(T,\sigma)$ and let $x\in L(T(u_1))$ with
		$\sigma(x)=r$ and $y\in L(T(u_2))$ with $\sigma(y)=s\neq r$. Then
		$(x,y)$ is a best match in $(T,\sigma)$ if and only if
		\begin{equation*}
		u_1\in  C_r(v)\cap C_{\neg s}(v)
		\quad\textrm{and}\quad
		u_2\in C_{s}(v).
		\end{equation*}
	\end{lemma}
	\begin{proof}
		If $L_s(u_1)=\emptyset$, then there is no best match of color $s$ for $x$
		in $L(T(u_1))$, i.e., any best match $\sigma(y')=s$ satisfies
		$v\preceq \lca(x,y')$. From $\lca(x,y)=v$ we see that $(x,y)$ is indeed a
		best match. On the other hand, if $L_s(u_1)\ne\emptyset$, then there is a
		leaf $y'\in L_s(u_1)$ with $\lca(x,y')\preceq u_1\prec v=\lca(x,y)$, and
		thus $y$ is not a best match for $x$.  \qed
	\end{proof}
	
	\begin{algorithm}
		\caption{Construction of $\G(T,\sigma)$}
		\label{alg:BMG}
		\algsetup{linenodelimiter=}
		\begin{algorithmic}
			\REQUIRE leaf-colored tree $(T,\sigma)$
			\FORALL {leaves $v$ of $T$, colors $r$}
			\STATE $L(T(v))=\{v\}$ 
			\IF {$\sigma(v)=r$}
			\STATE $\ell_{vr}=1$
			\ELSE
			\STATE $\ell_{vr}=0$
			\ENDIF
			\ENDFOR
			\FORALL {inner vertices $v$ of $T$ in postorder}           
			\FORALL {$u_1,u_2\in\child(v)$, $u_1\ne u_2$}
			\FORALL {$x\in L(T(u_1))$ and $y\in L(T(u_2))$}
			\STATE $(x,y)\in \G(T,\sigma)$ if $\ell_{u_1\sigma(y)}=0$
			\ENDFOR
			\ENDFOR
			\STATE $L(T(v))=\bigcup_{u\in\child(v)} L(T(u))$
			\FORALL {$u\in\child (v)$, colors $r\in S$}
			\STATE $\ell_{vr}=1$ if $\ell_{ur}=1$
			\ENDFOR
			\ENDFOR 
		\end{algorithmic}
	\end{algorithm}
	
	This observation yields the very simple way to construct $\G(T,\sigma)$.
	Algorithm \ref{alg:BMG} iterates over all pairs of vertices $x,y\in L$ such
	that each pair is visited exactly once by considering for every interior
	vertex $v$ exactly the pairs that are members of two distinct subtrees
	rooted at children $u_1$ and $u_2$ of $v$.  Since $y\in L_{\sigma(y)}(u_2)$
	and $x\in L_{\sigma(x)}(u_1)$ is guaranteed by construction, $(x,y)$ is a
	best match if and only if $L_{\sigma(y)}(u_1)=\emptyset$ by
	Lemma~\ref{lem:bm}. Using the precomputed binary variable $\ell_{vr}$ with
	value $1$ if $L_{r}(v)\ne \emptyset$ and $\ell_{vr}=0$ otherwise, this can
	be done in constant time $O(|L|)$. By traversing $T$ in postorder, finally,
	we can compute the lists of leaves $L(v)$ on the fly.  Since no subtree is
	revisited, there is no need to retain the $L(T(u))$ for the children, i.e.,
	for each vertex $v$, the lists of its children can simply be concatenated.
	Similarly, the variables $\ell_{vr}$ can be obtained while traversing $T$
	using the fact that $\ell_{vr}=1$ if and only if $\ell_{ur}=1$ for at least
	one of its children. Hence, Algorithm \ref{alg:BMG} runs in $O(|L|^2)$ time
	with $O(|L|\,|S|)$ memory using a single postorder traversal of $T$.
	
	The RBMG $G(T,\sigma)$ is now easily obtained from the BMG $\G(T,\sigma)$
	by extracting its symmetric part. Clearly the effort for this step is also
	bounded by $O(|L|^2)$.
	
	\subsection{Good Quartets} 
	
	We have seen in Section \ref{sect:goodbadugly} that at least some false
	positive edges are identified by good quartets. A convenient way of listing
	all good quartets $Q$ in $(\G,\sigma)$ makes use of the \emph{degree
		sequence} of $\G$, that is, the list
	$\alpha=((\alpha_x^+,\alpha_x^-)|x\in V(\G) )$ of pairs
	$(\alpha_x^+,\alpha_x^-)$ where $\alpha_x^+$ is the out-degree and
	$\alpha_x^-$ is the in-degree of the vertex $x\in V(\G)$ and the list is
	ordered in positive lexicographical order. One easily checks that a good
	quartet contains neither a \emph{2-switch} nor an \emph{induced 3-cycle},
	hence $Q$ is uniquely defined by its degree sequence
	$((2,1),(2,3),(2,3),(2,1))$ as a consequence of \cite[Thm.\
	1]{Cloteaux:14}. Regarding the coloring, it suffices to check that the two
	endpoints, that is, the vertices with indegree $1$, have the same color
	$\sigma(u)=\sigma(x)$. This already implies
	$\sigma(v),\sigma(w)\ne \sigma(u)=\sigma(x)$. Since there is an edge
	between $v$ and $w$, we also have $\sigma(v)\ne\sigma(w)$, i.e., the colors
	are determined up to a permutation of colors.  The false positive edge is
	the one connecting the two vertices with outdegree $3$.
	
	\clearpage
	
	\section{Additional Information on Simulated Scenarios}
	
	\begin{figure}[h] \label{Fig:SimInfo}
		\begin{tabular}{cc}
			\begin{tabular}{ll}
				\textbf{(a)} & \textbf{(b)} \\[5pt]        
				\includegraphics[width=0.4\textwidth]{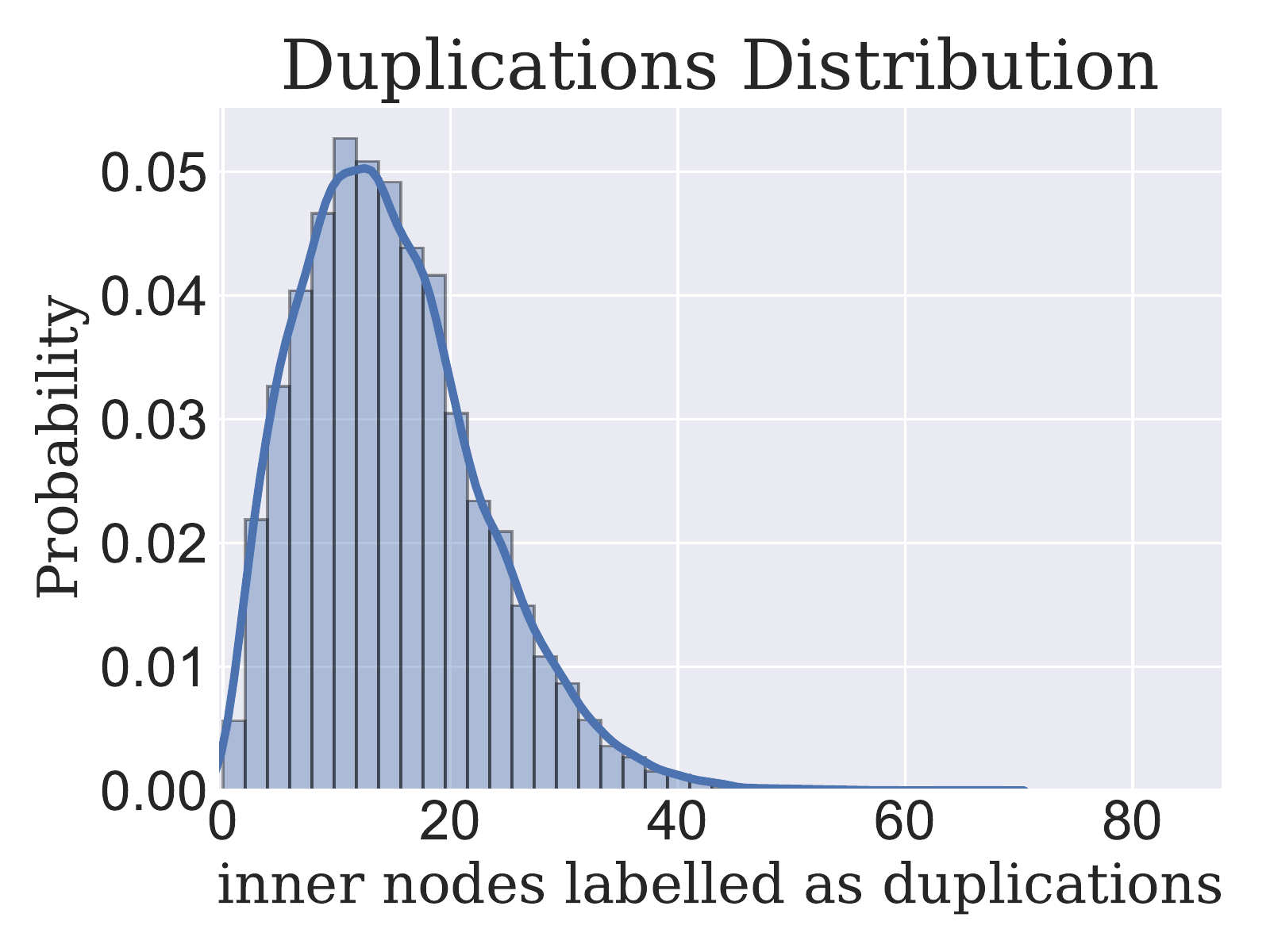} 
				& \includegraphics[width=0.4\textwidth]{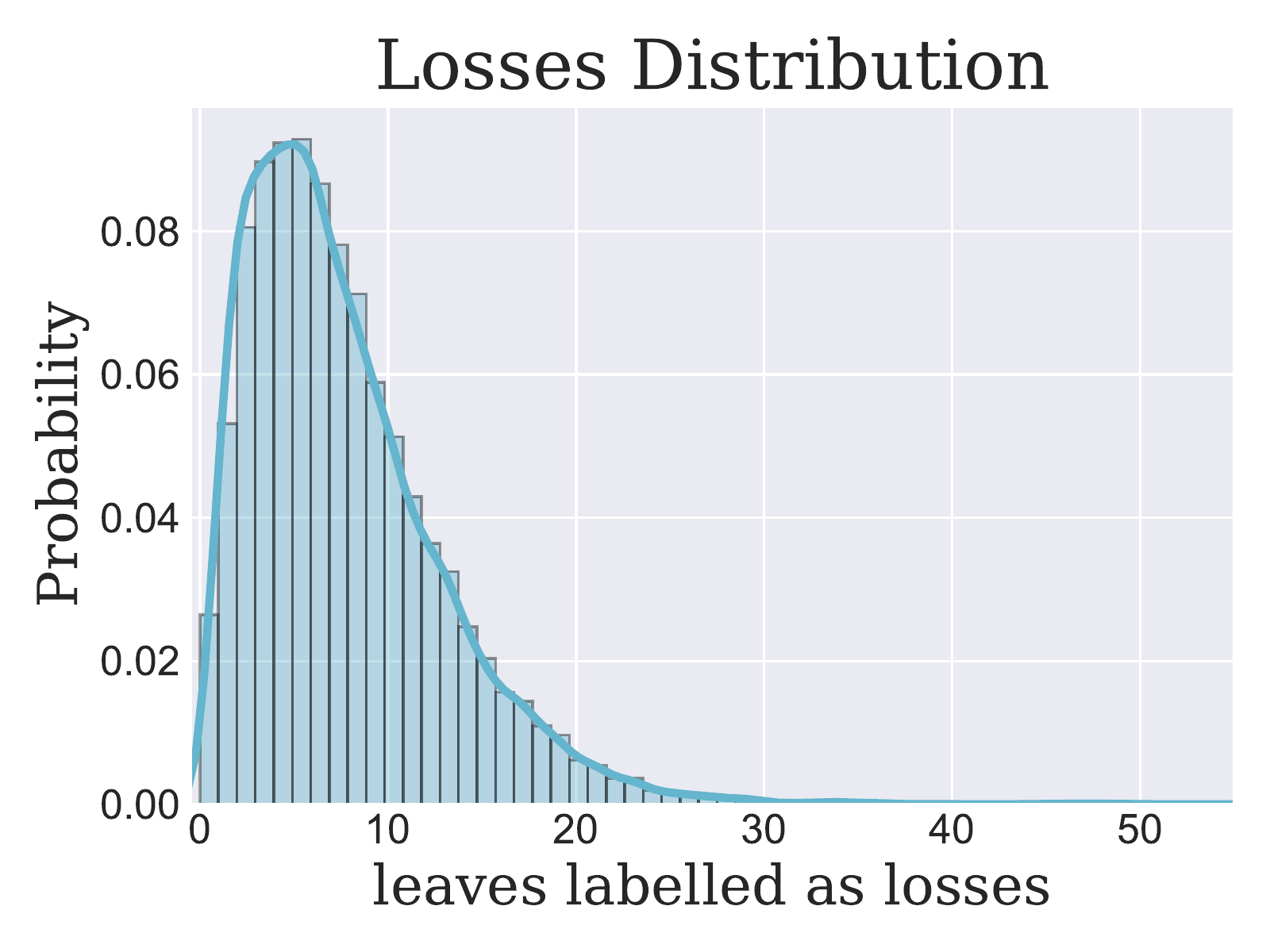} \\
				\textbf{(c)} & \textbf{(d)} \\[5pt]        
				\includegraphics[width=0.4\textwidth]{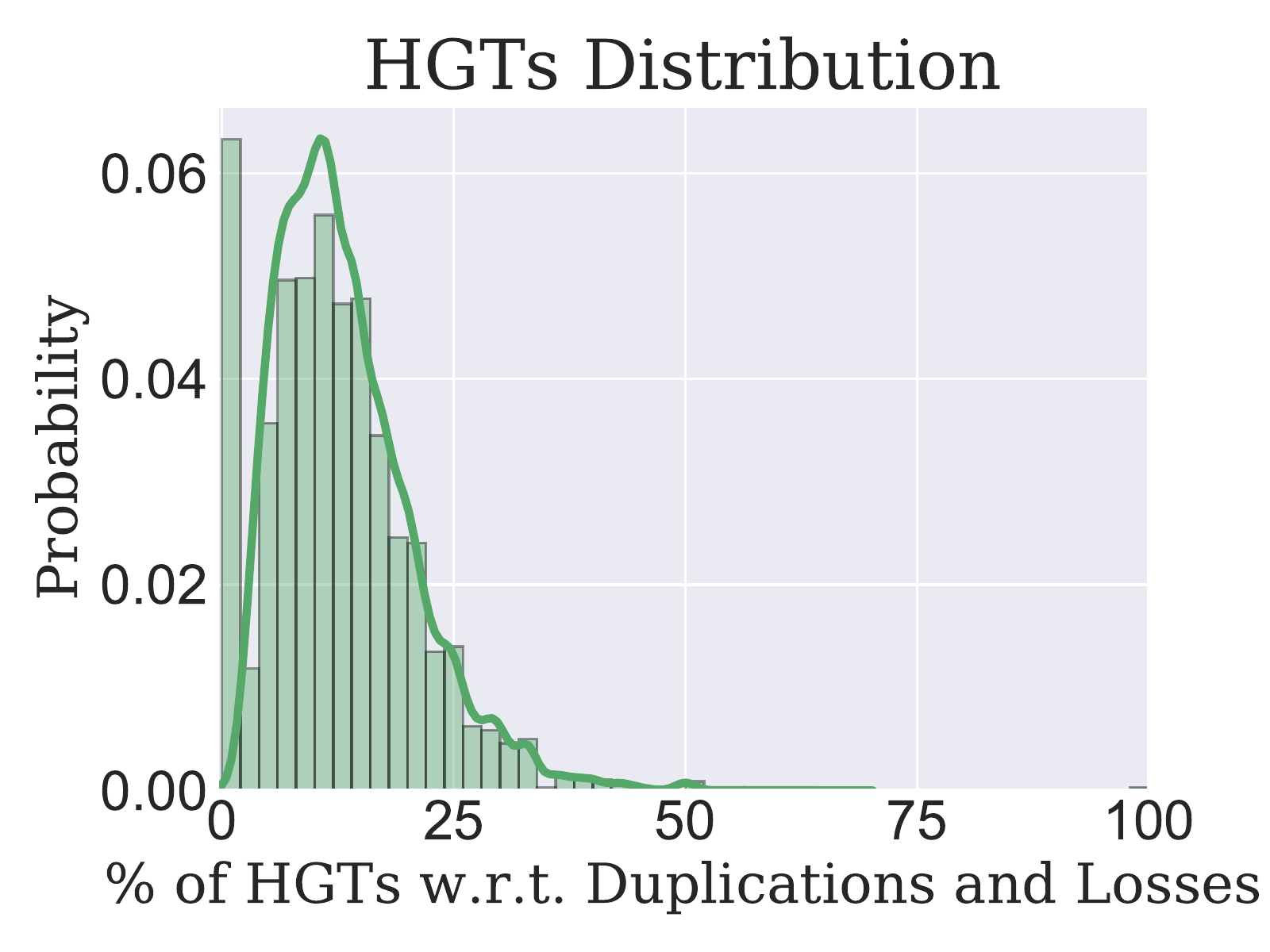} 
				& \includegraphics[width=0.4\textwidth]{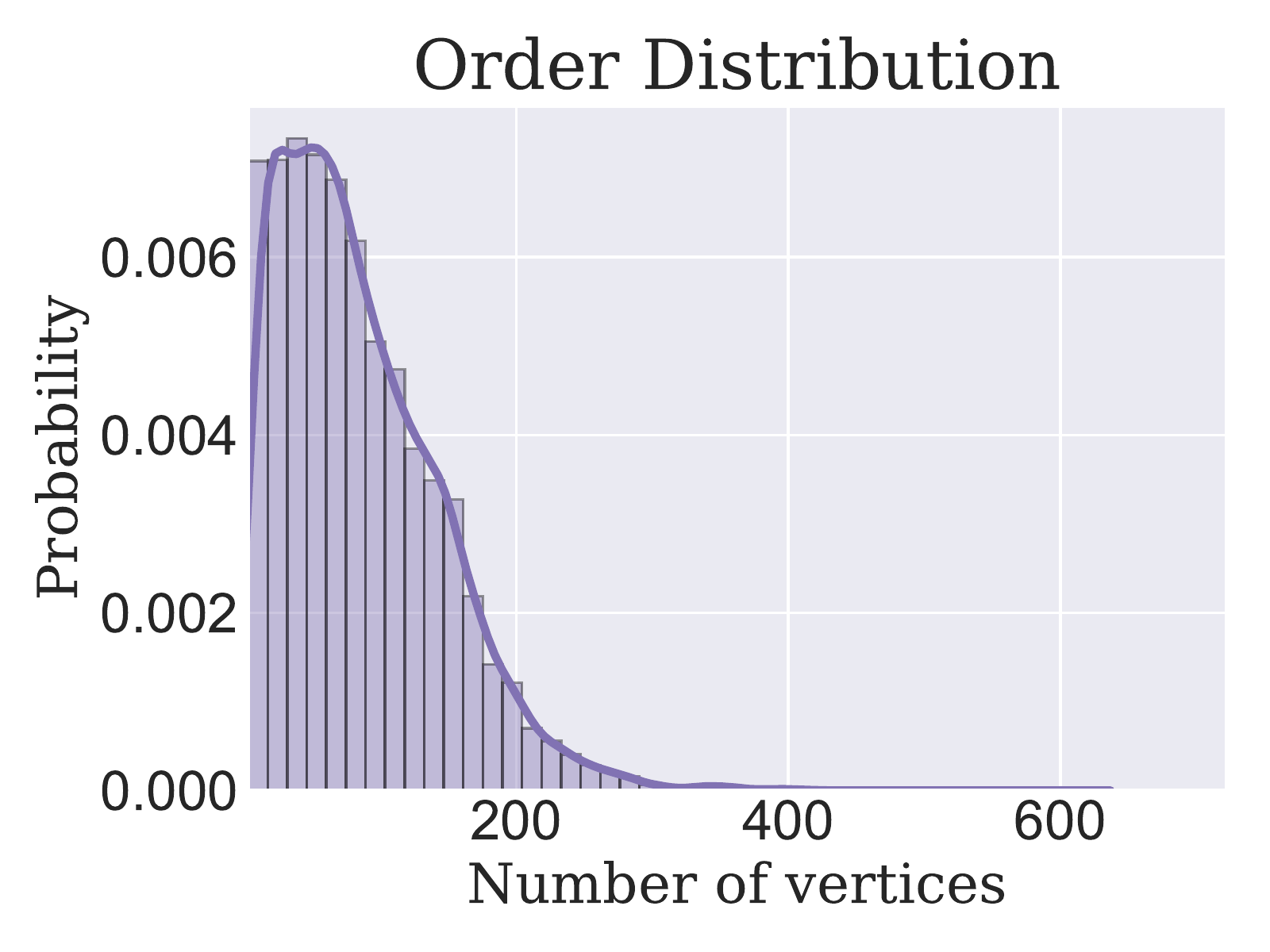} \\
			\end{tabular}\\[1cm]
			\begin{tabular}{|lccc|}    	
				\hline
				\multicolumn{1}{c}{} & \multicolumn{1}{c}{\textbf{Max.}} &
				\textbf{Min.} & \textbf{Average} \\
				\hline
				\textbf{Species}              & 100      & 3    & 44.5  \\
				\textbf{Extant genes}         & 722      & 3    & 84.62 \\
				\textbf{Speciations ($S_n$)}  & 667      & 2    & 73.46  \\   
				\textbf{Duplications ($D_n$)} & 88       & 0    & 14.66 \\
				\textbf{Losses ($L_n$)}       & 55       & 0    & 7.41  \\
				\textbf{HGT ($H_n\%$)}        & 100\%    & 0\%  & 11.76\% \\
				\hline
			\end{tabular}
		\end{tabular}
		\caption{Summary statistics of the 14,000 simulated scenarios. (a)--(c)
			Distributions of fraction of duplications, losses, and HGTs,
			respectively in the true gene trees $\tilde T$.  (d) Distribution of
			the number of extant genes in the observable gene tree $T$ and thus the
			number of vertices (order) of the best match graph $G(T,\sigma)$. The
			spline in each panel is a kernel density estimate.}
	\end{figure}

	\begin{table}[h]
		\caption{We simulated 11 batches with different ranges for the rates of
			duplications, losses, and HGT (columns 3 to 5), where the rates have been
			varied in steps of 0.01. In
			each batch, we simulated for each combination of rates exactly one
			scenario. The second column shows the total number of scenarios for each
			batch.}
		\begin{tabular}{cccccc}
			Batch & \# Scenarios
			& Duplication rates & Loss rates  & HGT rates      & \# Species\\
			1     & 1000	 & 0.75 - 0.84       & 0.7 - 0.79   & 0.1 - 0.19  & 3-100 \\
			2     & 1000	 & 0.85 - 0.94       & 0.85 - 0.94  & 0.1 - 0.19  & 3-100 \\
			3     & 1000	 & 0.80 - 0.89       & 0.80 - 0.89  & 0.1 - 0.19  & 3-100 \\
			4     & 1000	 & 0.70 - 0.79       & 0.70 - 0.79  & 0.1 - 0.19  & 3-100 \\
			5     & 1000	 & 0.90 - 0.99       & 0.90 - 0.99  & 0.1 - 0.19  & 3-100 \\
			6     & 1000	 & 0.85 - 0.94       & 0.75 - 0.84  & 0.1 - 0.19  & 3-100 \\
			7     & 1000	 & 0.90 - 0.99       & 0.90 - 0.99  & 0.15 - 0.24 & 3-100 \\
			8     & 1000	 & 0.90 - 0.99       & 0.90 - 0.99  & 0.15 - 0.24 & 3-100 \\
			9     & 1000	 & 0.65 - 0.74       & 0.65 - 0.74  & 0.10 - 0.19 & 3-100 \\
			10    & 1000	 & 0.85 - 0.94       & 0.75 - 0.84  & 0.15 -0.24  & 3-100 \\
			11    & 4000     & 0.75 - 0.94       & 0.75 - 0.94  & 0.15 -0.24  & 3-50  \\
		\end{tabular}
		\label{table:scenarios}
	\end{table}
	
	\clearpage
	
	\section{False Positive Edges in Non-Cograph 3-RBMGs} 
	
	In the following, we identify further false positive orthology assignments
	in the RBMG based on results that we recently derived in \citet{Geiss:19x}.
	We start by defining a color-preserving thinness relation that has been
	introduced in \citet{Geiss:19x}:
	\begin{definition} 
		For an undirected colored graph $(G,\sigma)$ two vertices $a$ and $b$ are
		in relation $\sthin$, in symbols $a\sthin b$, if $N(a) = N(b)$ and
		$\sigma(a) = \sigma(b)$. The equivalence class of $a$ is denoted by
		$[a]$.  $(G,\sigma)$ is called $\sthin$-thin if no two distinct vertices
		are in relation $\sthin$.
	\end{definition}
	
	\subsection{Type (B) 3-RBMGs} 
	
	Let $(G,\sigma)$ be a connected $\sthin$-thin 3-RBMG of Type
	\AX{(B)}. Lemma 25 of \cite{Geiss:19x} then implies that $(G,\sigma)$
	contains an induced path
	$P:=\langle \hat x_1 \hat y \hat z \hat x_2 \rangle$ with three distinct
	colors $\sigma(\hat x_1)=\sigma(\hat x_2)=:r$, $\sigma(\hat y)=:s$, and
	$\sigma(\hat z)=:t$, and $N_r(\hat y)\cap N_r(\hat z)=\emptyset$ such that
	the vertex sets
	\begin{itemize}
		\item[]$L_{t,s}^{P} \coloneqq \{y \mid  \langle xy\hat z\rangle \in
		\mathscr{P}_3 \text{ for any } x\in N_{r}(y)\}$,
		\item[]$L_{t,r}^{P} \coloneqq \{x\mid N_{r}(y)=\{x\} \text{ and }
		\langle xy\hat z\rangle \in \mathscr{P}_3\}
		\cup 
		\{x\mid x\in L[r],\, N_{s}(x)=\emptyset,\,
		L[s]\setminus L_{t,s}^P\neq\emptyset\}$,
		\item[]$L_{s,t}^{P} \coloneqq \{z \mid \langle xz\hat y\rangle \in \mathscr{P}_3
		\text{ for any } x\in N_{r}(z)\}$, 
		\item[]$L_{s,r}^{P} \coloneqq  \{x\mid N_{r}(z)=\{x\} \text{ and }
		xz\hat y \in \mathscr{P}_3\} \cup 
		\{x\mid x\in L[r], N_{t}(x)=\emptyset,
		L[t]\setminus L_{s,t}^P\neq\emptyset\}$,
		\item[]$L_{t}^{P} \coloneqq  L_{t,s}^{P} \cup L_{t,r}^{P}$,
		\item[]$L_{s}^{P} \coloneqq  L_{s,t}^{P} \cup L_{s,r}^{P}$, and 
		\item[]$L_*^P     \coloneqq  L \setminus (L_t^P \cup L_s^P)$ 
	\end{itemize}
	satisfy the following conditions:
	\begin{description}
		\item[\AX{(B2.a)}] If $x\in L_*^P[r]$, then $N(x)=L_*^P\setminus\{x\}$,
		\item[\AX{(B2.b)}] If $x\in L_t^P[r]$, then $N_s(x)\subset L_t^P$ and
		$|N_s(x)|\le 1$, and $N_t(x)= L_*^P[t]$,
		\item[\AX{(B2.c)}] If $x\in L_s^P[r]$, then
		$N_t(x)\subset L_s^P$ and $|N_t(x)|\le 1$, and $N_s(x)= L_*^P[s]$
		\item[\AX{(B3.a)}] If $y\in L_*^P[s]$, then
		$N(y)=L_s^P\cup (L_*^P\setminus\{y\})$, 
		\item[\AX{(B3.b)}] If $y\in L_t^P[s]$, then
		$N_r(y)\subset L_t^P$ and $|N_r(y)|\le 1$, and $N_t(y)=L[t]$,
		\item[\AX{(B4.a)}] If $z\in L_*^P[t]$, then
		$N(z)=L_t^P\cup (L_*^P\setminus\{z\})$, 
		\item[\AX{(B4.b)}] If $z\in L_s^P[t]$, then
		$N_r(z)\subset L_s^P$ and $|N_r(z)|\le 1$, and $N_s(z)=L[s]$.
	\end{description} 
	By construction, $\sigma(L_t^P)=\{r,s\}$ and $\sigma(L_s^P)=\{r,t\}$ and,
	as a consequence of Lemma 25 of \citet{Geiss:19x}, the sets $L_s^P$,
	$L_t^P$, and $L_*^P$ form a partition of $V(G)$. Furthermore, Lemma 33 of
	\citet{Geiss:19x} implies that any 3-colored induced path $P$ of the form
	$(r,s,t,r)$ that satisfies \AX{(B2.a)} to \AX{(B4.b)} is a good quartet
	w.r.t.\ some $(T,\sigma)$ explaining a BMG $(\G,\sigma)$ that contains
	$(G,\sigma)$ as its symmetric part.
	
	Our goal is to identify edges in $(G,\sigma)$ that can cannot be present in
	the orthology graph $\Theta$.  To this end we extend the leaf sets
	$L^P_*,L^P_s,L^P_t$ that have been introduced in \citet{Geiss:19x} for
	$\sthin$-thin 3-RBMGs, to general 3-RBMGs:
	\begin{definition}
		Let $(G,\sigma)$ be a 3-RBMG of Type \AX{(B)} with vertex set $L$ and
		colors $S=\{r,s,t\}$, and $(G/\sthin, \sigmasthin)$ with vertex set
		$\overline{L}$ be its $\sthin$-thin version. We set
		\begin{align*}
		L_s^P & \coloneqq \{x\mid x\in L, [x]\in \overline{L}_s^P\}\\
		L_t^P & \coloneqq \{x\mid x\in L, [x]\in \overline{L}_t^P\}\\
		L_*^P & \coloneqq \{x\mid x\in L, [x]\in \overline{L}_*^P\}
		\end{align*}
		if $(G,\sigma)$ is of Type \AX{(B)} and $(G/\sthin, \sigmasthin)$ B-like
		w.r.t.\ to the induced path $P$.
	\end{definition}
	
	The cases of Type \AX{(B)} and \AX{(C)} 3-RBMGs will be treated separately,
	starting with Type \AX{(B)}. We first need a technical result:
	\begin{lemma}\label{lem:tree-constraints}
		Let $(G,\sigma)$ be a connected 3-RBMG of Type \AX{(B)} with vertex set
		$L$, $(G/\sthin, \sigmasthin)$ its $\sthin$-thin version with vertex set
		$\overline{L}$, and $(T,\sigma)$ a leaf-labeled tree that explains
		$(G,\sigma)$. Moreover, let
		$P\coloneqq\langle [\tilde x_1] [\tilde y] [\tilde z] [\tilde
		x_2]\rangle$ for some good quartet
		$\langle \tilde x_1\tilde y\tilde z\tilde x_2\rangle$ in $\G(T,\sigma)$,
		and set $v\coloneqq \lca_T(\tilde x_1, \tilde x_2, \tilde y,\tilde
		z)$. Then the leaf sets $L_s^P$, $L_t^P$, and $L_*^P$, where
		$\sigma(\tilde x_1)=\sigma(\tilde x_2)=r$, $\sigma(\tilde y)=s$, and
		$\sigma(\tilde z)=t$, satisfy:
		\begin{itemize}
			\item[(i)] $L_t^P, L_s^P \subseteq L(T(v))$,
			\item[(ii)] If $L_c^P\cap L(T(v'))\neq \emptyset$ for some
			$v'\in\child(v)$ and $c\in \{s,t\}$, then
			\begin{itemize}
				\item[(a)] $L_{\overline c}^P\cap L(T(v'))=\emptyset$, where
				$\overline c \in \{s,t\}, \overline c\neq c$,
				\item[(b)] $\sigma(L(T(v')))\subseteq \sigma(L_c^P)$,
			\end{itemize}
			\item[(iii)] $\lca_T(a,b)=v$ for any $a\in L_*^P$, $b\notin L_*^P$ with
			$ab\in E(G)$.
		\end{itemize}
	\end{lemma}
	\begin{proof} 
		Throughout this proof we will often use the fact that $xy\in E(G)$ if and
		only if $[x][y]\in E(G/\sthin)$ for any $x,y\in L$ (cf.\ Lemma 5 of
		\citet{Geiss:19x}).
		
		Lemma 25 of \citet{Geiss:19x} implies
		$[\tilde x_1], [\tilde y] \in \overline{L}_t^P$ and
		$[\tilde x_2], [\tilde z]\in \overline{L}_s^P$, thus, by definition, we
		have $\tilde x_1, \tilde y \in L_t^P$ and
		$\tilde x_2, \tilde z\in L_s^P$.  Moreover, by Lemma 36 of
		\citet{Geiss:19x}, there exist distinct children $v_1,v_2\in\child(v)$
		such that $\tilde x_1, \tilde y \preceq_T v_1$ and
		$\tilde x_2, \tilde z\preceq_T v_2$.  Therefore
		$\tilde y \tilde z \in E(G)$ implies $\sigma(L(T(v_1)))=\{r,s\}$;
		otherwise there exists a leaf $z'\in L(T(v_1))\cap L[t]$ which implies
		$\lca_T(\tilde y, z')\prec_T v= \lca_T(\tilde y,\tilde z)$; a
		contradiction to $\tilde y \tilde z \in E(G)$. Analogously we obtain
		$\sigma(L(T(v_2)))=\{r,t\}$.
		
		\smallskip\par\noindent (i) By symmetry, it suffices to consider $L_t^P$
		in more detail, analogous arguments can then be applied to $L_s^P$.  Let
		$a\in L_t^P$, or equivalently $[a]\in \overline{L}_t^P$ by definition,
		and suppose first $\sigma(a)=s$. Then Property \AX{(B3.b)} implies
		$[a][\tilde z]\in E(G/\sthin)$. As a consequence of Lemma 5 of
		\citet{Geiss:19x} we thus have $a\tilde z\in E(G)$. Hence,
		$\tilde y \tilde z\in E(G)$ implies
		$\lca_T(a, \tilde z)=\lca_T(\tilde y,\tilde z)=v$ and thus,
		$a\preceq_T v$. We therefore conclude $L_t^P\cap L[s] \subseteq L(T(v))$.
		Now assume $\sigma(a)=r$. By Property \AX{(B2.b)}, we either have
		$N_s([a])=\emptyset$ or there exists a vertex $y\in L[s]$ such that
		$[y]\in \overline{L}_t^P$ and $N_s([a])=\{[y]\}$. In the latter case,
		since $[y]\in \overline{L}_t^P$ implies $y\in L_t^P$ and, in addition, it
		holds $L_t^P\cap L[s] \subseteq L(T(v))$, we have $y\preceq_T
		v$. Moreover, by \AX{(B3.b)}, it holds
		$[\tilde x_2][y]\notin E(G/\sthin)$, hence $\tilde x_2 y\notin E(G)$. As
		a consequence of the latter and the fact that $[a][y]\in E(G/\sthin)$
		implies $ay\in E(G)$, it must hold
		$\lca_T(a,y)\prec_T \lca_T(\tilde x_2,y)\preceq_T v$ and thus,
		$a\preceq_T v$. Otherwise, if $N_s([a])=\emptyset$, then there must exist
		a leaf $z\in L[t]$ such that $[z]\in N_t([a])$ due to the connectedness
		of $G/\sthin$, which is implied by the connectedness of $G$ (cf.\ Lemma 5
		of \citet{Geiss:19x}). Since $[a]\in \overline{L}_t^P$, Properties
		\AX{(B4.a)} and \AX{(B4.b)} immediately imply $[z]\in
		\overline{L}_*^P$. Then, by \AX{(B4.a)}, the edge $[\tilde x_1] [z]$ must
		be contained in $G/\sthin$, thus $\tilde x_1 z\in E(G)$. Since
		$\tilde x_1, \tilde z\preceq_T v$ by Lemma 36 of \citet{Geiss:19x}, it
		must thus hold
		$\lca_T(\tilde x_1, z)\preceq_T \lca_T(\tilde x_1, \tilde z)\preceq_T
		v$. Therefore $\tilde x_1 z, az\in E(G)$ implies
		$\lca_T(a, z)=\lca_T(\tilde x_1, z)\preceq_T v$ and thus, $a\preceq_T v$.
		Hence, $L_t^P\cap L[r] \subseteq L(T(v))$, which finally implies
		$L_t^P \subseteq L(T(v))$.
		
		\smallskip\par\noindent (ii) By symmetry, it again suffices to consider
		the case $c=t$. Let $a\in L_t^P\cap L(T(v'))$ for some
		$v'\in\child(v)$. Note that, by (i), such a leaf $a$ and inner vertex
		$v'$ must exist. We need to distinguish the two Cases (1) $\sigma(a)=s$
		and (2) $\sigma(a)=r$.
		
		Consider first Case (1), thus in particular $s\in\sigma(L(T(v')))$. Then,
		as $\sigma(L(T(v_2)))=\{r,t\}$, we have $v'\neq v_2$ and thus,
		$\lca_T(a,\tilde z)=v$. Hence, as $[a][\tilde z]\in E(G/\sthin)$ by
		Property \AX{(B3.b)} and therefore, $a\tilde z\in E(G)$, we can conclude
		$t\notin\sigma(L(T(v')))$ by analogous arguments as just used for showing
		$\sigma(L(T(v_1)))=\{r,s\}$. This implies (ii.b).  Now assume, for
		contradiction, that there exists a leaf $x\in L(T(v'))\cap L_s^P$. Since
		$t\notin\sigma(L(T(v')))$ and, by definition, $s\notin \sigma(L_s^P)$,
		this leaf $x$ must be of color $r$.  Clearly, either there exists a leaf
		$y\in L[s]$ such that $xy\in E(G)$ or $N_s(x)=\emptyset$. In the first
		case, we have $[x][y]\in E(G/\sthin)$ and thus, by \AX{(B2.c)},
		$[y]\in \overline{L}_*^P$ which implies $y\in L_*^P$. In particular, as
		$s\in \sigma(L(T(v')))$ and $xy\in E(G)$ implies
		$\lca_T(x,y)\preceq_T \lca_T(x,y')$ for any $y'\in L[s]$, we can conclude
		$y\preceq_T v'$. Moreover, since $[\tilde x_2]\in \overline{L}_s^P$,
		Property \AX{(B3.a)} implies $[\tilde x_2] [y]\in E(G/\sthin)$ and thus,
		$\tilde x_2 y\in E(G)$. However, since $v'\neq v_2$, we have
		$\lca_T(x,y)\preceq_T v'\prec_T v=\lca_T(\tilde x_2,y)$; a contradiction
		to $\tilde x_2 y \in E(G)$. We thus conclude $N_s(x)=\emptyset$. Hence,
		as $G$ is connected, there must exist a leaf $z'$ of color $t$ such that
		$xz'\in E(G)$, which implies $[x][z']\in E(G/\sthin)$. By Property
		\AX{(B2.c)}, we have $[z']\in \overline{L}_s^P$ and therefore,
		\AX{(B4.b)} implies $N_r([z'])=\{[x]\}$. Since
		$t\notin \sigma(L(T(v')))$, there is a $v''\in\child(v)\setminus \{v'\}$
		such that $z'\preceq_T v'' \prec_T v$. From $xz'\in E(G)$ and
		$\lca_T(x,z')=v$, we conclude that $r\notin \sigma(L(T(v'')))$.
		Moreover, Lemma 10 of \citet{Geiss:19x} implies that there exist leaves
		$x',y'\in L(T(v'))$ with $\sigma(x')=r$ and $\sigma(y')=s$ such that
		$x'y'\in E(G)$. Thus, as by assumption $N_s(x)=\emptyset$, we in
		particular have $[x]\neq [x']$. Since $r\notin \sigma(L(T(v'')))$ and
		$t\notin \sigma(L(T(v')))$, it follows $x'z'\in E(G)$ and therefore,
		$[x']\in N_r([z'])$; a contradiction to $N_r([z'])=\{[x]\}$. This implies
		(ii.a).
		
		Now consider Case (2), i.e., $\sigma(a)=r$. We first show that
		$\sigma(L(T(v')))\subsetneq\{r,s,t\}$ holds. Assume, for contradiction,
		that $L(T(v'))$ contains leaves $y\in L[s]$ and $z\in L[t]$. If
		$v'\neq v_2$, this implies $\lca_T(y,z)\prec_T v=\lca(y,\tilde z)$ and
		thus, $y\tilde z \notin E(G)$ and in particular
		$[y][\tilde z] \notin E(G/\sthin)$; a contradiction to \AX{(B4.b)}. One
		analogously obtains a contradiction for the case $v'\neq v_1$; therefore
		$\sigma(L(T(v')))\subsetneq\{r,s,t\}$ and we either have
		$\sigma(L(T(v')))\subseteq\{r,s\}$ or
		$\sigma(L(T(v')))\subseteq\{r,t\}$. If $\sigma(L(T(v')))=\{r\}$, then it
		clearly holds $N(x)=N(a)$ and thus $x\in L_t^P$ for any $x\in L(T(v'))$,
		hence (ii.a) and (ii.b) are trivially satisfied. If
		$\sigma(L(T(v')))=\{r,s\}$, then (ii.b) is trivially satisfied. Moreover,
		by Lemma 10 of \citet{Geiss:19x}, $L(T(v'))$ contains leaves $x'\in L[r]$
		and $y'\in L[s]$ such that $x'y'\in E(G)$. Hence, we have
		$[x'][y']\in E(G/\sthin)$ and Property \AX{(B4.b)} implies
		$[y'][\tilde z]\in E(G/\sthin)$ and thus, $y'\tilde z\in E(G)$. As
		$\sigma(L(T(v_2)))=\{r,t\}$ and $\sigma(L(T(v')))=\{r,s\}$, we clearly
		have $v'\neq v_2$ and thus,
		$\lca_T(x',y')\preceq_T v'\prec_T v=\lca_T(\tilde x_2,y')$. Hence,
		$\tilde x_2 y'\notin E(G)$, which implies
		$N([y'])\neq \overline{L}_s^P\cup (\overline{L}_*^P\setminus \{[y']\})$
		since $\tilde x_2\in L_s^P$. Therefore, by Property \AX{(B3.a)}, we have
		$[y']\notin \overline{L}_*^P$, implying $y'\notin L_*^P$. We thus
		conclude $y'\in L_t^P$. Hence, we can apply the argumentation of
		Case (1) (by substituting $a=y'$) in order to infer (ii.a).\\
		Finally, for contradiction, assume $\sigma(L(T(v')))=\{r,t\}$. In
		particular, this implies $v_1\neq v'$. Clearly, either there exists a
		leaf $y\in L[s]$ such that $ay\in E(G)$ (and thus
		$[a][y]\in E(G/\sthin)$) or $N_s(a)=\emptyset$. In the latter case, since
		$G$ is connected, there must be a leaf $z\in L[t]$ such that $az\in E(G)$
		and $[a][z]\in E(G/\sthin)$. In particular, as
		$\sigma(L(T(v')))=\{r,t\}$, this implies $z\preceq_T v'$. By \AX{(B2.b)},
		we have $[z]\in \overline{L}_*^P$ and thus, by \AX{(B4.a)}, it follows
		$[\tilde x_1] [z]\in E(G/\sthin)$ implying $\tilde x_1 z\in E(G)$; a
		contradiction since
		$\lca_T(z,a)\preceq_T v' \prec_T v=\lca_T(z,\tilde x_1)$. Hence, there
		must exist a leaf $y\in L[s]$ such that $ay\in E(G)$. By \AX{(B2.b)}, we
		have $N_s([a])=\{[y]\}$ and $[y]\in \overline{L}_t^P$. Then \AX{(B3.b)}
		implies $N_r([y])\subset \overline{L}_t^P$. It is easy to see that this
		implies $N_r(y)\subset L_t^P$. Since $s\notin \sigma(L(T(v')))$, there
		must exist a vertex $v''\in\child(v)\setminus \{v'\}$ such that
		$y\preceq_T v''\prec_T v=\lca_T(a,y)$. One easily checks that
		$ay\in E(G)$ implies $r\notin \sigma(L(T(v'')))$. Together with
		$\sigma(L(T(v_2)))=\{r,t\}$, this implies
		$\lca_T(\tilde x_2,y)=v\preceq_T \lca_T(x'',y)$ and
		$\lca_T(\tilde x_2,y)=v\preceq_T \lca_T(\tilde x_2,y')$ for any
		$x''\in L[r]$ and $y'\in L[s]$. Thus, $\tilde x_2 y\in E(G)$, which, as
		$\tilde x_2\in L_s^P$, contradicts $N_r(y) \subset L_t^P$. We therefore
		conclude that $\sigma(L(T(v')))=\{r,t\}$ is not possible, which finally
		completes the proof.
		
		\smallskip\par\noindent (iii) Since, by definition, $V(G)$ is partitioned
		into $L_s^P$, $L_t^P$, and $L_*^P$, the leaf $b$ must be either contained
		in $L_t^P$ or $L_s^P$. Suppose first $b\in L_t^P$. Since
		$[a][b]\in E(G/\sthin)$ follows from $ab\in E(G)$, Properties
		\AX{(B2.a)}, \AX{(B3.a)}, and \AX{(B4.a)} immediately imply
		$\sigma(a)=t$. Moreover, by (i), there exists some $v'\in \child(v)$ such
		that $b\preceq_T v'\prec_T v$, and, by (ii.b),
		$\sigma(L(T(v')))\subseteq \sigma(L_t^P)=\{r,s\}$. Hence, as
		$\sigma(a)=t$, we can conclude $\lca_T(a,b)\succeq_T v$. Similarly,
		$\sigma(L(T(v')))\subseteq \{r,s\}$ implies $\lca_T(b,\tilde z)=v$, thus
		it must hold $\lca_T(a,b)\preceq_T\lca_T(b,\tilde z)=v$ because of
		$ab\in E(G)$. In summary, this implies $\lca_T(a,b)=v$. Analogous
		arguments can be applied to the case $b\in L_s^P$.
	\end{proof}

	\noindent Lemma~\ref{lem:tree-constraints} can now be used to identify a
	potentially very large set of edges that cannot be present in the orthology
	graph $\Theta$.
	\begin{theorem}
		Let $T$ and $S$ be planted trees, $\sigma: L(T)\to L(S)$ a surjective
		map, and $\mu$ a reconciliation map from $(T,\sigma)$ to $S$ determining
		an event labeling $t_T$ on $T$. Moreover, let the leaf sets $L_t^P$,
		$L_s^P$, and $L_*^P$ be defined w.r.t.\ $P$, which is the $\sthin$-thin
		version of some good quartet of the form $(r,s,t,r)$ in $(\G,\sigma)$
		with color set $S=\{r,s,t\}$. Then $t_T(\lca_T(a,b))=\DUPL$ for any edge
		$ab\in E(G)$ such that $a\in L_{\star}^P$ and $b\notin L_{\star}^P$,
		where $\star \in \{s,t,*\}$.
		\label{cor:B-erase}
	\end{theorem}
	\begin{proof}
		Let $P=\langle [x_1] [y] [z] [x_2]\rangle$, i.e., in particular
		$\sigma(x_1)=\sigma(x_2)=r$, $\sigma(y)=s$, and $\sigma(z)=t$, and let
		$v\coloneqq \lca_T(x_1,x_2,y,z)$. Then, by Lemma 36 of \citet{Geiss:19x},
		there exist distinct $v_1,v_2\in\child(v)$ such that $x_1,y\preceq_T v_1$
		and $x_2,z\preceq_T v_2$. As $[x_1], [y]\in \overline{L}_t^P$ and
		$[x_2],[z]\in \overline{L}_s^P$ by Lemma 25 of \citet{Geiss:19x} and
		thus, by definition, $x_1, y\in L_t^P$ and $x_2,z\in L_s^P$,
		Lemma~\ref{lem:tree-constraints}(ii.b) in particular implies
		$\sigma(L(T(v_1)))=\{r,s\}$ and $\sigma(L(T(v_2)))=\{r,t\}$.
		\\
		Now, if $a\in L_t^P$, $b\in L_s^P$, it follows from Lemma
		\ref{lem:tree-constraints}(ii.a) that $\lca_T(a,b)=v$. On the other hand,
		if $a\in L_*^P$ and either $b\in L_s^P$ or $b\in L_t^P$, then we also
		have $\lca_T(a,b)=v$ by Lemma~\ref{lem:tree-constraints}(iii). Since
		$\sigma(L(T(v_1)))\cap \sigma(L(T(v_2)))=\{r\} \neq \emptyset$, we
		conclude from Lemma~\ref{lem:disjointSpecies} that $\mu(v)\notin V^0(S)$,
		which implies $t_T(v)\neq \SPEC$. Therefore we have $t_T(v)=\DUPL$.
	\end{proof}
	
	\subsection{Type \AX{(C)} 3-RBMGs}
	
	Let $(G,\sigma)$ be a connected $\sthin$-thin 3-RBMG of Type
	\AX{(C)}. Lemma 27 of \cite{Geiss:19x} then implies that $(G,\sigma)$
	contains an induced hexagon
	$H:=\langle \hat x_1 \hat y_1 \hat z_1 \hat x_2 \hat y_2 \hat z_2\rangle$
	with three distinct colors $\sigma(\hat x_1)=\sigma(\hat x_2)=:r$,
	$\sigma(\hat y_1)=\sigma(\hat y_2)=:s$, and
	$\sigma(\hat z_1)=\sigma(\hat z_2)=:t$, and $|N_t(\hat x_1)|>1$ such that
	the vertex sets
	\begin{itemize}
		\item[] $ L_t^H \coloneqq \{x \mid \langle x \hat z_2 \hat y_2 \rangle \in
		\mathscr{P}_3\}\cup \{y \mid \langle y \hat z_1 \hat x_2 \rangle\in
		\mathscr{P}_3\}$,
		\item[] $ L_s^H \coloneqq \{x \mid \langle x \hat y_2 \hat z_2 \rangle \in
		\mathscr{P}_3\} \cup \{z \mid \langle z \hat y_1 \hat x_1 \in \rangle
		\mathscr{P}_3\}$,
		\item[] $ L_r^H \coloneqq \{ y \mid \langle y \hat x_2 \hat z_1 \rangle \in
		\mathscr{P}_3\} \cup \{z \mid \langle z \hat x_1 \hat y_1 \rangle\in
		\mathscr{P}_3\}$, and
		\item[] $L_*^H \coloneqq V(G)\setminus (L_r^H\cup L_s^H\cup L_t^H)$
	\end{itemize}
	satisfy the following conditions:
	\begin{description}
		\item[\AX{(C2.a)}] If $x\in L_*^H[r]$, then
		$N(x)=L_r^H\cup (L_*^H\setminus\{x\})$,
		\item[\AX{(C2.b)}] If $x\in L_t^H[r]$, then
		$N_s(x)\subset L_t^H$ and $|N_s(x)|\le 1$, and
		$N_t(x)= L_*^H[t]\cup L_r^H[t]$,
		\item[\AX{(C2.c)}] If $x\in L_s^H[r]$, then
		$N_t(x)\subset L_s^H$ and $|N_t(x)|\le 1$, and
		$N_s(x)= L_*^H[s]\cup L_r^H[s]$
		\item[\AX{(C3.a)}] If $y\in L_*^H[s]$, then
		$N(y)=L_s^H\cup (L_*^H\setminus\{y\})$,
		\item[\AX{(C3.b)}] If $y\in L_t^H[s]$, then
		$N_r(y)\subset L_t^H$ and $|N_r(y)|\le 1$, and
		$N_t(y)=L_*^H[t]\cup L_s^H[t]$,
		\item[\AX{(C3.c)}] If $y\in L_r^H[s]$, then
		$N_t(y)\subset L_r^H$ and $|N_t(y)|\le 1$, and
		$N_r(y)=L_*^H[r]\cup L_s^H[r]$,
		\item[\AX{(C4.a)}] If $z\in L_*^H[t]$, then
		$N(z)=L_t^H\cup (L_*^H\setminus\{z\})$,
		\item[\AX{(C4.b)}] If $z\in L_s^H[t]$, then
		$N_r(z)\subset L_s^H$ and $|N_r(z)|\le 1$, and
		$N_s(z)=L_*^H[s]\cup L_t^H[s]$,
		\item[\AX{(C4.c)}] If $z\in L_r^H[t]$, then
		$N_s(z)\subset L_r^H$ and $|N_s(z)|\le 1$, and
		$N_r(z)=L_*^H[r]\cup L_t^H[r]$.
	\end{description} 
	By construction, $\sigma(L_t^H)=\{r,s\}$, $\sigma(L_s^H)=\{r,t\}$, and
	$\sigma(L_r^H)=\{s,t\}$ and, as a consequence of Lemma 27 of
	\citet{Geiss:19x}, the sets $L_r^H$, $L_s^H$, $L_t^H$, and $L_*^H$ form a
	partition of $V(G)$.  Similarly to the Type \AX{(B)} case, we extend the
	leaf sets $L^H_*,L^H_r,L^H_s,L^H_t$ that have been introduced in
	\citet{Geiss:19x} for $\sthin$-thin 3-RBMGs of Type \AX{(C)}, to general
	Type \AX{(C)} 3-RBMGs:
	
	\begin{definition}
		Let $(G,\sigma)$ be a 3-RBMG of Type \AX{(C)} with vertex set $L$ and
		colors $S=\{r,s,t\}$, and $(G/\sthin, \sigmasthin)$ with vertex set
		$\overline{L}$ be its $\sthin$-thin version. We set
		\begin{align*}
		L_r^H & \coloneqq \{x\mid x\in L, [x]\in \overline{L}_r^H\}\\
		L_s^H & \coloneqq \{x\mid x\in L, [x]\in \overline{L}_s^H\}\\
		L_t^H & \coloneqq \{x\mid x\in L, [x]\in \overline{L}_t^H\}\\
		L_*^H & \coloneqq \{x\mid x\in L, [x]\in \overline{L}_*^H\}
		\end{align*}
		if $(G,\sigma)$ is of Type \AX{(C)} and $(G/\sthin, \sigmasthin)$ C-like
		w.r.t.\ to the hexagon $H$.
	\end{definition}
	
	Again, we can identify edges in $(G,\sigma)$ that are necessarily are false
	positives in the orthology graph $\Theta$.  A similar procedure as in the
	Type \AX{(B)} case will be applied to Type \AX{(C)} 3-RBMGs, again starting
	with an analogous technical result:
	
	\begin{lemma}\label{lem:tree-constraints-C}
		Let $(G,\sigma)$ be a connected 3-RBMG of Type \AX{(C)} with vertex set
		$L$, $(G/\sthin,\sigmasthin)$ its $\sthin$-thin version with vertex set
		$\overline{L}$, and $(T,\sigma)$ a leaf-labeled tree that explains
		$(G,\sigma)$. Moreover, let
		$H\coloneqq\langle [\tilde x_1] [\tilde y_1] [\tilde z_1] [\tilde x_2]
		[\tilde y_2] [\tilde z_2]\rangle$ for some induced hexagon
		$\langle \tilde x_1 \tilde y_1 \tilde z_1 \tilde x_2 \tilde y_2 \tilde
		z_2\rangle$ in $\G(T,\sigma)$ with $|N_t([\tilde x_1])|>1$ and
		$\sigma(\tilde x_1)=\sigma(\tilde x_2)=r$,
		$\sigma(\tilde y_1)=\sigma(\tilde y_2)=s$, and
		$\sigma(\tilde z_1)=\sigma(\tilde z_2)=t$, and set
		$v\coloneqq \lca_T(\tilde x_1, \tilde x_2, \tilde y_1, \tilde y_2,\tilde
		z_1, \tilde z_2)$. Then the leaf sets $L_r^H$, $L_s^H$, $L_t^H$, and
		$L_*^H$ satisfy:
		\begin{itemize}
			\item[(i)] $L_r^H, L_s^H, L_t^H \subseteq L(T(v))$,
			\item[(ii)] If $L_c^H\cap L(T(v'))\neq \emptyset$ for some
			$v'\in\child(v)$ and $c\in \{r,s,t\}$, then
			\begin{itemize}
				\item[(a)] $L_{\overline c}^H\cap L(T(v'))=\emptyset$, where
				$\overline c \in \{r,s,t\}, \overline c\neq c$,
				\item[(b)] $\sigma(L(T(v')))\subseteq \sigma(L_c^H)$,
			\end{itemize}
			\item[(iii)] $\lca_T(a,b)=v$ for any $a\in L_*^H$, $b\notin L_*^H$ with
			$ab\in E(G)$.
		\end{itemize}
	\end{lemma}
	\begin{proof}
		The proof of Lemma~\ref{lem:tree-constraints-C} closely follows the
		arguments leading to Lemma~\ref{lem:tree-constraints}. In particular, we
		again use the fact that $xy\in E(G)$ if and only if
		$[x][y]\in E(G/\sthin)$ for any $x,y\in L$ (cf.\ Lemma 5 of
		\citet{Geiss:19x}).
		
		By Lemma 27 of \citet{Geiss:19x}, we have
		$[\tilde x_1], [\tilde y_1] \in \overline{L}_t^H$,
		$[\tilde x_2], [\tilde z_1]\in \overline{L}_s^H$, and
		$[\tilde y_2], [\tilde z_2]\in \overline{L}_r^H$, hence
		$\tilde x_1, \tilde y_1 \in L_t^H$, $\tilde x_2, \tilde z_1\in L_s^H$,
		and $\tilde y_2, \tilde z_2\in L_r^H$.  Moreover, by Lemma 39(iii) of
		\citet{Geiss:19x}, there exist distinct children
		$v_1,v_2,v_3\in\child(v)$ such that
		$\tilde x_1, \tilde y_1 \preceq_T v_1$,
		$\tilde x_2, \tilde z_2\preceq_T v_2$, and
		$\tilde y_2, \tilde z_2\preceq_T v_3$. In particular, since
		$\tilde y_1 \tilde z_1 \in E(G)$, it must hold
		$\sigma(L(T(v_1)))=\{r,s\}$ as otherwise there exists a leaf
		$z'\in L(T(v_1))\cap L[t]$ which implies
		$\lca_T(\tilde y_1, z')\prec_T v=\lca_T(\tilde y_1,\tilde z_1)$; a
		contradiction to $\tilde y_1 \tilde z_1 \in E(G)$. One analogously checks
		$\sigma(L(T(v_2)))=\{r,t\}$ and $\sigma(L(T(v_3)))=\{s,t\}$.
		
		\smallskip\par\noindent(i) By symmetry, it suffices to consider $L_t^H$
		in more detail, analogous arguments can then be applied to $L_s^H$ and
		$L_r^H$. Let $a\in L_t^H$, or equivalently $[a]\in \overline{L}_t^H$, and
		suppose first $\sigma(a)=r$. Then Property \AX{(C2.b)} implies
		$[a][\tilde z_2]\in E(G/\sthin)$ and thus, $a\tilde z_2\in E(G)$. As
		$\tilde x_1 \tilde z_2\in E(G)$, we thus have
		$\lca_T(a,\tilde z_2)=\lca_T(\tilde x_1,\tilde z_2)=v$, hence
		$a\preceq_T v$. We therefore conclude $L_t^H\cap L[r] \subseteq
		L(T(v))$. Analogously, we obtain $a\preceq_T v$ for $\sigma(a)=s$ as a
		consequence of Property \AX{(C3.b)}. In summary, we obtain
		$L_t^H\subseteq L(T(v))$.
		
		\smallskip\par\noindent(ii) Again invoking symmetry, it suffices to
		consider the case $c=t$. Let $a\in L_t^H\cap L(T(v'))$ for some
		$v'\in\child(v)$. First, let $\sigma(a)=r$. Then, as
		$r\notin\sigma(L(T(v_3)))$, we have $v'\neq v_3$ and thus,
		$\lca_T(a,\tilde z_2)=v$. Hence, as $[a][\tilde z_2]\in E(G/\sthin)$ by
		Property \AX{(C2.b)} and thus $a\tilde z_2\in E(G)$, we can conclude
		$t\notin\sigma(L(T(v')))$ using the same line of reasoning used above for
		showing $\sigma(L(T(v_1)))=\{r,s\}$.  This implies (ii.b).  Now assume,
		for contradiction, that there exists either (1) a leaf
		$x\in L(T(v'))\cap L_s^H$ or (2) a leaf $y\in L(T(v'))\cap L_r^H$.\\
		In Case (1), since $t\notin\sigma(L(T(v')))$ and, by definition,
		$s\notin \sigma(L_s^H)$, this leaf $x$ must be of color $r$. In
		particular, since $L_s^H$ and $L_t^H$ are disjoint, we have $x\neq
		a$. Hence, it must hold $s\in \sigma(L(T(v')))$ as otherwise $N(x)=N(a)$;
		contradicting $a\in L_t^H$, $x\in L_s^H$, and
		$L_s^H\cap L_t^H=\emptyset$. %an der stelle vorher teil unten
		This immediately implies $v'\neq v_2$ because
		$s\notin \sigma(L(T(v_2)))$. By Property \AX{(C2.c)}, as
		$[\tilde y_2]\in \overline{L}_r^H[s]$, we have
		$[x][\tilde y_2]\in E(G/\sthin)$ and thus, $x\tilde y_2\in
		E(G)$. However, since $s\in\sigma(L(T(v')))$, there exists a leaf
		$y'\preceq_T v'$ with $\sigma(y')=s$, which implies
		$\lca_T(x,y')\preceq_T v' \prec_T v=\lca_T(x,\tilde y_2)$ because of
		$\tilde y_2\preceq_T v_3\neq v'$; a contradiction to
		$x\tilde y_2\in E(G)$.
		\par\noindent
		Hence, assume Case (2), i.e., there exists $y\in L(T(v'))\cap
		L_r^H$. Since $t\notin\sigma(L(T(v')))$ and, by definition,
		$r\notin \sigma(L_r^H)$, the leaf $y$ must be of color $s$, which in
		particular implies $v'\neq v_2$. As $t\notin\sigma(L(T(v')))$ and
		$s\notin\sigma(L(T(v_2)))$, one easily checks that $y\tilde z_1\in
		E(G)$. However, as $y\in L_r^H$ and thus $[y]\in \overline{L}_r^H$,
		Property \AX{(C3.c)} implies $[\tilde z_1] \in \overline{L}_r^H$, hence
		$\tilde z_1 \in L_r^H$; a contradiction since $\tilde z_1 \in L_s^H$.
		\par\noindent
		In summary, we conclude that $L_{\overline c}^H\cap L(T(v'))=\emptyset$,
		where $\overline c \in \{r,s\}$, hence (ii.a) is satisfied for $c=t$.
		Analogous arguments can be used to demonstrate that properties (ii.a) and
		(ii.b) are also satisfied for $\sigma(a)=s$.
		
		\smallskip\par\noindent(iii) Since, by definition, $V(G)$ is partitioned
		into $L_r^H$, $L_s^H$, $L_t^H$, and $L_*^H$, the leaf $b$ must be either
		contained in $L_r^H$, $L_s^H$, or $L_t^H$. Suppose first $b\in
		L_t^H$. Then, since $[a][b]\in E(G/\sthin)$ follows from $ab\in E(G)$,
		Properties \AX{(C2.a)}, \AX{(C3.a)}, and \AX{(C4.a)} immediately imply
		$\sigma(a)=t$. Moreover, by (i), there exists some $v'\in \child(v)$ such
		that $b\preceq_T v'\prec_T v$ and, by (ii.b),
		$\sigma(L(T(v')))\subseteq \sigma(L_t^H)=\{r,s\}$. Hence, as
		$\sigma(a)=t$, we can conclude $\lca_T(a,b)\succeq_T v$. Similarly,
		$\sigma(L(T(v')))\subseteq \{r,s\}$ implies $\lca_T(b,\tilde z_1)=v$,
		thus it must hold $\lca_T(a,b)\preceq_T \lca_T(b,\tilde z_1)=v$ because
		of $ab\in E(G)$. In summary, this implies $\lca_T(a,b)=v$. Analogous
		arguments can be applied to the cases $b\in L_s^H$ and $b\in L_r^H$.
	\end{proof}
	
	\noindent Similar to Type (B) 3-RBMGs, we use
	Lemma~\ref{lem:tree-constraints-C} to finally identify false positive
	edges.
	
	\begin{theorem}
		Let $T$ and $S$ be planted trees, $\sigma: L(T)\to L(S)$ a surjective
		map, and $\mu$ a reconciliation map from $(T,\sigma)$ to $S$ determining
		an event labeling $t_T$ on $T$. Moreover, let the leaf sets $L_r^H$,
		$L_s^H$, $L_t^H$, and $L_*^H$ be defined w.r.t.\ $H$, which is the
		$\sthin$-thin version of some hexagon
		$H'=\langle x_1 y_1 z_1 x_2 y_2 z_2\rangle$ of the form $(r,s,t,r,s,t)$
		and $|N_t(x_1)>1|$ in $(\G,\sigma)$ with color set $S=\{r,s,t\}$. Then
		$t_T(\lca_T(a,b))=\DUPL$ for any edge $ab\in E(G)$ such that
		$a\in L_{\star}^H$ and $b\notin L_{\star}^H$, where
		$\star \in \{r,s,t,*\}$.
		\label{cor:C-erase}
	\end{theorem}
	\begin{proof}
		Let $v\coloneqq \lca_T(x_1,x_2,y_1,y_2,z_1,z_2)$. Again, we have
		$[x_1], [y_1]\in \overline{L}_t^H$, $[x_2],[z_1]\in \overline{L}_s^H$,
		and $[y_2],[z_2]\in \overline{L}_r^H$ by Lemma 27 of \citet{Geiss:19x}
		and thus, $x_1, y_1\in L_t^H$, $x_2,z_1\in L_s^H$, $y_2,z_2\in
		L_r^H$. Moreover, by Lemma 39(iii) of \citet{Geiss:19x}, there exist
		distinct $v_1,v_2,v_3\in\child(v)$ such that $x_1,y_1\preceq_T v_1$,
		$x_2,z_1\preceq_T v_2$, and $y_2,z_2\preceq_T v_3$. As
		$x_1, y_1\in L_t^H$, $x_2,z_1\in L_s^H$, $y_2,z_2\in L_r^H$,
		Lemma~\ref{lem:tree-constraints-C}(ii.b) in particular implies
		$\sigma(L(T(v_1)))=\{r,s\}$, $\sigma(L(T(v_2)))=\{r,t\}$, and
		$\sigma(L(T(v_3)))=\{s,t\}$.
		\\
		Now, if $a\in L_c^H$, $b\in L_{\overline c}^H$, where $c=\{r,s,t\}$ and
		$\overline c \in \{r,s,t\}, \overline c\neq c$, it follows from Lemma
		\ref{lem:tree-constraints-C}(ii.a) that $\lca_T(a,b)=v$. On the other
		hand, if $a\in L_*^H$ and $b\in L_c^H$, then we also have $\lca_T(a,b)=v$
		by Lemma~\ref{lem:tree-constraints-C}(iii). Since
		$\sigma(L(T(v_i)))\cap \sigma(L(T(v_j))) \neq \emptyset$ for
		$1\le i < j \le 3$, we conclude from Lemma~\ref{lem:disjointSpecies} that
		$\mu(v)\notin V^0(S)$, which implies $t_T(v)\neq \SPEC$. Therefore we
		have $t_T(v)=\DUPL$.
	\end{proof}
\end{appendix}

\bibliographystyle{plainnat}
\bibliography{cobmg}

\end{document}